\documentclass[12pt]{article}
\usepackage{amsmath}
\usepackage{graphicx}
\usepackage{natbib}
\usepackage{url} 

\newcommand{\blind}{0}

\usepackage{bm}
\usepackage{ragged2e}
\usepackage{booktabs}
\usepackage{multirow}
\usepackage{amsmath} 
\usepackage{comment}
\usepackage{tikz}

\usetikzlibrary{shapes.geometric, arrows}
\tikzstyle{startstop} = [rectangle, rounded corners, 
minimum width=1cm, 
minimum height=1cm,
text centered, 
text width=2cm, 
draw=black, 
fill=red!30]

\tikzstyle{control} = [trapezium, 
trapezium stretches=true, 
trapezium left angle=70, 
trapezium right angle=110, 
minimum width=2cm, 
minimum height=1cm, text centered, 
text width=2cm,
draw=black, fill=blue!30]

\tikzstyle{trisomic} = [trapezium, 
trapezium stretches=true,
trapezium left angle=70, 
trapezium right angle=110, 
minimum width=2cm, 
minimum height=1cm, text centered, 
text width=2cm,
draw=black, fill=blue!30]

\tikzstyle{stimulated} = [rectangle, 
minimum width=1.5cm, 
minimum height=1cm, 
text centered, 
text width=2cm, 
draw=black, 
fill=magenta!50]

\tikzstyle{treatment} = [rectangle, 
minimum width=1cm, 
minimum height=1cm, 
text centered, 
text width=1cm, 
draw=black, 
fill=orange!40]

\tikzstyle{class} = [rectangle, rounded corners,
minimum width=1cm, 
minimum height=1cm, 
text centered, 
draw=black, 
fill=green!100]

\tikzstyle{class2} = [rectangle, rounded corners,
minimum width=1cm, 
minimum height=1cm, 
text centered, 
draw=black, 
fill=orange!100]
\tikzstyle{class3} = [rectangle, rounded corners,
minimum width=1cm, 
minimum height=1cm, 
text centered, 
draw=black, 
fill=cyan!70]

\tikzstyle{class4} = [rectangle, rounded corners,
minimum width=1cm, 
minimum height=1cm, 
text centered, 
draw=black, 
fill=magenta!100]

\tikzstyle{outcome} = [rectangle, rounded corners,  
minimum width=1cm, 
minimum height=1cm, 
text centered, 
draw=black, 
fill=red!70]

\tikzstyle{arrow} = [thick,->,>=stealth]

\newcommand\abs[1]{\left|#1\right|}
\newcommand{\vertiii}[1]{{\left\vert\kern-0.25ex\left\vert\kern-0.25ex\left\vert #1 
    \right\vert\kern-0.25ex\right\vert\kern-0.25ex\right\vert}}
\usepackage{threeparttable}
\usepackage{array}

\usepackage{amsthm}
\newtheorem{assumption}{Assumption}
\newtheorem{theorem}{Theorem}
\newtheorem{lemma}{Lemma}

\newtheorem{definition}{Definition}

\newcommand{\argmin}{arg\,min}

\usepackage{makecell}
\renewcommand{\theadalign}{bl}

\usepackage{xcolor}

\usepackage{lscape}
\usepackage[hidelinks]{hyperref}
\usepackage{rotating}
\usepackage{titlesec}
\usepackage{amssymb}

\begin{document}

\def\spacingset#1{\renewcommand{\baselinestretch}%
{#1}\small\normalsize} \spacingset{1}
\renewcommand{\arraystretch}{0.7}


\if0\blind
{
  \title{\bf \textcolor{black}{High-Dimensional Covariate-Dependent Gaussian Graphical Models}}
  \author{JIACHENG WANG\thanks{ This work was supported by the Natural Sciences and Engineering Research Council of Canada.}
    \hspace{.2cm}
    \\
    Department of Mathematics and Statistics, York University, \\
    4700 Keele Street, Toronto M3J 1P3, Canada. \\
    and \\
    XIN GAO$^{\ast}$ \\
    Department of Mathematics and Statistics, York University, \\
    4700 Keele Street, Toronto M3J 1P3, Canada. \\
    Email: xingao@yorku.ca}
  \maketitle
} \fi

\if1\blind
{
  \bigskip
  \bigskip
  \bigskip
  \begin{center}
    {\LARGE\bf Title}
\end{center}
  \medskip
} \fi

\bigskip
\begin{abstract}
Motivated by dynamic biologic network analysis, we propose a covariate-dependent Gaussian graphical model (cdexGGM) \textcolor{black}{for capturing network structure that varies with covariates through a novel parameterization.} Utilizing a likelihood framework, our methodology jointly estimates all dynamic edge and vertex parameters. We further develop statistical inference procedures to test the dynamic nature of the underlying network. Concerning large-scale networks, we perform composite likelihood estimation with an $\ell_1$ penalty to discover sparse dynamic network structures. We establish the estimation error bound in $\ell_2$ norm and validate the sign consistency in the high-dimensional context. We apply our method to an influenza vaccine data set to model the dynamic gene network that evolves with time. We also investigate a Down syndrome data set to model the dynamic protein network which varies under a factorial experimental design. These applications demonstrate the applicability and effectiveness of the proposed model. The supplemental materials for this article are available online. 
\end{abstract}

\noindent%
{\it Keywords:}  Dynamic network analysis; Covariate-dependent GGM; $\ell_1$ penalization; Maximum likelihood estimation; Composite likelihood estimation.
\vfill

\newpage
\spacingset{1.75} 
\section{Introduction}
\label{sec:intro}
\textcolor{black}{\indent
Dynamic networks} are commonly observed in biological applications. For example, in the context of genetics, it has been demonstrated that the genetic networks can be impacted by common genetic variants \citep{FrancoLuisM.2013Igao, NicaAlexandraC.2013Eqtl}. Moreover, the genetic networks can also be affected by conditions such as cell type, cell cycle, stress response, and DNA damage \citep{ZhangZhihua2006Dcis, WeighillDeborah2022Pggr}. To describe the dynamic nature of the genetic network and highlight the differences in the network under different conditions, an effective approach is to perform dynamic network analysis. As another example, for those highly infectious vaccine-preventable diseases, getting immunized through vaccination serves as an effective strategy to prevent the spread of viruses. Researchers may desire to employ dynamic network analysis to establish a series of networks that evolve through time (i.e. before vaccination and at different time points after vaccination) for a more in-depth study of genetic factors involved in the diverse immune responses of healthy adults to a vaccine. More precisely, this dynamic network allows us to potentially recognize a group of genes that could be a critical determinant of the vaccine's effect on the immune response, which in turn might unlock innovative research opportunities in vaccine development \citep{FrancoLuisM.2013Igao}.

Over the last few decades, network analysis has become a prominent focus of research. One of the most widely used methodologies is known as the Gaussian graphical model (GGM), which aims to investigate interactions through conditional dependencies. In the field of Gaussian graphical models, many significant contributions have been made, including \citet{alma991011687879705164}, \cite{MeinshausenNicolai2006HGaV}, \cite{H2007IigG}, 
\cite{DobraAdrian2011BIfG},
\cite{alma991036188596505164},
\cite{10.1214/14-BA889},
\cite{PiotrZwiernik1},
\cite{DrtonMathias2017SLiG},
\cite{alma991029818979705164},
\cite{PiotrZwiernik2},
\cite{GraczykPiotr2022Msit},
and numerous others. 

However, the majority of existing works focus on investigating static network structures. New methods are needed to model dynamic network structures. A series of methods have been proposed for time-varying graphical models, including \citet{ZhouShuheng2010Tvug},
\citet{KolarMladen2010ETN}, \citet{KolarMladen2012Enwj}, \citet{wang2014inference}, \citet{MontiRicardoPio2014Etbc}, \citet{GibberdAlexanderJ.2017REoP} and \citet{YangJilei2020ETGM}. Other researchers proposed the covariate-adjusted GGMs, allowing the mean to be dependent on covariates while preserving a constant variance structure. \citet{GUOJIAN2011Jeom} conducted simultaneous estimation of multiple GGMs using a hierarchical penalty to characterize a series of graphs sharing a certain amount of common structure. 
\cite{YinJianxin2010NCM} introduced a nonparametric approach for modeling conditional dependencies, enabling the covariance matrix to vary with covariates. \citet{MladenKolar2010Osnc} focused on the high dimensional setting and developed a locally weighted kernel estimator for modeling sparse conditional dependencies. \citet{NiYang2019BGR} further explored the estimation procedure of covariate-dependent acyclic graph from the perspective of the Bayesian regression framework. \textcolor{black}{
\citet{JMLR:v23:21-0102} addressed the general problem of estimating undirected Gaussian graphical models (GGMs) conditioned on covariates (GGMx) using a Bayesian approach,
which allows for both the edge strength and the sparsity pattern of the underlying graph structure to change as functions of covariates. By constructing an MCMC algorithm for posterior inference, the estimates of the precision matrices can be obtained, and their positive definiteness can be guaranteed regardless of the covariates \citep{JMLR:v23:21-0102}.} 
\citet{ZhangJingfei2022HGGR} introduced a new type of Gaussian graphical regression model to describe covariate-dependent network structure, allowing both the mean structure and precision matrix to linearly vary with covariates. The model assumes all off-diagonal parameters of the precision matrix to vary with covariates whereas the diagonal parameters remain constant. They proposed a two-stage procedure that first estimates the ratios of off-diagonal and diagonal parameters by regressing each vertex against all the other vertices and then estimates the diagonal parameters via a moment method. \textcolor{black}{In contrast to existing methods that directly model the underlying graph structure as functions of covariates, \citet{Niu2} presented a novel Bayesian covariate-dependent graphical model and introduced an intermediate latent variable layer acting as a hidden bridge between the graphs and covariates, 
which offered a simpler alternative to parameterizing graph structure in terms of covariates.}

In this paper, we propose covariate-dependent GGM \textcolor{black}{(cdexGGM)} via a novel parameterization that guarantees the positive definiteness of $\bm{\Sigma}^{-1}$ under modest assumptions. We further relax the restriction on the diagonal parameters and allow both off-diagonal and diagonal parameters of the precision matrix to change with covariates. We propose to employ the likelihood approach to jointly estimate all the edge and vertex parameters. Our estimation procedure can ensure the symmetry of the estimates of the precision matrices as the likelihood approach inherently accounts for symmetric structure, while estimating parameters separately through the regression framework cannot guarantee this. For networks with a relatively small number of vertices, we apply maximum likelihood estimation and establish the asymptotic distribution of the estimates, which enables us to conduct statistical inferences on the dynamic network. To characterize large-scale dynamic networks with sparsity, we propose penalized composite likelihood estimation with $\ell_1$ penalization \citep{Lindsay, 
FearnheadPaul2002Almf,
CoxD.R.2004Anop, VarinCristiano2008Ocml, 
LarribeF.2011OCLI,
RibatetMathieu2012BIFC}. Theoretical properties of the estimator regarding parameter consistency and sign consistency are established in the high-dimensional context. Comprehensive simulation studies display satisfactory performance of the proposed likelihood procedure. We also provide two case studies - one in influenza vaccine, the other in Down syndrome, to demonstrate the practicability and efficiency of the approach. The organization of the paper is as follows. We establish the model setup for covariate-dependent GGMs \textcolor{black}{(cdexGGM)} and propose the maximum likelihood estimation for fixed $p$ in Section~\ref{Section 2}. In Section~\ref{Section 3}, we present the composite likelihood approach for large $p_n$ scenarios, along with a discussion on theoretical properties, including the estimation error bound and feature selection consistency. Numerical simulations for both the small-scale and large-scale networks are displayed in  Section~\ref{Section 4}.  Two examples of \textcolor{black}{covariate-dependent} network analysis are provided in Section~\ref{Section 5}. 

\section{Maximum Likelihood Approach for Covariate-Dependent Gaussian Graphical Models with Fixed $p$}
\label{Section 2}
\subsection{Model Setup and Notations}

\textcolor{black}{\indent
Let $
\bm Y = 
\big(\bm Y_1, 
\bm Y_2,
...,
\bm Y_n\big)^T
$} be a collection of independent random vectors, with each $\bm Y_m \sim N\big(\bm \mu_m(x_m), \bm \Sigma_m(x_m)\big), m=1,\cdots,n,$ where $x_m$ denotes the observed covariate and $Y_m = (Y_{m1},Y_{m2},...,Y_{mp})^{T}$. For the sake of simplicity, we assume that the observations are centered, namely $\bm \mu_m(x_m) = \bm 0$. First consider the simple case that $x_m$ is a univariate covariate, which is min-max scaled so that $0\leq x_m \leq 1.$  
We assume that $\bm \Sigma_m^{-1}(x_m)$ can be formulated as
\begin{equation}  \label{eq:newlabel1}
\begin{split}
\bm{\Sigma}_m^{-1}(x_m) = 
\bm{K}_m
&= x_m\bm{Q}_1 + (1-x_m)\bm{Q}_0  \\
&= 
\bm{Q}_0 + x_m (\bm{Q}_1 - \bm{Q}_0)
= 
\bm{Q}_0 + x_m \bm{P}_1. 
\end{split}
\end{equation}

The matrix $\bm Q_0$ is the precision matrix when $x_m=0$, and $\bm Q_1$ is the precision matrix when $x_m=1.$ For any arbitrary value of $x_m \in [0,1],$ the precision matrix is a weighted average of these two matrices. If $\bm{Q}_0$ and $\bm{Q}_1$ are assumed to be positive definite, $\bm{K}_m$ is guaranteed to be positive definite. After the reformulation, we have $\bm{K}_m$ expressed in a regression format where $\bm Q_0$ is the intercept matrix, and $\bm P_1$ is the slope matrix associated with the observed covariate.

We can further extend the model \eqref{eq:newlabel1} to the general setting of multiple covariates
\begin{align} \label{eq:newlabel2}
\bm{\Sigma}_m^{-1}(\bm{x_m}) = 
\bm{K}_m &= 
\sum_{h=1}^{H} 
x_m^{(h)}\bm{Q}_h +(1-\sum_{h=1}^H \frac{x_m^{(h)}}{H} )\bm{Q}_0\\ 
&= 
 \bm{Q}_0 + 
\sum_{h=1}^{H} 
x_m^{(h)}(\bm{Q}_h - \bm{Q}_0/H) 
=
\bm{Q}_0 + \sum_{h=1}^{H} 
x_m^{(h)} \bm P_{h},
\end{align}
where 
$
\bm{P}_h = \bm{Q}_h - \bm{Q}_0/H
$, $H$ is the number of covariates, $\bm{x_m}=(x_m^{(1)},\dots,x_m^{(H)})^T,$ and $x_m^{(h)}\in [0,1]$ represents the $h$th covariate of the $m$th observation. In \eqref{eq:newlabel2}, $\bm \Sigma_{m}^{-1}(\bm{x_m})$ is a linear combination of $\bm{Q}_0, \dots, \bm{Q}_H$ with non-negative coefficients. Under the assumption that all these matrices are positive definite,  $\bm \Sigma_{m}^{-1}(\bm{x_m})$  is guaranteed to be positive definite for any combinations of possible values of covariates. The $\bm{Q}_0$ is the baseline precision matrix when all covariates are set to zero. The matrices $\bm P_1,\cdots, \bm P_h$ can be treated as the slopes of the covariates to describe how the current network structure deviates from its baseline state resulting from $x_m^{(h)}.$ This general model of covariate-dependent GGMs \textcolor{black}{(cdexGGM)} encompasses standard GGMs, group-specific GGMs, and time-varying GGMs as special cases, as summarized in Table~\ref{table:0}.

\begin{table}
\caption{Special Cases of Covariate-Dependent GGMs}
{\begin{tabular}{ |c | c|  }
 \hline
 Special Cases  & Conditions 
 \\    \hline
 Standard GGMs   &  $ \forall \; h=1,\cdots, H, x_m^{(h)} = 0 $   \\ \hline
 Group-specific GGMs & $\forall \; h=1,\cdots, H, x_m^{(h)}$ is categorical  \\ \hline
 Time-varying GGMs & $h = 1, x_m^{(h)}$ is continuous  \\
 \hline
 General covariate-dependent GGMs &  $\forall \;  h=1,\cdots, H, x_m^{(h)}$ is continuous or categorical
 \\ \hline
\end{tabular}}
\label{table:0}
\end{table}

We use the following notations $\big(\bm{Q}_0\big)_{ij} = \alpha_{ij}, \big(\bm{P}_h\big)_{ij} = [\theta_{ij}]_h$ representing the $(i,j)$th element of $\bm{Q}_0$ and  $\bm{P}_h$ respectively. The set of unknown parameters can be formulated as $\bm{\beta}=(\bm{\alpha},\bm{\theta})^T,$ where $\bm{\alpha}$ is the vector of all $\alpha_{ij}, i\leq j$ and $\bm{\theta}$ is the vector of all $[\theta_{ij}]_h, i\leq j.$ Furthermore, we define an index set
$
E=\left\{(i, j): \; i \leq j; \; i,j = 1,...,p\right\}.
$
For $\forall \; u \in E,$ we define a indicator matrix $\bm{T}^{u}$ as
\begin{equation}      \label{eq:pro4}
	(\bm{T}^{u})_{ij}= \left\{            
	\begin{array}{ll}   
	1 \quad & \text{when }{\{i,j\} = u} \text{ or } {\{j,i\} = u},\\
	0 & \text{otherwise} .	
\end{array}\right.    
\end{equation}
Therefore, each matrix can be written as
\begin{equation} \label{eq:pro5}
	\bm Q_0 = \sum_{s: s \in E }\alpha_{s}\bm{T}^{s}, 
	\quad \bm P_h = \sum_{u:u \in E}[\theta_{u}]_h \bm{T}^{u}.
\end{equation}
Using the notations in \eqref{eq:pro4} and \eqref{eq:pro5}, the precision matrix can be rewritten as 
\begin{equation} \label{eq:pro6}
\bm K_m = 
\sum_{s:s \in E}\alpha_{s}\bm{T}^{s} + 
\sum_{h=1}^{H}
x_m^{(h)} \Big\{
\sum_{u: u \in E }[\theta_{u}]_{h}\bm{T}^{u}
\Big\}.
\end{equation}

\subsection{Maximum Likelihood Estimation}
\textcolor{black}{\indent
In this section,} we outline the maximum likelihood estimation procedure for small-scale \textcolor{black}{covariate-dependent} networks. Let 
$\mathcal{L}(\bm{\beta})$ 
represent the joint loglikelihood function,
\begin{align}
\label{multi-normal}
 \mathcal{L}(\bm{\beta}) =
\sum_{m=1}^{n} \mathcal{L}_m (\bm{\beta}|\bm{Y}_m)
=
-\frac{np}{2}\log{2\pi}-\frac{1}{2}
\sum_{m=1}^{n}\log{|\bm{\Sigma}_m|} 
 - 
\frac{1}{2}\sum_{m=1}^{n}
tr(
\bm{Y}_m^T\bm \Sigma_m^{-1}\bm Y_m
).
\end{align}
Note that \eqref{multi-normal} belongs to exponential family with the canonical parameters $ \bm{Q}_0,
\bm{P}_1,
\cdots,
\bm{P}_H
$ and sufficient statistics $T_{Q_0} = \sum_{m=1}^{n}\bm{Y}_m \bm{Y}_m^{T}, T_{P_h} = \sum_{m=1}^{n}x_m^{(h)}\bm{Y}_m \bm{Y}_m^{T}, h=1,\cdots, H$. For $\forall \;u,v,s,t \in E, \text{and } \forall \; h_1, h_2=1,\cdots,H,$ the first derivatives and the second derivatives of loglikelihood are given by
\begin{align*}
&\frac{\partial \mathcal{L}(\bm{\beta}) }{\partial [\theta_{u}]_{h_1}}=
\frac{1}{2}\sum_{m=1}^{n}
\Big\{
x_m^{(h_1)}tr(\bm\Sigma_m \bm{T}^u) - x_m^{(h_1)}tr(\bm W_m\bm{T}^u) 
\Big\},
\\
&\frac{\partial \mathcal{L}(\bm{\beta}) }{\partial \alpha_{s}}=
\frac{1}{2}\sum_{m=1}^{n}
\Big\{
tr(\bm\Sigma_m \bm{T}^s) - tr(\bm W_m\bm{T}^s) 
\Big\} ,
\\
&\frac{\partial^2 \mathcal{L}(\bm{\beta}) }{\partial [\theta_{u}]_{h_1} \partial [\theta_{v}]_{h_2}}=
-\frac{1}{2}\sum_{m=1}^{n}x_m^{(h_1)}
x_m^{(h_2)}tr(\bm{T}^u\bm \Sigma_m \bm{T}^v \bm \Sigma_m) ,
\\&
\frac{\partial^2 \mathcal{L}(\bm{\beta}) }{\partial [\theta_{u}]_{h_1} \partial \alpha_{s}}=
-\frac{1}{2}\sum_{m=1}^{n}
x_m^{(h_1)}tr(\bm{T}^u\bm \Sigma_m \bm{T}^s \bm \Sigma_m),
\\
&\frac{\partial^2 \mathcal{L}(\bm{\beta})}{\partial \alpha_{s} \partial \alpha_{t}}=
-\frac{1}{2}\sum_{m=1}^{n}tr(\bm{T}^s\bm \Sigma_m \bm{T}^t \bm \Sigma_m),
\end{align*}
where $\bm W_m = \bm Y_m\bm Y_m^T.$ For simplicity, we use the notation $\bm{ \mathcal{L}}_m^{(1)}(\bm \beta|\bm Y_m), \bm{ \mathcal{L}}_m^{(2)}(\bm \beta|\bm Y_m)$ representing the score vector and Hessian matrix for each observation respectively. Define the expected individual negative Hessian matrix as 
$
\bm V_m(\bm{\beta}) = E\Big\{  \big\{\bm{\mathcal{L}}_m^{(1)}(\bm \beta|\bm Y_m) \big\} \big\{ \bm{\mathcal{L}}_m^{(1)}(\bm \beta|\bm Y_m)\big\} ^ T
\Big\}
=
-E\Big\{
\bm{\mathcal{L}}_m^{(2)}(\bm \beta|\bm Y_m)
\Big\}
$. It can be shown that the information matrix $\bm{\mathcal{I}}(\bm \beta)$ takes the following form
\begin{align*}
\bm{\mathcal{I}}(\bm \beta)
=
\sum_{m=1}^{n} 
\bm V_m(\bm{\beta})
= 
\sum_{m=1}^{n}\bigg\{
\Big(\bm X_m \bm X_m^{T}\Big)
\otimes
\Big(
\bm \Sigma_m \otimes
\bm \Sigma_m
\Big)
\bigg\},
\end{align*}
where $\otimes$ represents the Kronecker product and $\bm X_m = \Big(1, x_m^{(1)}, x_m^{(2)}, \cdots, 
x_m^{(H)}
\Big)^{T}$. 
We adopt the iterative coordinate updating procedure proposed by \citet{JensenSorenTolver1991GCAf} and \citet{alma991011687879705164} to find the maximum likelihood estimate ($\text{MLE}$). The detailed updating equations at each step are outlined as follows,
\begin{equation} \label{eq:pro3}
\begin{split}
&
[{\hat{\theta}_u}^{(t+1)}]_{h_1} = 
[{\hat{\theta}_u}^{(t)}]_{h_1} + 
\frac{\Delta_u}{\sum_{m=1}^{n}\big\{{x_m}^{(h_1)}\big\}^2tr(\bm{T}^u\bm {\hat{\Sigma}}_m^{(t)}\bm{T}^u\bm{\hat{\Sigma}}_m^{(t)}) +\frac{{\Delta_u}^2}{2}} ,
\\
&
{\hat{\alpha}_s}^{(t+1)} = 
{\hat{\alpha}_s}^{(t)} + 
\frac{\Delta_s}{\sum_{m=1}^{n}
tr(\bm{T}^s\bm{\hat{\Sigma}}_m^{(t)}\bm{T}^s\bm{\hat{\Sigma}}_m^{(t)}) +\frac{{\Delta_s}^2}{2}} ,
\\
&
\Delta_u =
\sum_{m=1}^{n}x_m^{(h_1)}
\bigg\{
tr(\bm{\hat{\Sigma}}_m^{(t)}\bm{T}^u)-tr(\bm W_m\bm{T}^u)
\bigg\}  ,
\\ 
&
\Delta_s =
\sum_{m=1}^{n}
\bigg\{
tr(\bm{\hat{\Sigma}}_m^{(t)}\bm{T}^s)-tr(\bm W_m\bm{T}^s)
\bigg\} ,
\end{split}
\end{equation} 
where $\Delta_{u}, \Delta_{s}$ are the first derivatives of $\mathcal{L}(\bm{\beta})$ evaluated at the $(t)$th iteration.

\subsection{Theoretical Properties of MLE}
\textcolor{black}{\indent
In this section, we investigate the asymptotic behavior of the proposed maximum likelihood estimate when $p$ is fixed.}

\begin{assumption} 
\label{assumption 2.1}
Let $\Omega$ represent the parameter space for $\bm{\beta}$. We assume there exists an open subset $\bm{\omega} \in \bm{\Omega}$ that includes the true parameter $\bm{\beta}^0$ such that for $\bm{\beta} \in \bm{\omega},$ the eigenvalues of $\bm Q_0$ and $\bm Q_h,h=1,\cdots, H,$ are bounded below by $L_1>0$ and bounded above by $L_2>0.$ 
\end{assumption}

\begin{assumption}
\label{assumption 2.2}
Assume that the average Hessian matrix  $\frac{1}{n}\sum_{m=1}^{n}\bm V_m(\bm{\beta}^0) \xrightarrow{} \bm V$ as n goes to infinity, where  $\bm{\beta}^0$ 
represents the true value of $\bm{\beta}$ and $\bm V$ is positive definite.
\end{assumption}

\begin{theorem}
\label{theorem 2.1}
Under \textbf{Assumptions} \ref{assumption 2.1}-\ref{assumption 2.2}, there exists a local maximizer $\hat{\bm \beta}$ of $\mathcal{L}(\bm{\beta})$ such that $\vert \vert \hat{\bm \beta} - \bm{\beta}^{0} \vert \vert_2 = \mathcal{O}_p(n^{-\frac{1}{2}})$. Furthermore, 
the multivariate version of Lindeberg condition \eqref{eq:Lindeberg} holds: $\forall \; \epsilon >0$,
\begin{equation} \label{eq:Lindeberg}
\lim_{n \to \infty} \frac{1}{n}
\sum_{m=1}^{n} E\Big\{
 \vert \vert \bm{\mathcal{L}}_m^{(1)}(\bm \beta^0) \vert \vert_2 ^2 \cdot
I\Big(|| \bm{\mathcal{L}}_m^{(1)}(\bm \beta^0) ||_2 \geq \epsilon \sqrt{n} 
\Big)
\Big\}=0,
\end{equation}
and thus, the local maximizer $\hat{\bm{\beta}}$ satisfies the asymptotic normality
\begin{align*}
\sqrt{n}(\hat{\bm \beta} - \bm \beta^0)  \xrightarrow{\enskip d \enskip}
N\Big(\bm 0,  \bm V^{-1}  \Big)
\;\; \text{as }
n \xrightarrow{\enskip \enskip} \infty .
\end{align*}
\end{theorem}

Using the asymptotic distribution of the maximum likelihood estimate, we can construct test statistics and conduct statistical inference on the underlying \textcolor{black}{covariate-dependent graph structure.} For example, to test if the network is static versus dynamic, we could conduct inference on the slope parameters and test the hypothesis $H_0: \bm{\theta}=0$  versus $H_1: \bm{\theta}\neq 0.$ If the covariate is categorical, the hypotheses above can be used to test the equality of graphical models among different groups. We can also focus on a single parameter of the slope matrices and construct Wald-type statistic to test hypotheses $H_0: [\theta_{ij}]_h=0$  versus $H_1: [\theta_{ij}]_h\neq 0.$  Such test can be used to assess the $h$th covariate's effect on the $(i,j)$th edge in the graphical model.

\section{Composite Likelihood Approach for Covariate-dependent Gaussian Graphical Models with Large $p_n$}
\label{Section 3}
\subsection{Penalized Composite Log-likelihood Formulation}

\textcolor{black}{\indent In this section,} we use the notation $p_n$ to represent the dimension of random vectors and assume $p_n$ to increase with sample size. Note that the estimation procedure proposed earlier involves matrix inversion in each iteration, leading to computational challenges as $p_n$ grows. Moreover, large-scale networks typically exhibit sparse network structure. Motivated by these insights, we introduce a penalized composite likelihood approach 
\citep{Lindsay, 
FearnheadPaul2002Almf,
CoxD.R.2004Anop, VarinCristiano2008Ocml, 
LarribeF.2011OCLI,
RibatetMathieu2012BIFC,
GaoXin2015EoSG}
for large $p_n$ scenarios. The conditional distribution of a random variable 
$ Y_{mj}$ given $\bm Y_{m,V\setminus\{j\}}
$ is
$N \big(
      {\mu_{Y_{mj}|\bm Y_{m,V\setminus\{j\}}}}, 
   \sigma_{Y_{mj}|\bm Y_{m,V\setminus\{j\}}}^{2}
\big)$
with
\begin{align*} \label{eq:pro2}
&
\mu_{Y_{mj}|\bm Y_{m,V\setminus\{j\}}} = 
-\Big[
\alpha_{jj} + \sum_{h=1}^{H} x_m^{(h)}[\theta_{jj}]_{h}
\Big]^{-1}
\sum_{i=1,i\neq j}^{p_n} \Big[
\alpha_{ji} + \sum_{h=1}^{H}x_m^{(h)}
[\theta_{ji}]_{h}
\Big]y_{mi},
\\
&
\sigma_{Y_{mj}|\bm Y_{m,V\setminus\{j\}}}^{2} = \Big[\alpha_{jj} + \sum_{h=1}^{H}x_m^{(h)}[\theta_{jj}]_h\Big]^{-1},
\end{align*}
where $V=\big\{1,\cdots,p_n\big\}$ denotes the vertex set, and $V\setminus\{j\}$ denotes all the vertices except $j$.
The individual conditional loglikelihood function of $Y_{mj}$ given $\bm Y_{m,V\setminus\{j\}}$ is denoted by $l_c(Y_{mj}, \bm \beta)$. We construct the composite loglikelihood function by incorporating all the conditional distributions. Therefore, the overall negative joint composite loglikelihood function is given by 
\begin{align*}
&
-{l}_c(\bm{Y}, \bm \beta) = -\sum_{m=1}^{n} l_c(\bm Y_{m}, \bm \beta)
=
-\sum_{m=1}^{n} \sum_{j=1}^{p_n}  l_c(Y_{mj}, \bm \beta)
\\
&
=
-\frac{1}{2}\sum_{m=1}^{n}
\sum_{j=1}^{p_n}
\log{
\Big[
\alpha_{jj} + \sum_{h=1}^{H}x_m^{(h)}[\theta_{jj}]_{h}
\Big]} 
\\
& + \frac{1}{2}\sum_{m=1}^{n} \sum_{j=1}^{p_n} \Big[
\alpha_{jj} + \sum_{h=1}^{H}x_m^{(h)}
[\theta_{jj}]_{h} \Big]
\Big\{y_{mj}+\Big[
\alpha_{jj} + 
\sum_{h=1}^{H}x_m^{(h)}[\theta_{jj}]_{h} \Big]^{-1}
\sum_{i=1,i\neq j}^{p_n}
\Big[\alpha_{ji} + \sum_{h=1}^{H}x_m^{(h)}[\theta_{ji}]_{h} \Big]y_{mi}
\Big\}^2
\\
&
+ const.
\end{align*}
\textcolor{black}{\indent
The penalized composite} loglikelihood estimator $\hat{\bm{\beta}}$ is determined by minimizing the following 
objective function
\begin{equation} \label{Objective}
\min_{\bm \beta}Q(\bm \beta)
= 
-l_c(\bm{Y}, \bm {\beta}) 
+ n \lambda
\Big\{
\sum_{w \in \mathcal{E}} \abs{\alpha_w}
+ 
\sum_{s \in \mathcal{E}}
\sum_{h=1}^{H}
\abs{[\theta_{s}]_h}
\Big\},
\end{equation}
where $\mathcal{E} = \{ (i,j): \; i < j; \; i,j = 1,\cdots, p_n\}$ is the off-diagonal index set.
\textcolor{black}{The $\ell_1$ penalty is included in the objective function to enforce sparsity on the network structure. It should be emphasized that the sparsity patterns of $\bm{Q}_0, \bm{P}_1, \cdots, \bm{P}_{H}$ can differ,} \textcolor{black}{leading to graph structures that change with covariates.}
\textcolor{black}{
Let $\mathcal{E}_0, \mathcal{E}_1, \cdots, \mathcal{E}_{H}$ be the true nonzero off-diagonal parameter index sets of $\bm{Q}_0, \bm{P}_1, \cdots, \bm{P}_H$, namely $\mathcal{E}_0 = \{ (i,j): \alpha_{ij}^0 \neq 0, i < j 
\}$, $\mathcal{E}_h = \{ (i,j): \left[\theta_{ij}^0\right]_h \neq 0, i < j \}$, where $i, j = 1, \cdots, p_n, h = 1,\cdots, H$ and $\alpha_{ij}^0, \left[\theta_{ij}^0\right]_1, \cdots, \left[\theta_{ij}^0\right]_H$ are the true values of $\alpha_{ij}, \left[\theta_{ij}\right]_1, \cdots, \left[\theta_{ij}\right]_H$. Then the parametrization \eqref{eq:pro6} can be rewritten as
}
\begin{equation} \label{eq:pro7}
\bm K_m = \bm{\Sigma}_m^{-1} = 
\sum_{s:s \in \mathcal{E}_0 \cup V }\alpha_{s}\bm{T}^{s} + 
\sum_{h=1}^{H}
x_m^{(h)} \Big\{
\sum_{u: u \in \mathcal{E}_h \cup V }[\theta_{u}]_{h}\bm{T}^{u}
\Big\}.
\end{equation}
\newline
\textcolor{black}{
Under this parametrization, 
the sparsity pattern varies at specific values of the covariates. 
For example, if we have two covariates $x^{(1)}$ and $x^{(2)}$, the precision matrix is given by
$\bm{Q}_0+x^{(1)}\bm{P}_1 + x^{(2)}\bm{P}_2$. Let $\mathcal{G}_0,$ $ \mathcal{G}_1,$ and $ \mathcal{G}_2$ be the graph structures corresponding to $\bm{Q}_0, \bm{P}_1$ and $\bm{P}_2$. There are four possible graph structures: the graph structure is $\mathcal{G}_0$ when the covariates belong to the set $A=\{x^{(1)} = 0, x^{(2)} = 0\}$; the graph structure is $\mathcal{G}_0 + \mathcal{G}_1$ for the set $B=\{x^{(1)} \in (0, 1], x^{(2)} = 0\}$; the graph structure is $\mathcal{G}_0 + \mathcal{G}_2$ for the set $C=\{x^{(1)} = 0, x^{(2)} \in (0, 1]\}$; the graph structure is $\mathcal{G}_0 + \mathcal{G}_1 + \mathcal{G}_2$ for the set $D=\{x^{(1)} \in (0, 1], x^{(2)} \in (0, 1]\}$. We could obtain these possible graph structures as long as sets A, B, C, and D each have a non-zero probability under the covariate distribution. On the other hand, if the covariates are uniformly distributed within the unit square, the graph structure would almost surely be $\mathcal{G}_0 + \mathcal{G}_1 + \mathcal{G}_2$. However, the strengths of the edges keep varying as a linear function of the covariates.
}

\textcolor{black}{
As the true set of nonzero off-diagonal parameters is typically unknown in practice, we conduct penalized composite likelihood estimation and obtain sparse estimates of the parameters. Section~\ref{Section 3.3} shows that our proposed penalized composite likelihood estimator satisfies the parameter consistency and model selection consistency, }
\textcolor{black}{demonstrating its effectiveness in recovering the true graph structure which 
varies with covariates.
}

\subsection{Estimation Procedure}
\label{Section 3.2}
\textcolor{black}{\indent
To obtain the penalized maximum composite likelihood estimate, we propose to}
perform the coordinate descent algorithm \citep{FriedmanJerome2007PCO, TsengP.2001Coab}. A thorough exploration including the convergence properties of the estimators using the coordinate descent algorithm with the $\ell_1$ penalization can be found in \citet{BrehenyPatrick2011CDAF} and \citet{MazumderRahul2011SCDW}. Specifically, the closed-form updating expressions for off-diagonal parameters can be derived by differentiating $Q(\bm{\beta})$ with respect to $[\theta_{s}]_{h_1}, \alpha_{w}$ for all $s, w =(a,b) \in \mathcal{E}$ and $h_1=1,\cdots, H.$ The updating equations for each off-diagonal parameter given all the other parameters are as follows,
\begin{align*}
& \hat{\alpha}_{w} = 
\frac{S \bigg( -\Big\{\frac{2}{n}\sum_{m=1}^{n}
y_{ma}y_{mb} + \mathcal{I}_{1.1} + \mathcal{I}_{1.2}
\Big\},
\lambda
\bigg)}
{\mathcal{I}_{1.3}},
\\
& [\hat{\theta}_{s}]_{h_1} =
\frac{S \bigg(-\Big\{\frac{2}{n}\sum_{m=1}^{n}x_m^{(h_1)}y_{ma}y_{mb} + \mathcal{I}_{2.1} + \mathcal{I}_{2.2}
\Big\}, \lambda
\bigg)}
{\mathcal{I}_{2.3}},
\end{align*}
where $S(z,\lambda)=sign(z)(|z|-\lambda)_{+}$ represents the soft-thresholding operator. The notation $(|z|-\lambda)_{+}$ is defined as 
\begin{align*}
	(|z|-\lambda)_{+}= \left\{                
	\begin{array}{ll}   
		|z|-\lambda, & {\text{if}\,|z|-\lambda > 0,}\\
		0, & \text{otherwise}.	
	\end{array} \right. 
\end{align*} 

The terms $\mathcal{I}_{1.1},\mathcal{I}_{1.2}, \cdots, \mathcal{I}_{2.3}$ are defined in Section S1 of the supplementary materials. For the diagonal parameters, the partial derivatives are given as follows
\begin{equation}
\label{diagonal equations}
\begin{split}
\frac{\partial Q(\bm \beta)}{\partial \alpha_{jj}}
=&
- \frac{1}{2}
\sum_{m=1}^{n}\Big[\alpha_{jj} + \sum_{h=1}^{H}x_m^{(h)} [\theta_{jj}]_h \Big]^{-2}
\Big( 
\sum_{i=1,i\neq j}^{p_n}
\Big[\alpha_{ji} + \sum_{h=1}^{H}x_m^{(h)}[\theta_{ji}]_h\Big]y_{mi}
\Big)^2
\\
& 
- \frac{1}{2}\sum_{m=1}^{n}\Big[\alpha_{jj} + \sum_{h=1}^{H}x_m^{(h)} [\theta_{jj}]_h \Big]^{-1}
+ \frac{1}{2}
\sum_{m=1}^{n}y_{mj}^{2},
\\
\frac{\partial Q(\bm \beta)}{\partial [\theta_{jj}]_{h_1}}
=&
 - \frac{1}{2}\sum_{m=1}^{n}x_m^{(h_1)}\Big[\alpha_{jj} + \sum_{h=1}^{H}x_m^{(h)} [\theta_{jj}]_h \Big]^{-2}
\Big( 
\sum_{i=1,i\neq j}^{p_n}
\Big[\alpha_{ji} + \sum_{h=1}^{H}x_m^{(h)}[\theta_{ji}]_h\Big]y_{mi}
\Big)^2
\\
& -\frac{1}{2}\sum_{m=1}^{n}x_m^{(h_1)}\Big[\alpha_{jj} + \sum_{h=1}^{H}x_m^{(h)} [\theta_{jj}]_h \Big]^{-1}
+ \frac{1}{2}\sum_{m=1}^{n}x_m^{(h_1)}y_{mj}^{2},
\end{split}
\end{equation}
which do not have closed-form updating equations. We employ Broyden's method \citep{NocedalJorge2006NO} to solve this system of nonlinear equations. Consider a special case where the diagonal elements $(\bm \Sigma_m^{-1})_{jj}$ remain constant, i.e. $[\theta_{jj}]_h = 0, \alpha_{jj} > 0$ for all $j = 1,\cdots,p_n$ and $h=1,\cdots,H$. This suggests that the conditional variances are not influenced by covariates. Under this scenario, the closed-form updating expressions for diagonal parameters can be directly obtained, which simplifies the estimation procedure. See Section S2 of the supplementary materials for more details.

\subsection{Theoretical Properties}
\label{Section 3.3}
\textcolor{black}{\indent
In this section,} we provide the estimation error bound for the penalized composite likelihood estimator and demonstrate its model selection consistency under some regularity conditions. Let the score function and the Hessian matrix of ${l}_c(\bm{\beta})$ be denoted as ${l}^{(1)}_c(\bm{\beta}),{l}^{(2)}_c(\bm{\beta})$ respectively. We use the notation $\mathcal{S}_1$ to represent the set of true nonzero off-diagonal parameters and $\mathcal{S}_2$ to denote the set of all diagonal parameters. Then the set of true nonzero parameters can be denoted as $\bm{\mathcal{S}} = \mathcal{S}_1 \cup \mathcal{S}_2$. Define $\bm{\mathcal{S}}^c$ as the set of all zero parameters. Let $s_n$ denote the maximum number of nonzero parameters within each row across all rows and all matrices $\bm{Q}_0, \bm{P}_1,\cdots,\bm{P}_H.$ Let $q_n$ denote the total number of true nonzero off-diagonal parameters,
where $q_n = \mathcal{O}(p_ns_n)$.

\begin{assumption}
\label{assumption 3.1}
Define $H(\bm{\beta}) = E\big\{ 
-\frac{1}{n}
{l}^{(2)}_c(\bm{\beta})
  \big\}$ and there exists a neighborhood $||\bm{\beta} - \bm{\beta}^0||_{2} < \eta$ for some constant $\eta>0$ such that $H(\bm{\beta})$ has eigenvalues bounded away from zero and infinity.
  Furthermore, the absolute values of eigenvalues of $E \big\{
\frac{\partial^3 -l_c(\bm Y_m, \bm \beta) } {\partial \bm{\beta} \partial \bm{\beta}^T \partial \beta_{u}}
\big\}$ are bounded by some constant, where $\beta_u$ denotes any off-diagonal parameter.
\end{assumption}

\begin{assumption}
\label{assumption 3.2}
There exists a constant $d > 0$ such that $s_n^4 = \mathcal{O}(p_n^d)$ and $p_n^{1+d}\log{p_n}=o(n)$.
\end{assumption}

\begin{theorem}
\label{Theorem 2.2}
Under \textbf{Assumptions} \ref{assumption 3.1}-\ref{assumption 3.2}, if the tuning parameter $\lambda$ satisfies $\delta_1 \Big(
\frac{ s_n^2 \log{p_n}}{n}
\Big)^{\frac{1}{2}} \leq \lambda \leq  \delta_2 \Big(
\frac{ s_n^3 \log{p_n}}{n}
\Big)^{\frac{1}{2}}$ for some constants $\delta_1, \delta_2 >0$, then there exists a local minimizer $\bm{\hat{\beta}}$ of the objective function $Q(\bm{\beta})$ such that
$
||\bm{\hat{\beta}} - \bm{\beta^0}||_2 
= 
\mathcal{O}_p\Big\{ 
\Big( 
\frac{
p_n^{1 + d} \log{p_n}}{n} 
\Big)^{\frac{1}{2}} 
\Big\} .
$ 
\end{theorem}
\noindent Note that if we further assume that $s_n$ is bounded, the estimation error bound would be
$$||\bm{\hat{\beta}} - \bm{\beta^0}||_2 
= 
\mathcal{O}_p\Big\{ 
\Big( 
\frac{
p_n \log{p_n}}{n} 
\Big)^{\frac{1}{2}} 
\Big\}.
$$

\begin{assumption}
\label{assumption 3.3}
Let $H_{\mathcal{S}_1\mathcal{S}_1}^{0} = E
\Big\{
- \frac{1}{n}
l_c ^{(2)}(\bm{\beta^0})_{\mathcal{S}_1\mathcal{S}_1}
\Big\}
$ and 
$
H_{\mathcal{S}^{c}\mathcal{S}_1}^{0} = E\Big\{
-\frac{1}{n}
l_c ^{(2)}(\bm{\beta^0})_{\mathcal{S}^{c}\mathcal{S}_1}
\Big\}
$.
There exists some positive constant $ \xi \in (0,1)$ such that
\begin{align*}
\vertiii{  
H_{\mathcal{S}^{c}\mathcal{S}_1}^{0} \big(H_{\mathcal{S}_1\mathcal{S}_1}^{0}\big)^{-1}
}_{\infty} \leq 1-\xi ,
\end{align*}  
where $\vertiii{\cdot}_{\infty}$ is the maximum absolute row sum of a matrix.
\end{assumption}

This assumption states that the expected Hessian matrix satisfies the incoherence condition, which assumes that any two variables without an edge between them should not impose a large effect on variable pairs connected by edges in a Gaussian graphical model. This condition is widely used in high dimensional statistics, including graphical models and regularized regression models \citep{ZhaoPeng2006Omsc, MeinshausenNicolai2006HGaV, NIPS2008_61f2585b, RavikumarPradeep2011Hceb}. In the theorem above, we establish the consistency results for all vertex and edge parameters. As a next step, we investigate the problem of model selection focusing on the edge parameters only, while the vertex parameters are treated as nuisance parameters and are assumed to be known.

\begin{theorem}
\label{Theorem 2.3}
Under \textbf{Assumptions} \ref{assumption 3.1}-\ref{assumption 3.3} and \textbf{Theorem} \ref{Theorem 2.2}, suppose that
$q_n^2p_n^{d}\log{p_n} = o(n)$ and 
the minimum non-zero parameter satisfies
\begin{align*}
\min_{u \in \mathcal{S}_1}
\abs{
\beta_{u}^0
} := \min_{\substack{(i,j) \in \mathcal{S}_1 \\ h = 1,\cdots, H}} \big\{ \abs{[\theta_{ij}^0]_{h}}, \abs{\alpha_{ij}^0} \big\}
\geq 
c^{\ast}\sqrt{q_n}
\lambda
\end{align*}
for some constant $c^{\ast} > 0$. Then the penalized composite likelihood estimator $\bm{\hat{\beta}}$ satisfies $sign(\hat{\bm{\beta}}) = sign(\bm{\beta}^0)$ with probability tending to one.
\end{theorem}

\section{Numerical Simulations}
\label{Section 4}
\textcolor{black}{\indent
In this section, we provide several numerical studies to evaluate the performance of our proposed methods. Simulation results are summarized and reported over 100 data sets.
}

\subsection{Covariate-dependent Gaussian Graphical Models with Fixed $p$} \label{sec:simulation study 1}

\textcolor{black}{\indent
For the first simulation}, we consider small-sized \textcolor{black}{cdexGGM} with $p=10.$ Two different graphical structures are simulated for $Q_0,$ and $Q_1$ matrices, including a general network and a chain network with the sample sizes $n = 500$ or $3000.$ One covariate $x_m^{(1)}$ is simulated at 5 different values and for each specific value of the covariate, there are $n/5$ observations. The precision matrices are established by $\bm \Sigma_m^{-1} = x_m^{(1)}\bm Q_1 + (1-x_m^{(1)})\bm Q_0$. In the first case, the matrices $\bm Q_0$ and $\bm Q_1$ are generated as two random positive definite matrices by the format of $\bm{A} \bm{A} ^{T} + c \cdot \bm{I},$ where $A$ is a random square matrix, $I$ is the identity matrix, and $c$ is an arbitrary positive constant. 
In the second case, we generate a tridiagonal precision matrix corresponding to a chain network as outlined in \citet{Fan_2009}. Specifically, each off-diagonal element of the inverse matrices of $\bm Q_1$ and  $\bm Q_0$ is generated as 
$(\bm Q_1^{-1})_{ij} = \exp\big\{-1/2\abs{s_i-s_j}\big\},$ and $
(\bm Q_0^{-1})_{ij} = \exp\big\{-1/2\abs{s_i^{\prime}-s_j^{\prime}}\big\}
,$ where $s_1$ and $s_1^{\prime}$ are randomly initialized as positive numbers,
$
s_i - s_{i-1} \sim Unif(0.5, 1),$ and $
s_i^{\prime} - s_{i-1}^{\prime} \sim Unif(0.5, 1)
$. We evaluate the average $\ell_2$ norm of the error vector $\bm{e}=\hat{\bm{\beta}} - \bm{\beta}^0 $ over 100 simulations
and further divide it by the number of parameters as a measure of error for the proposed estimates. 
As shown by Table~\ref{sim table 1}, the measure of error is around 0.05 for $n=500,$ and around 0.01 for $n=3000.$  

\begin{table}
\caption{Simulation results of MLE for small-sized \textcolor{black}{cdexGGM}}  
\begin{threeparttable}
\begin{tabular}{>{\centering\arraybackslash}p{2.7cm}>{\centering\arraybackslash}p{2.7cm}>{\centering\arraybackslash}p{2.7cm}>{\centering\arraybackslash}p{2.7cm}>{\centering\arraybackslash}p{2.7cm}>
{\centering\arraybackslash}p{2.7cm}>
{\centering\arraybackslash}p{2.7cm}}  \toprule
Sample Size & Case & Matrix & Mean $||\bm{e}||$ & Mean $||\bm{e}||$ /\# of Para. \\ 
    \midrule
    \multirow{4}{*}{500} & 
    \multirow{2}{*}{General}&
    $\bm Q_0$  &  
    2.6645(1.0731) &
    0.0484
    \\
    \addlinespace &
    &$\bm Q_1$  &
         2.6304(0.9897) & 
         0.0478 
    \\
    \addlinespace
    &
     \multirow{2}{*}{Chain}&
     $\bm Q_0$  &  
    2.8332(0.6690) &
    0.0515
    \\
    \addlinespace
    & 
    &$\bm Q_1$  &
         2.8655(0.6518) & 
         0.0521
    \\
    \midrule
    \multirow{4}{*}{3000} & 
    \multirow{2}{*}{General}&
    $\bm Q_0$  &  
    0.9573(0.3243) &
    0.0174
    \\
    \addlinespace &
    &$\bm Q_1$  &
         0.8403(0.3017) & 
         0.0153 
    \\
    \addlinespace
    &
     \multirow{2}{*}{Chain}&
    $\bm Q_0$  &  
    0.9705(0.1571) &
    0.0176
    \\
    \addlinespace
    & 
    &$\bm Q_1$  &
         0.9835(0.1603) & 
         0.0179 
    \\
    \bottomrule
\end{tabular}
\label{sim table 1}
\begin{tablenotes}
\item The number of vertices $p$ is 10 and the total $\#$ of parameters is 55 for each matrix. The standard errors are listed in parentheses.  
\end{tablenotes}
\end{threeparttable}
\end{table}

\subsection{Covariate-dependent Gaussian Graphical Models with Large $p_n$} 
\label{sec:simulation study 2}

\textcolor{black}{\indent
In this section, we consider high dimensional covariate dependent GGMs with} \textcolor{black}{$p_n = 30$ and $p_n = 50$, along with the same sample size $n = 3000$ and one covariate $H=1$. The simulation results for $p_n=100$ are also provided in Section S4.2 of the supplementary materials. We generate random matrices $\bm{Q}_0$ and $\bm{Q}_1$ with a sparse network structure, using the construction of $\bm{A} \bm{A} ^{T} + c \cdot \bm{I},$ where $A$ is a random square matrix. The diagonal elements of the precision matrix are set to be linearly varied with the covariate in this study.
}
\textcolor{black}{The tuning parameter $\lambda$} is determined by composite loglikelihood EBIC \citep{Foygel2010, GaoXin2010CLBI}\textcolor{black}{, which is given below}
\begin{align}
EBIC_{\gamma} 
&= 
- 2l_c(\hat{\bm{\beta}_c}) + df\log{n}
+ 4 df \gamma \log{p_n},
\end{align}
where $df$ denotes the total number of nonzero off-diagonal entries and $\gamma \in [0,1]$ is a user-specified constant. \textcolor{black}{We perform a sensitivity analysis on the influenza vaccination dataset to investigate the impact of $\gamma$ on the estimation results. Please refer to Section S3 of the supplementary materials for details. For our simulation studies, we set $\gamma = 1$.}
We use different model evaluation metrics including specificity, sensitivity, and Matthews Correlation Coefficient (MCC), which are defined as 
\begin{align*}
	sensitivity = \frac{TP}{TP+FN},
 \qquad \qquad
	specificity = \frac{TN}{TN+FP} ,
 \\
	MCC = \frac{TP\times TN - FP \times FN}{\sqrt{\left(TP+FP\right) \left(TP+FN\right)\left(TN+FP\right)\left(TN+FN\right)}},
\end{align*}
where $TP$, $TF$, $FP$, and $FN$ represent the total number of true positives, true negatives, false positives, and false negatives respectively. 

\textcolor{black}{
The proposed method (cdexGGM) is compared with four other methods: graphical lasso \citep{GUOJIAN2011Jeom}, fused graphical lasso, group graphical lasso \citep{DanaherPatrick2014jglf}, and the Gaussian graphical regression model \citep{ZhangJingfei2022HGGR}. Graphical lasso (glasso) is employed to jointly estimate multiple graphical models for different categories and the tuning parameter $\lambda$ is selected by the Bayesian information criterion (BIC) \citep{GUOJIAN2011Jeom}.  
The fused graphical lasso (FGL) and the group graphical lasso (GGL) are similar approaches with different penalties, both utilized for modeling Gaussian graphical models across different categories. The GGL imposes a weaker similarity constraint (only encourages similar sparsity patterns) on the multiple precision matrices compared to the FGL, which not only emphasizes similar edge values across multiple categories but also enforces similar sparsity patterns. The tuning parameters $\lambda_1, \lambda_2$ are selected based on an approximation of the Akaike Information Criterion (AIC) \citep{DanaherPatrick2014jglf}. The Gaussian graphical regression model (RegGMM) employs a regression-based approach to estimate graphical models that vary with covariates. The tuning parameters $\lambda$ and $\lambda_g$ are jointly selected by $5$-fold cross-validation, as suggested in \citet{ZhangJingfei2022HGGR}. To ensure a fair comparison, the covariate is uniformly selected from 0 to 1 with 5 different values. 
}

\textcolor{black}{
Table~\ref{sim table 2} highlights the average of these metrics over 100 simulation repetitions. Standard errors for each metric are provided in parentheses. For $p_n=30$, our proposed method demonstrates superior performance in sensitivity and MCC, while maintaining comparable specificity to other approaches. Both the GGL and FGL demonstrate acceptable performance,
despite not being intended for modeling covariate-dependent network structure. The RegGMM assumes constant diagonal parameters across all levels of covariates and therefore does not perform well under the varying diagonal parameters setting. Our proposed method outperforms competing approaches across all measured criteria when $p_n = 50$. Other methods perform well in sensitivity, but at the expense of specificity, leading to reduced MCC scores.
}

\begin{table}
\caption{\label{sim table 2}Simulation results of our proposed method and other four competing methods.}
\centering 
\begin{threeparttable}
\begin{tabular}{>{\centering\arraybackslash}p{0.7cm}>
{\centering\arraybackslash}p{1.3cm}>{\centering\arraybackslash}p{1.5cm}>{\centering\arraybackslash}p{1.5cm}>
{\centering\arraybackslash}p{2.5cm}>
{\centering\arraybackslash}p{2.5cm}>{\centering\arraybackslash}p{2.5cm}} \toprule
$p_n$ & $\#$ of Para. & Method & Matrix &  
  Sensitivity & Specificity & MCC   \\
    \midrule
    \multirow{10}{*}{$30$} & 
    \multirow{10}{*}{$465$} 
    & 
    \multirow{2}{*}{cdexGGM}
    & 
    $\bm{Q}_0$  &  
    0.8747(0.0307) & 0.6576(0.0526) & 0.4401(0.0424) \\
    \addlinespace
    &  &
    & $\bm{Q}_1$ &
         0.9251(0.0232) & 0.6275(0.0573) & 0.4388(0.0450) \\
    \addlinespace
    & 
    & 
    \multirow{2}{*}{GGL}
    &
    $\bm{Q}_0$ &
    0.7904(0.0207) & 
    0.6183(0.0224) & 
    0.3361(0.0259)
    \\
    \addlinespace
    & &
    & $\bm{Q}_1$ &
    0.7876(0.0179) & 0.7099(0.0169) & 0.4036(0.0217) \\
    \addlinespace
    &  & 
    \multirow{2}{*}{FGL}
    &
    $\bm{Q}_0$ &
    0.7902(0.0208) &
    0.6147(0.0225) &
    0.3328(0.0259)
    \\
    \addlinespace
    &  & 
    & $\bm{Q}_1$ &
    0.7880(0.0178)   &
    0.7090(0.0168)   &
    0.4030(0.0221) \\
    \addlinespace
    &  &
    \multirow{2}{*}{glasso}
    &  
    $\bm{Q}_0$ &
    0.8003(0.0194)  &
    0.5848(0.0226)  &
    0.3160(0.0243)  \\
    \addlinespace
    &  & 
    &  $\bm{Q}_1$   &
    0.7926(0.0162)  &
    0.6918(0.0166)  &
    0.3903(0.0198)  \\
    \addlinespace
    &  & 
    \multirow{2}{*}{RegGMM}
    &  
    $\bm{Q}_0$ &
    0.6035(0.1476)  &
    0.6333(0.1083)  &
    0.2018(0.0574) \\
    \addlinespace
    &  & 
    &  $\bm{Q}_1$  &
    0.7135(0.1190)  &
    0.5850(0.1442)  &
    0.2458(0.0491)
    \\
    \midrule
    \multirow{10}{*}{$50$} & 
    \multirow{10}{*}{$1275$} 
    & 
    \multirow{2}{*}{cdexGGM}
    & 
    $\bm{Q}_0$  & 0.9091(0.0236)
    & 0.8202(0.0239) &
    0.3737(0.0251)
    \\
    \addlinespace
    & & 
    &  $\bm{Q}_1$  & 0.7838(0.0282)
    & 0.9011(0.0145) & 0.4152(0.0275)
     \\
    \addlinespace
    &  & 
    \multirow{2}{*}{GGL}
    & $\bm{Q}_0$  & 0.7568(0.0203)
    & 0.7603(0.0101)  & 0.2454(0.0130)
    \\
    \addlinespace
    &  & 
    & $\bm{Q}_1$  &  0.8098(0.0292)
    &  0.6551(0.0105)   & 
    0.1948(0.0131)
    \\
    \addlinespace
    &  & \multirow{2}{*}{FGL}
    &  $\bm{Q}_0$  & 0.7570(0.0203)
    &  0.7577(0.0103)    &
    0.2435(0.0130)
    \\
     \addlinespace
    & & 
    & $\bm{Q}_1$  & 0.8104(0.0285)
    &  0.6400(0.0101)  &
    0.1872(0.0127)
    \\
    \addlinespace
    &  & \multirow{2}{*}{glasso}
    &  $\bm{Q}_0$  
    & 0.7573(0.0203)
    & 0.7528(0.0108)
    & 0.2400(0.0129)
    \\
    \addlinespace
    & &
    & $\bm{Q}_1$  
    & 0.8113(0.0250)
    & 0.6249(0.0098)
    & 0.1801(0.0113)
    \\
    \addlinespace
    &  & \multirow{2}{*}{RegGMM}
    &  $\bm{Q}_0$  &  0.8025(0.0739)
    & 0.5431(0.0909)  &
    0.1470(0.0215)
    \\
     \addlinespace
    & & 
    & $\bm{Q}_1$  &  0.7904(0.0917)
    &  0.5039(0.1077)  &
    0.1211(0.0251)
    \\
    \bottomrule
\end{tabular} 
\begin{tablenotes} \item The standard errors for each evaluation metric are shown in parentheses. The $\#$ of true nonzero parameters of $\bm{Q}_0, \bm{Q}_1$ are 123 and 114 when $p_n = 30$, and 106 and 102 when $p_n=50$.
\end{tablenotes}
\end{threeparttable}
\end{table}

\subsection{High-Dimensional Covariate-dependent Gaussian Graphical Models with Multiple Covariates}
\label{sec:simulation study 3}

\textcolor{black}{\indent
In the third simulation study, we limit our comparison to our proposed method and RegGMM, as both are designed for modeling covariate-dependent network structures. We examine two covariates, generated uniformly from 0 to 1. We set $p_n = 20$ and $p_n = 30$, with the sample size $n=3000$. For the purpose of a fair comparison, we consider the diagonal entry of the precision matrix to be constant across all levels of covariates.  Refer to Section S4.1 of the supplementary materials for a detailed explanation of the simulation process for $\bm{Q}_0, \bm{P}_1, \bm{P}_2$.}

\textcolor{black}{As shown in Table~\ref{sim table 3}, both methods demonstrate strong performance in estimating the sparsity patterns of the slope matrices $\bm{P}_1, \bm{P}_2$. In terms of the baseline precision matrix $\bm{Q}_0$, our method 
strikes a balance between high sensitivity and relatively good specificity, while RegGMM offers slightly higher sensitivity but at the cost of specificity. Overall, the proposed method exhibits consistent and reliable performance across all settings and evaluation metrics.
}

\begin{sidewaystable}
\caption{\label{sim table 3}Simulation results of our proposed method and RegGMM on multiple covariates}
\centering
\begin{threeparttable}
\begin{tabular}{>{\centering\arraybackslash}p{1.0cm}>
{\centering\arraybackslash}p{1.5cm}>{\centering\arraybackslash}p{1.5cm}>{\centering\arraybackslash}p{1.5cm}>{\centering\arraybackslash}p{2.4cm}>{\centering\arraybackslash}p{2.5cm}>
{\centering\arraybackslash}p{2.5cm}>
{\centering\arraybackslash}p{2.5cm}} \toprule
    $p_n$ & Method & Matrix &  \# of Para. &
  \# of Nonzero Para. &
  Sensitivity & Specificity & MCC  \\
    \midrule
    \multirow{6}{*}{$20$} & \multirow{3}{*}{cdexGGM} &
    $\bm{Q}_0$  & 210 & 55 &  
    0.8409(0.0409) & 0.6870(0.0780) & 0.4215(0.0601) \\
    \addlinespace
    &  & $\bm{P}_1$ & 190 & 67 &
         0.7576(0.0328) & 0.8410(0.0945) & 0.5997(0.0986) \\
    \addlinespace
    &  & $\bm{P}_2$ & 190 & 70 &
         0.9097(0.0268) & 0.8351(0.0820) & 0.7286(0.0790) \\
    \addlinespace
    & \multirow{3}{*}{RegGMM}
    & $\bm{Q}_0$  &  210  & 55  &
    0.8866(0.0498) &  0.3015(0.0393)
    & 0.1652(0.0495)
\\
    \addlinespace
    &  & $\bm{P}_1$ & 190 &  67 &
    0.6760(0.0433)  & 0.8913(0.0634)
    &  0.5926(0.0776)
\\
    \addlinespace
    &  & $\bm{P}_2$ & 190  & 70  &
    0.7943(0.0332)  & 0.8521(0.0488)
    & 0.6427(0.0589)
\\
    \midrule
    \multirow{6}{*}{$30$} & 
    \multirow{3}{*}{cdexGGM} &
    $\bm{Q}_0$ & 465 & 116 &
    0.8877(0.0276) & 0.7432(0.0505) & 0.5202(0.0447) \\
    \addlinespace
    & & $\bm{P}_1$ &  435 & 161 &
         0.8035(0.0296) & 0.8954(0.0450) & 0.7040(0.0493) \\
     \addlinespace
    &  & $\bm{P}_2$ & 435 & 147 &
         0.8301(0.0281) & 0.8874(0.0527) & 0.7125(0.0598) \\
    \addlinespace
    & \multirow{3}{*}{RegGMM}
    & $\bm{Q}_0$  & 465  & 116  &
    0.8888(0.0328)  & 0.3221(0.0250)  &
    0.1873(0.0303)
\\
    \addlinespace
    &  & $\bm{P}_1$ & 435  & 161  &
    0.6350(0.0510) & 0.8837(0.0602)  &
    0.5481(0.0547)
\\
    \addlinespace
    &  & $\bm{P}_2$ & 435  & 147  &
    0.6972(0.0300) & 0.8345(0.0491) &
    0.5317(0.0547)
\\
    \midrule
\end{tabular}
\begin{tablenotes}
\item The sample size is $3000$, and the standard errors over 100 datasets are listed in parentheses. 
\end{tablenotes}
\end{threeparttable}
\end{sidewaystable}

\section{Real Data Analysis}
\label{Section 5}
\subsection{Influenza Vaccination Study}
\textcolor{black}{\indent
Nowadays, influenza continues to be a major concern for public health, including the most prevalent strains influenza A (H1N1), influenza B, and influenza A (H3N2). Insights obtained from animal models have revealed that host genetic factors can have a profound impact on both the immune responses and susceptibility to influenza infection \citep{TrammellRitaA2008Gsar,SrivastavaBarkha2009Hgbs}. The immune response to vaccination, much like in viral infection, is expected to exhibit variability that is influenced by genotype as well. Motivated by this, we apply our proposed model to a global gene expression data set (before and after trivalent influenza vaccination in humans) intending to characterize the dynamic regulatory gene networks and examine the \textcolor{black}{varying dependency structure at different time points.} This analysis may provide insights into pivotal genes that greatly contribute to the immune response to influenza vaccination. 
}

The gene expression data set is downloaded from the National Center for Biotechnology Information 
website ( https://www.ncbi.nlm.nih.gov/geo/query/acc.cgi?acc=GSE48024). The total sample size is 848, including 119 healthy adult male volunteers (aged 19 - 41 years) and 128  healthy adult female volunteers (aged 19 - 41 years). It records global transcript abundance in peripheral blood RNA specimens before (day 0) and at three time points (days 1, 3, and 14) after vaccination. Raw data is normalized by background adjustment, variance stabilization transformation \citep{LinSimonM.2008Mvtf} and robust spline normalization using the R package {\it lumi} \citep{DuPan2008lapf}. 

Similar to \citet{FrancoLuisM.2013Igao}, we opt to narrow down our focus to 68 genes that exhibit both a transcriptional response to the vaccine and some evidence of genetic regulation. In summation, the cleaned data set has $p_n=68$ genes, with the sample size $n=848$ and four different time points (day 0, days 1, days 3, days 14). We apply the penalized composite loglikelihood approach to analyze the data set. The estimated time-varied regulatory networks and sparsity pattern are given in Figure~\ref{Real data set 1 image 1} and Figure~\ref{Real data set 1 image 2}, respectively. 

As shown by our analysis, several gene interactions experience a rapid decrease after vaccination, including ABCA7 \& LRRC37A4, KIAA0391 \& RABEP1, C3AR1 \& OAS1 and DYNLT1 \& OAS1. Some of those even disappear at the end of the 14 days, like D4S234E \& OAS1, EMR3 \& OAS1 and GRINA \& RABEP1. In contrast, getting vaccinated substantially enhances the interactions DIP2A \& MGC57346 and OAS1 \& TAP2. It also produces the newly-formed edges DIP2A \& NT5DC3 and OAS1 \& TIMM10. These findings demonstrate some evidence that DIP2A and OAS1 are crucial to people's immune response. This aligns with the results reported in \citet{FrancoLuisM.2013Igao}, which recognizes the importance of DIP2A and OAS1 in the humoral immune response to influenza vaccination. In general, our research validates the pivotal role of a specific set of genes in the human immune response to vaccination. Further biological investigations are necessary to fully understand the dynamic relationships among these genes with respect to vaccination time and their biological significance. 

\begin{figure}
	\centering
\includegraphics[width=15cm, height=7.5cm]{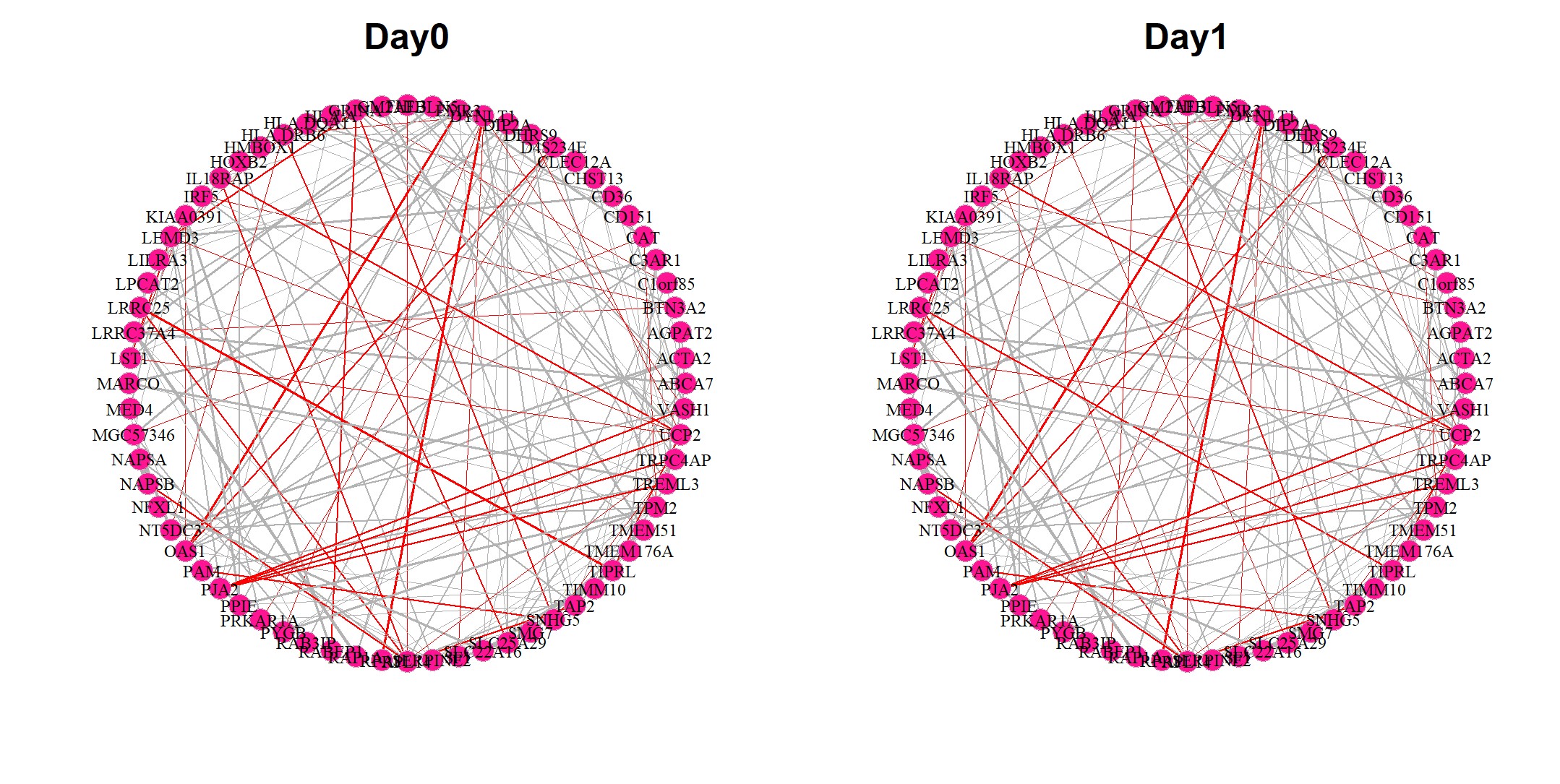}\\
\includegraphics[width=15cm, height=7.5cm]{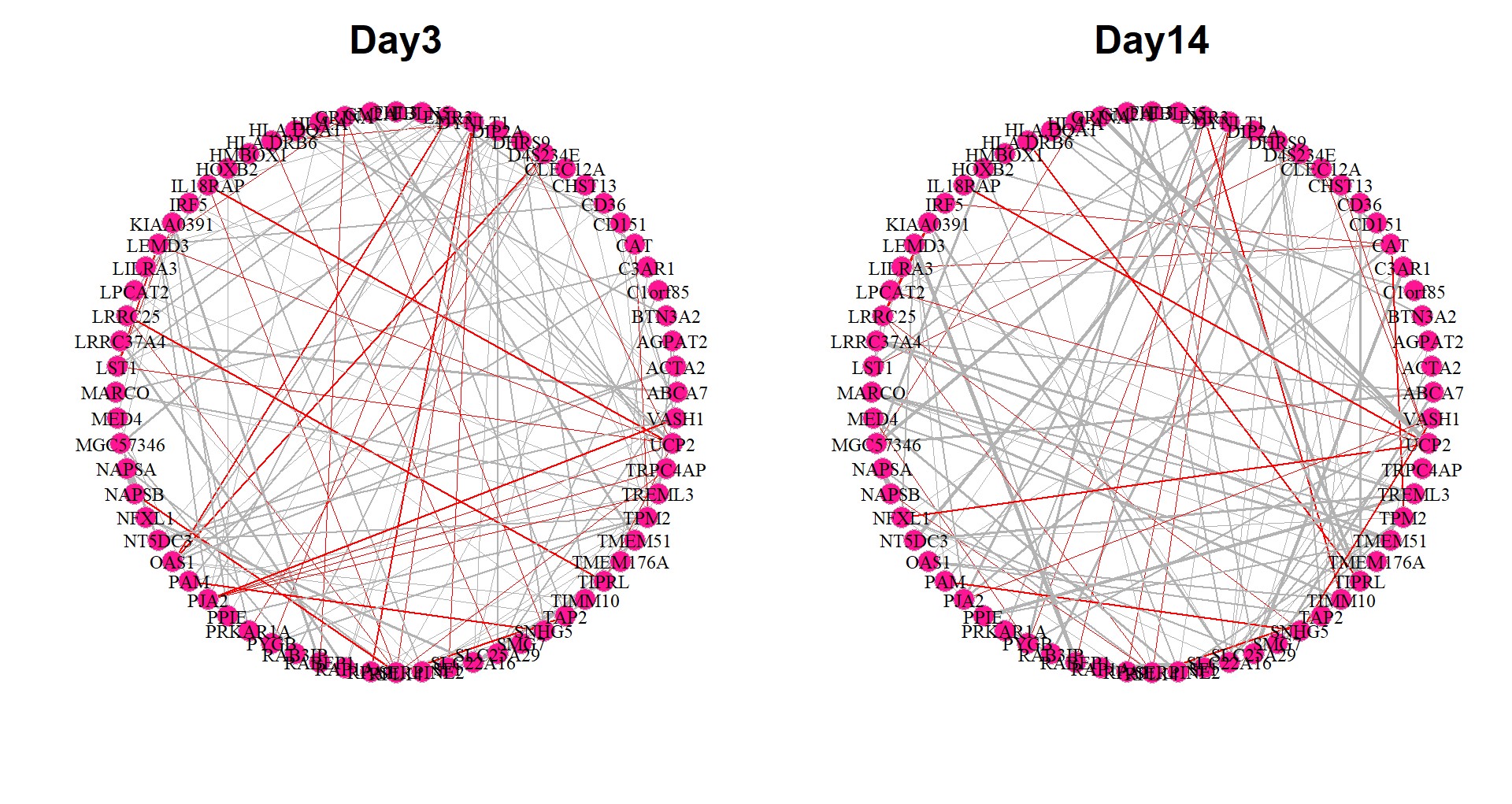}
	\caption{The estimated gene regulatory networks for the influenza vaccine data set. The time points "Day0", "Day1", "Day3" and "Day14" represent the day before getting vaccinated, the first day after vaccination, the third day after vaccination, and two weeks after vaccination. Edges in red represent negative partial correlations between genes, while edges in gray stand for positive partial correlations. The strength of interactions between two genes is displayed by the thickness of the edges.
 }
\label{Real data set 1 image 1}
\end{figure}

\begin{figure}
\begin{center}
\centerline{\includegraphics[width=15cm, height=7.5cm]{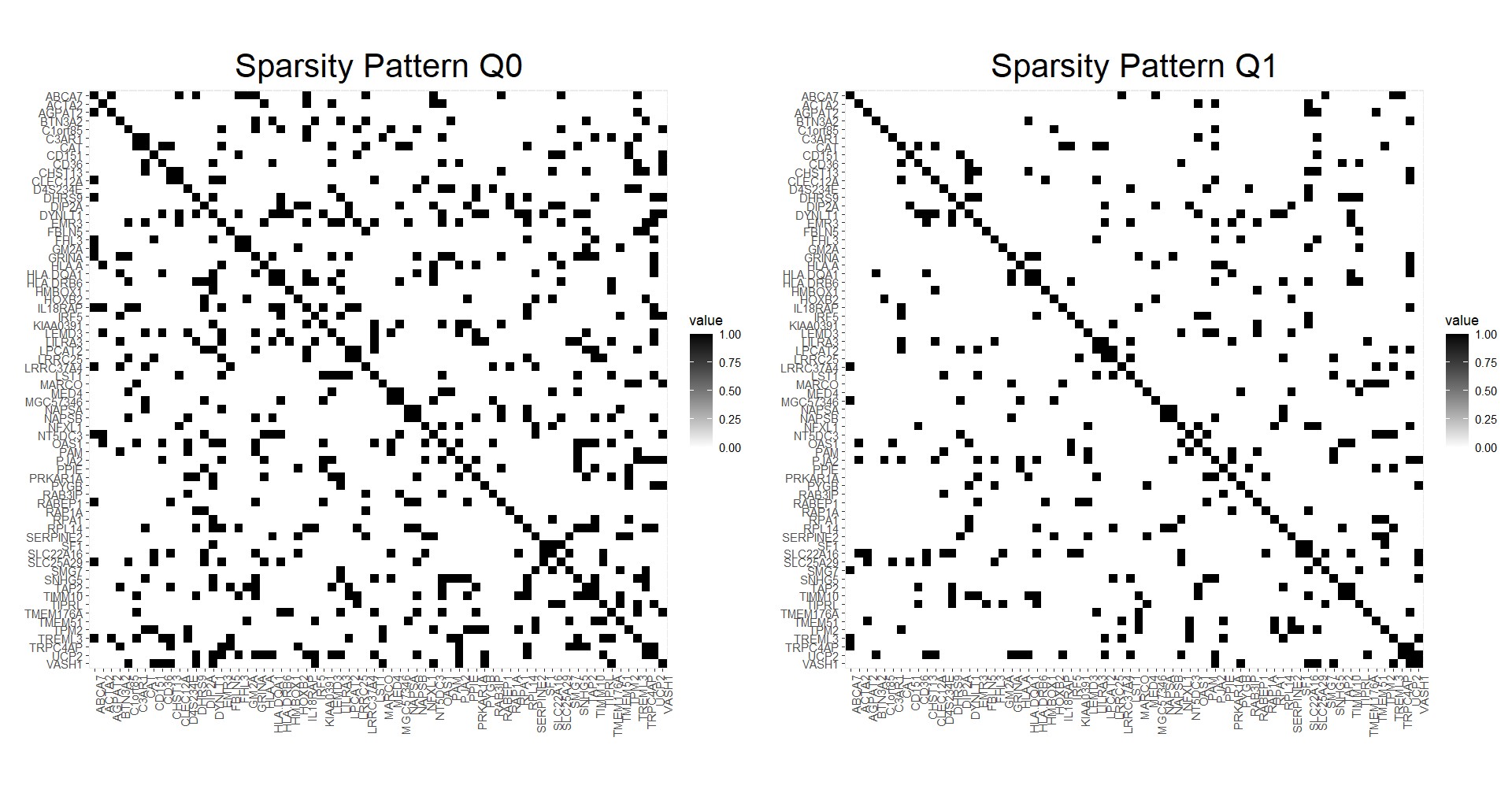}}
\end{center}
\caption{The estimated sparsity pattern for the influenza vaccine data set. Each image has 68 rows and 68 columns, representing the elements of $\bm{Q}_0, \bm{Q}_1$ respectively.  The corresponding cell is marked in black if the estimated entry is nonzero. Otherwise, it is white.}
\label{Real data set 1 image 2}
\end{figure}

\subsection{Mice Protein Expression} 
\label{Real Data Analysis 2}

\textcolor{black}{\indent Down syndrome (DS), also known as trisomy 21, is a genetic disorder caused by the presence of an extra copy of chromosome 21 pertaining to intellectual disability, influencing approximately one in a thousand newborns globally. It is generally considered that the excessive expression of genes encoded by the extra copy of chromosome gives rise to the disruption of regular pathways and typical reactions towards stimulation, and further leads to learning and memory deficiency. In the following analysis, we utilize the proposed covariate-dependent GGM method to analyze a mice protein expression data set to discover differences and similarities among protein networks influenced by three different covariates. The covariates include the treatment effect of injecting the drug memantine (m) or not (s), the factor of having the stimulation to learn (CS) or not (SC), and the factor of normal genotype (c) or abnormal trisomy (t). The resulting 8 classes of mice are grouped by the level of these three factors. Each class is examined and labeled by the associated learning outcome including normal learning, failed learning and rescued learning (refer to Figure~\ref{Figure: dataset 2 No1} for details). We show that our proposed model enables us to investigate the differences in the networks across different levels of three experimental factors. These findings shed light on how the genetic factor, medical treatment, and stimulation factor influence the interactions among the proteins, and how these interactions in turn affect the learning outcomes. 
}

This mice protein expression data set \citep{HigueraClara2015Sfmi} is downloaded from the UCI Machine Learning Repository (https://archive.ics.uci.edu/ml/datasets/Mice+Protein+Expr\\ession). It comprises the expression levels of 77 proteins that produce detectable signals in the nuclear fraction of the cortex. In total, there are 72 mice involved in the study, with 38 of them serving as control mice and 34 as trisomic mice with Down syndrome. In the experiment, 15 measurements of each protein are recorded per mouse. Accordingly, there are $570$ measurements for control mice and $510$ measurements for trisomic mice, with a total of 1080 measurements per protein. The mice are evenly distributed across eight classes, with each class comprising around seven to ten mice. Excluding those proteins with missing values and further conducting feature selection through a random forest algorithm, we finally narrow down our focus to the top 35 proteins with the highest feature importance scores. We proceed to center each observation for the subsequent analysis.

To investigate how covariates influence the underlying dependence structure, we employ the covariate-dependent GGMs through the proposed estimation approach. The following Figure~\ref{Figure Data set 2: Networks} illustrates the estimated protein network structures for the 35 selected proteins across eight classes. We further examine the estimated $[\theta_{ij}]_h$ in slope matrices $\bm{P}_1$, $\bm{P}_2,$ an $\bm{P}_3.$ We select the top-ranked $[\hat{\theta}_{ij}]_h$ in terms of magnitude and perform hypothesis testing $H_0: [\theta_{ij}]_h=0$ versus $H_1: [\theta_{ij}]_h\neq 0$, which is used to test whether or not
the corresponding covariate has a significant effect on the specific protein interactions. 
We construct Wald-type statistics and the standard errors are obtained through the bootstrap method. Table~\ref{Table 5} highlights those pairs of protein interactions and provides the associated p-values.   For example, we observe a positive effect of the trisomy genotype on the interactions Tau\_N \& P3525\_N, PKCA\_N \& pNUMB\_N and pRSK\_N \& RSK\_N, while the following protein pairs RAPTOR\_N \& pGSK3B\_N, PKCA\_N \& CAMKII\_N and pRSK\_N \& RSK\_N exhibit a negative response to the injection of memantine. It is interesting to note that hypothesis testing on the pair pRSK\_N \& RSK\_N is significant for both genotype and treatment. Furthermore, the trisomic geneotype is increasing this interaction, whereas the medical treatment is decreasing this interaction. This suggests that the treatment alleviates the harmful impact of the abnormal trisomy genotype on this specific pair of protein interaction. Additionally, we observe that both the memantine treatment and the learning stimulation assistance impose a positive impact on pJNK\_N \& pGSK3B\_N and meanwhile,  negatively influence pNUMB\_N \& pMTOR\_N. These results may offer valuable insights into the importance of those protein pairs in the mechanism of learning disabilities in Down syndrome.

As shown in Figure~\ref{Figure: dataset 2 No1}, different mice classes have different learning outcomes, including normal learning, rescued learning, and failed learning. We are interested in comparing dependence structures between two different classes with different learning outcomes. We highlight two sets of comparisons: 1)  the class "c-CS-m" demonstrating normal learning outcome versus the class "t-CS-m" demonstrating rescued learning outcome; 2) the class "c-CS-s" demonstrating normal learning outcome versus the class "t-CS-s" demonstrating failed learning outcome. We perform a two-sample test to compare the equality of partial correlations between any given protein pair across the two classes. The test statistic is constructed based on the composite likelihood estimates of the partial correlations and the standard errors are obtained through the bootstrap method. Both investigations reveal a notable enhancement in the interactions ITSN1\_N \& BRAF\_N, DYRK1A\_N \& ITSN1\_N, pGSK3B\_Tyr216\_N \& SHH\_N, pP70S6\_N \& pRSK\_N, S6\_N \& ADARB1\_N, AcetylH3K9\_N \& S6\_N and pCAMKII\_N \& ADARB1\_N in normal mice, compared to those experiencing either rescued or failed learning. The p-values of the two-sample tests are provided in Table~\ref{Table 6}. These findings provide evidence that the interactions between these protein pairs are essential to the learning outcome. Some of our results agree with existing findings in biological literature. For example, \citet{MalakootiNakisa2020TLIo}  
demonstrates that ITSN1\_N contributes to normal learning and memory in mice. \citet{AhmedMd.Mahiuddin2015Pdaw} reveals that overexpression of DYRK1A can lead to learning and memory deficits. Both proteins have significant results in our two-sample test. Additional investigations into the significant protein interactions identified in our analysis are necessary to elucidate the underlying biological mechanisms.

\begin{figure}
\begin{tikzpicture}[node distance=3.3cm]

\node (start) [startstop] {Mice};

\node (control) [control, below of = start, xshift = -3.5cm] {Control Mice (c)};

\node [label=right:{\textbf{\;\;\;\;\;\;\;\; \;\;\;\;\;\;\;Genotype}}](trisomic) [trisomic, below of = start, xshift = 3.5cm] {Trisomic Mice (t)};

\node (in1) [stimulated, below of=control, xshift = -1.7cm] {Yes: Context Shock (CS)};

\node (in2) [stimulated, below of=control, xshift = 1.7cm] {No: Shock Context (SC)};

\node (in3) [stimulated, below of=trisomic, xshift = -1.7cm] {Yes: Context Shock (CS)};

\node [label=right:{\textbf{\;\;\;\;Stimulated to learn}}](in4) [stimulated, below of=trisomic, xshift = 1.7cm] {No: Shock Context (SC)};

\node (tr1) [treatment, below of=in1, xshift=-0.9cm] {Yes};
\node (tr2) [treatment, below of=in1, xshift=0.9cm] {No};
\node (tr3) [treatment, below of=in2, xshift=-0.9cm] {Yes};
\node (tr4) [treatment, below of=in2, xshift=0.9cm] {No};

\node (tr5) [treatment, below of=in3, xshift=-0.9cm] {Yes};
\node (tr6) [treatment, below of=in3, xshift=0.9cm] {No};
\node (tr7) [treatment, below of=in4, xshift=-0.9cm] {Yes};
\node [label=right:{\textbf{\;Treatment}}] (tr8) [treatment, below of=in4, xshift=0.9cm] {No};

\node (class1) [class3, below of=tr1] {c-CS-m};
\node (class2) [class3, below of=tr2] {c-CS-s};
\node (class3) [class2, below of=tr3] {c-SC-m};
\node (class4) [class2, below of=tr4] {c-SC-s};

\node (class5) [class, below of=tr5] {t-CS-m};
\node (class6) [outcome, below of=tr6] {t-CS-s};
\node (class7) [class2, below of=tr7] {t-SC-m};
\node [label=right:{\textbf{  Classes}}] (class8) [class2, below of=tr8] {t-SC-s};

\node (outcome1) [class3, below of=class1, xshift = 0.7cm] {Normal};
\node (outcome2) [class2, below of=class3, xshift = 0.7cm] {No Learning};
\node (outcome3) [class, below of=class5] {Rescued};
\node (outcome4) [outcome, below of=class6] {Failed};
\node [label=right:{\textbf{\;\;\;\;\;Learning Outcome}}](outcome5) [class2, below of=class7, xshift = 0.7cm] {No Learning};

\draw [arrow] (start) -- (control);
\draw [arrow] (start) -- (trisomic);
\draw [arrow] (control) -- (in1);
\draw [arrow] (control) -- (in2);
\draw [arrow] (trisomic) -- (in3);
\draw [arrow] (trisomic) -- (in4);
\draw [arrow] (in1) -- (tr1);
\draw [arrow] (in1) -- (tr2);
\draw [arrow] (in2) -- (tr3);
\draw [arrow] (in2) -- (tr4);
\draw [arrow] (in3) -- (tr5);
\draw [arrow] (in3) -- (tr6);
\draw [arrow] (in4) -- (tr7);
\draw [arrow] (in4) -- (tr8);

\draw [arrow] (tr1) -- (class1);
\draw [arrow] (tr2) -- (class2);
\draw [arrow] (tr3) -- (class3);
\draw [arrow] (tr4) -- (class4);

\draw [arrow] (tr5) -- (class5);
\draw [arrow] (tr6) -- (class6);
\draw [arrow] (tr7) -- (class7);
\draw [arrow] (tr8) -- (class8);

\draw [arrow] (class1) -- (outcome1);
\draw [arrow] (class2) -- (outcome1);
\draw [arrow] (class3) -- (outcome2);
\draw [arrow] (class4) -- (outcome2);

\draw [arrow] (class5) -- (outcome3);
\draw [arrow] (class6) -- (outcome4);

\draw [arrow] (class7) -- (outcome5);
\draw [arrow] (class8) -- (outcome5);
\end{tikzpicture}
\caption{Eight Classes of Mice and the associated Learning Outcomes}
\label{Figure: dataset 2 No1}
\end{figure}
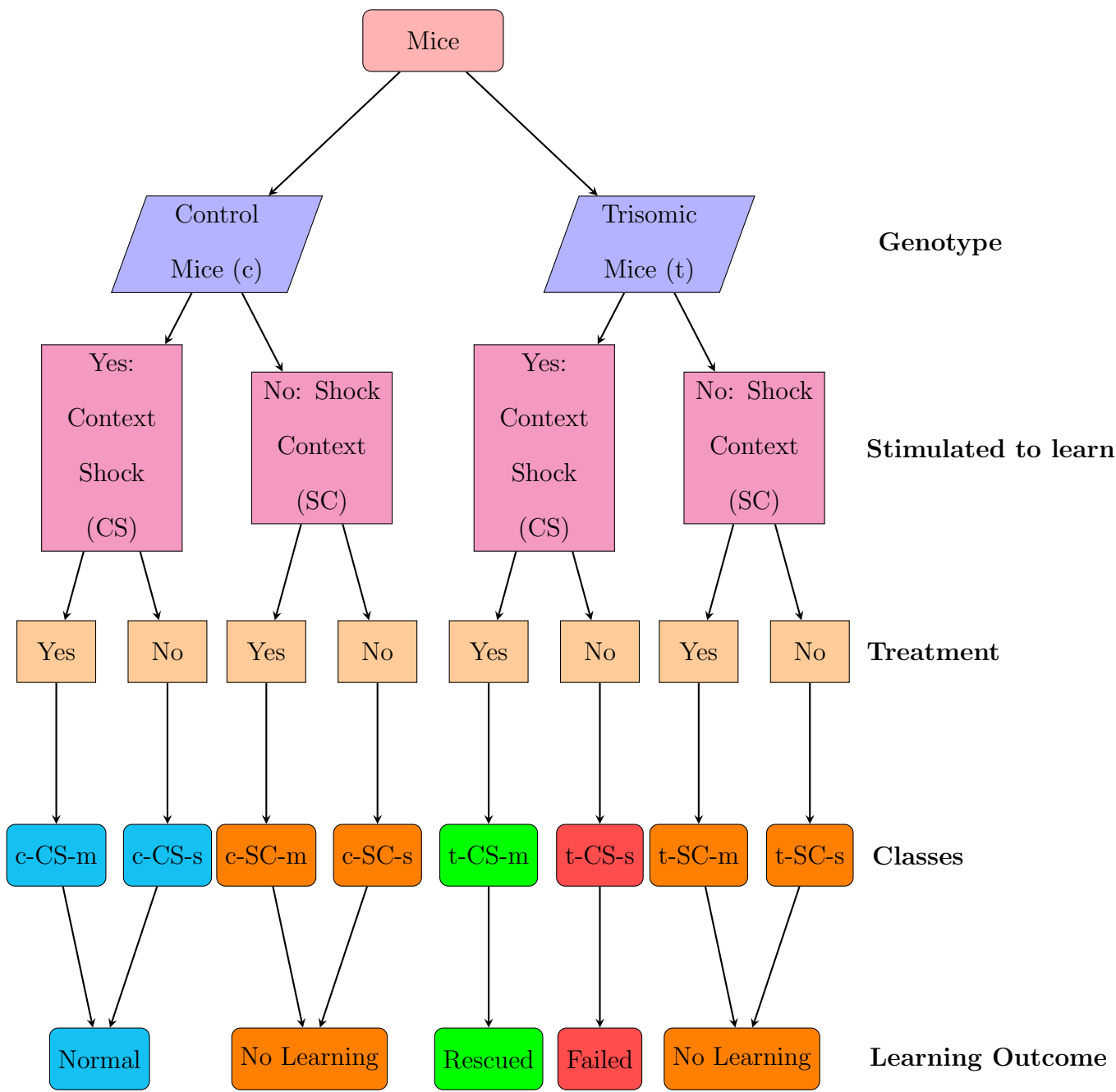

\begin{figure}
\begin{center}
\centerline{\includegraphics[width=17cm, height=8.5cm]{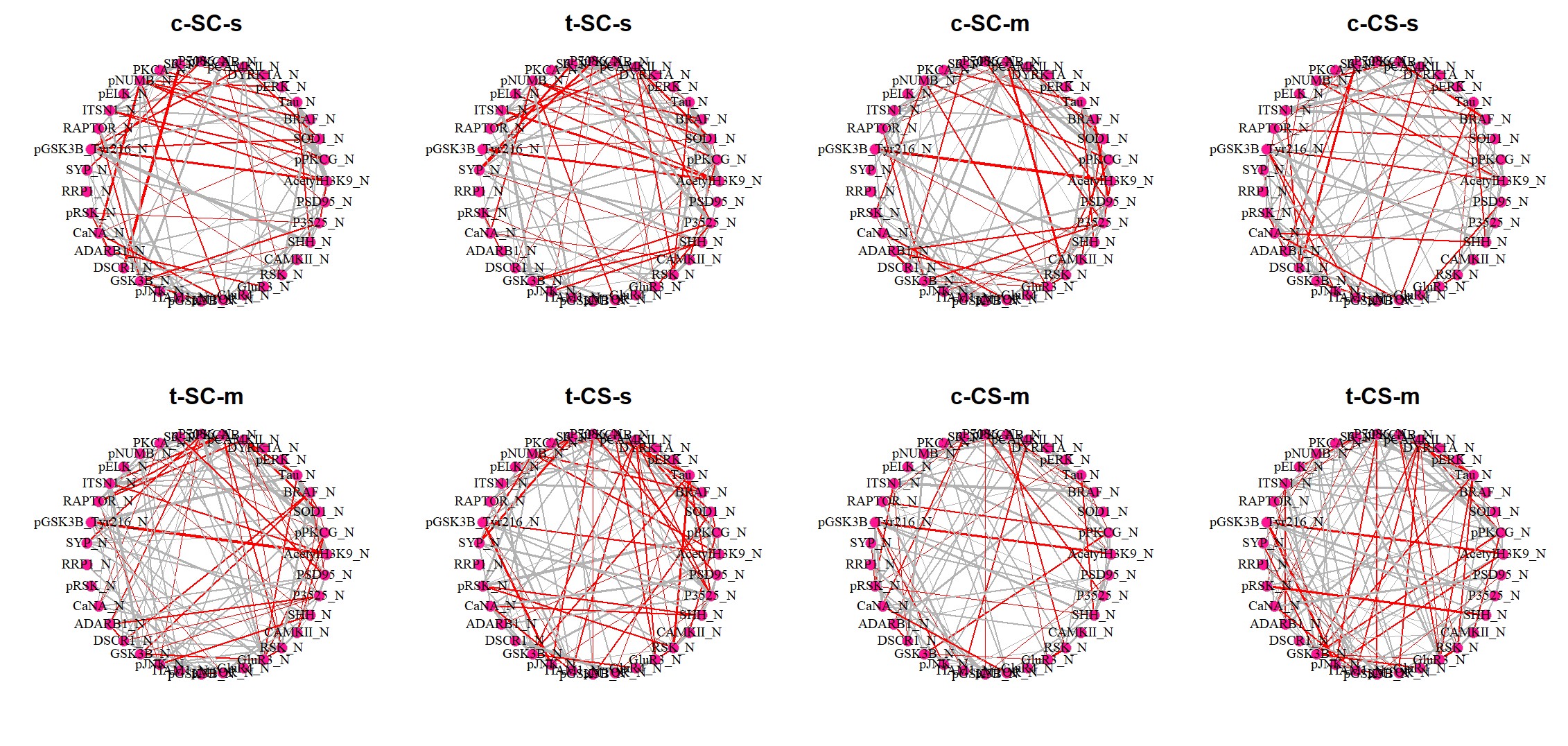}}
\end{center}
\caption{Estimated protein dependence structure: c-CS-m (control group, stimulated to learn, memantine); t-CS-m (trisomy group, stimulated to learn, memantine); c-CS-s (control group, stimulated to learn, saline); c-SC-m (control group, not stimulated to learn, memantine); t-CS-s (trisomy group, stimulated to learn, saline); t-SC-m (trisomy group, not stimulated to learn, memantine); c-SC-s (control group, not stimulated to learn, saline); t-SC-s (trisomy group, not stimulated to learn, saline). }
\label{Figure Data set 2: Networks}
\end{figure}

\begin{table}
\caption{\label{Table 5} Top-ranked gene pairs with significant covariate effects}
\begin{threeparttable}
\begin{tabular}{>{\centering\arraybackslash}p{1.0cm}>{\centering\arraybackslash}p{1.0cm}>{\centering\arraybackslash}p{5.5cm}>{\centering\arraybackslash}p{5.5cm}} \toprule
\multicolumn{1}{c}{Covariates} &
\multicolumn{1}{c}{Matrices} & 
\multicolumn{2}{c}{Interactions}
\\
\cline{3-4}
    & & Positive $[\hat{\theta}_{ij}]_h$ & Negative $[\hat{\theta}_{ij}]_h$ \\
    \midrule
    Genotype & 
    $\bm{P}_1$ & \thead{Tau\_N \& P3525\_N \,(9.86E-11) \\} 
    \thead{PKCA\_N \& pNUMB\_N \,(5.83E-3) \\}
    \thead{ \textcolor{black}{pRSK\_N \& RSK\_N \,(8.67E-4)} \\}
    \thead{RSK\_N \& P3525\_N \,(4.55E-4) \\}
    \thead{Tau\_N \& RSK\_N \,(1.60E-2) }
     & 
    \thead{pNUMB\_N \& RAPTOR\_N \,(5.19E-3)\\}
    \thead{ pJNK\_N \& P3525\_N \,(9.12E-3) \\}
    \thead{
    Tau\_N \& PKCA\_N \,(1.19E-5) \\}  \thead{Tau\_N \& SHH\_N \,(2.95E-5) \\}
    \thead{pNUMB\_N \& pGSK3B\_N \,(1.41E-2)}
    \\
    \midrule
    Treatment & $\bm{P}_2$
    &  
    \thead{pNUMB\_N \& RSK\_N \,(4.11E-4) \\}
    \thead{pNUMB\_N \& pRSK\_N \,(1.75E-7) \\}
    \thead{
    PKCA\_N \& P3525\_N \,(1.24E-2) \\}
    \thead{
    \textcolor{black}{
    pJNK\_N \& pGSK3B\_N \,(7.78E-2)} \\}
    \thead{
    pGSK3B\_N \& P3525\_N \,(8.47E-2) }
    & 
    \thead{
    RAPTOR\_N \& pGSK3B\_N \,(2.86E-2) \\
    }
    \thead{
    \textcolor{black}{pRSK\_N \& RSK\_N \,(1.54E-3)}
    }
    \thead{
    PKCA\_N \& CAMKII\_N \,(4.99E-4) \\
    }
    \thead{ \textcolor{black}{
    pNUMB\_N \& pMTOR\_N \,(3.59E-5)} \\
    }
    \thead{
    Tau\_N \& RSK\_N \,(3.08E-2) }
    \\
    \midrule
    Behaviour & $\bm{P}_3$
    &
    \thead{ \textcolor{black}{
    pJNK\_N \& pGSK3B\_N \,(2.75E-2)} \\}
    \thead{GluR3\_N \& RSK\_N \,(3.50E-2) \\
    }
    \thead{
    pMTOR\_N \& P3525\_N \,(8.48E-9) \\
    }
    \thead{pNUMB\_N \& pGSK3B\_N \,(7.12E-2) \\}
    \thead{
    CAMKII\_N \& SHH\_N \,(4.37E-4)
    }
    &
    \thead{
    PKCA\_N \& pGSK3B\_N \,(1.27E-2) \\}
    \thead{
    pJNK\_N \& SHH\_N \,(2.74E-7) \\}
    \thead{
    PKCA\_N \& P3525\_N \,(1.68E-4) \\}
    \thead{
    RAPTOR\_N \& RSK\_N \,(8.51E-2) \\}
    \thead{
    \textcolor{black}{
    pNUMB\_N \& pMTOR\_N \,(6.88E-5)}
    } \\
\bottomrule
\end{tabular}
\begin{tablenotes}
\item Note: The associated p-values are provided in parentheses. 
\end{tablenotes}
\end{threeparttable}
\end{table}

\renewcommand{\theadalign}{bc}
\begin{table}
\caption{\label{Table 6} Test of equality of partial correlations across two classes with different learning outcomes}
\begin{threeparttable}
\begin{tabular}{>{\centering\arraybackslash}p{3.5cm}>{\centering\arraybackslash}p{2.7cm}>{\centering\arraybackslash}p{2.7cm}>{\centering\arraybackslash}p{2.5cm}>{\centering\arraybackslash}p{2.5cm}} \toprule
\multicolumn{1}{c|}{Edges} & 
\multicolumn{4}{c}{Comparison between two classes}
\\
\cline{2-5}
    \multicolumn{1}{c|}{} & \thead{Normal learning \\ c-CS-m} & \thead{Rescued learning \\ t-CS-m} & \multicolumn{1}{|c}{\thead{Normal learning \\ c-CS-s}}
     & \thead{Failed learning \\ t-CS-s} \\
    \midrule
    BRAF\_N \& ITSN1\_N & \multicolumn{2}{c}{1.75E-11}
     &   \multicolumn{2}{c}{1.35E{-8}}
    \\ 
    \midrule
    DYRK1A\_N \& ITSN1\_N & 
    \multicolumn{2}{c}{5.85E{-5}} & 
    \multicolumn{2}{c}{2.86E{-4}}
    \\
    \midrule
    \thead{pGSK3B\_Tyr216\_N \& \\ SHH\_N}
    & \multicolumn{2}{c}{$<$2.2E-16} & 
    \multicolumn{2}{c}{$<$2.2E-16} 
    \\
    \midrule
    pP70S6\_N \& pRSK\_N &
    \multicolumn{2}{c}{1.60E{-8}} & \multicolumn{2}{c}{3.28E{-4}}
    \\ 
    \midrule
    S6\_N \& ADARB1\_N & 
    \multicolumn{2}{c}{$<$2.2E-16} & \multicolumn{2}{c}{$<$2.2E-16}
    \\
    \midrule
    AcetylH3K9\_N \& S6\_N &
    \multicolumn{2}{c}{$<$2.2E-16} & \multicolumn{2}{c}{$<$2.2E-16}
    \\
    \midrule
    pCAMKII\_N \& ADARB1\_N & \multicolumn{2}{c}{$<$2.2E-16} & \multicolumn{2}{c}{$<$2.2E-16}
    \\
\bottomrule
\end{tabular}
\begin{tablenotes}
\item Note: The table provides the p-values associated with the two-sample tests to compare the equality of partial correlations between protein pairs across two different classes. 
\end{tablenotes}
\end{threeparttable}
\end{table}

\section{\textcolor{black}{Some Remarks on the Covariates}}

\textcolor{black}{ \indent
The main reason for assuming bounded covariates in our model is to ensure the positive definiteness of the precision matrix, given any possible values of covariates. Specifically, the precision matrix $\bm \Sigma_{m}^{-1}(\bm{x_m})$ is assumed to be a linear combination of $\bm{Q}_0, \dots, \bm{Q}_H$ with coefficients $1 - \sum_{h=1}^H \frac{x_m^{(h)}}{H}, x_m^{(1)}, \cdots, x_m^{(H)}$, as shown in \eqref{eq:newlabel2}.
Assuming that $\bm{Q}_0, \cdots, \bm{Q}_H$ are all positive definite matrices, the corresponding precision matrix is ensured to be positive definite for any covariates in the range $[0, 1]$. When the covariates are bounded, one can apply Min-Max scaling to transform the covariates into the range of $[0,1].$ }

\textcolor{black}{To make Min-Max scaling less sensitive to outliers, hypothetical minimum and maximum values can be used instead of the true actual extremes of a given dataset. This results in a more robust transformation while maintaining values within the range $[0, 1]$. In particular, the hypothetical minimum and maximum values can be chosen sufficiently large, ensuring that the transformation remains stable despite adding or removing a data point, thereby preserving the consistency of model estimation and interpretation. In asymptotic theory, all theoretical results remain valid regardless of the type of transformation applied to the covariates, as long as the covariates are transformed to $[0,1]$.}

\textcolor{black}{
However, when covariates are unbounded, non-linear transformation techniques such as the quantile transformation can be applied. This strategy starts by determining the CDF of the original feature and then maps the original values to a uniform distribution. 
Nevertheless, this transformation reshapes the underlying distribution of the covariates and may be more suitable for a nonlinear relationship between the pre-transformed covariates and the underlying graph structure.
}

\section{Discussion} \label{sec:Conclusion}
\textcolor{black}{\indent
In this paper,} we present a covariate-dependent Gaussian graphical model (cdexGGM) for modeling \textcolor{black}{dynamic network structures that vary with covariates via a novel parameterization,} ensuring the positive definiteness of precision matrices at different values of covariates. We propose to jointly estimate dynamic edge and vertex parameters through a likelihood framework. For networks with small dimensions, we conduct the maximum likelihood estimation and establish the asymptotic distribution of the estimator. This enables us to construct various statistical inference procedures on the underlying dynamic network. In the context of large networks, we employ the composite likelihood method with an $\ell_1$ penalty to estimate sparse network structures. We provide the estimation error bound in $\ell_2$ norm and establish the model selection consistency of the proposed method. The proposed algorithm combines the coordinate descent algorithm and Broyden's method to update the dynamic edge and vertex parameters respectively. The algorithm performs well in high-dimensional settings and exhibits satisfactory performance. As shown by the real data analysis, the proposed cdexGMM can be widely applied to model emerging dynamic biological networks in real-world applications.

\section*{Supplementary Materials}

\begin{description}
\item[Supplementary Materials cdexGMM.pdf:] Additional information related to this article, including some notations, a special case of the proposed method, simulation studies, sensitivity analysis, and technical proofs. (Supplementary Materials cdexGMM, .pdf file)
\item[cdexGMM Folder:] A folder including the R scripts for simulation studies, real data analysis, and a real dataset. (cdexGMM, folder)
\item[Read me.txt in cdexGMM Folder:]
A detailed description of the files in the cdexGMM Folder. (Read me, .txt file)
\end{description}

\section*{Funding}
The Natural Sciences and Engineering Research Council of Canada supported this work.

\section*{Disclosure Statement}
The authors report there are no competing interests to declare.

\bibliographystyle{chicago}
\bibliography{biblio}

\begin{thebibliography}{}

\bibitem[\protect\citeauthoryear{Ahmed, Dhanasekaran, Block, Tong, Costa,
  Stasko, and Gardiner}{Ahmed et~al.}{2015}]{AhmedMd.Mahiuddin2015Pdaw}
Ahmed, M.~M., A.~R. Dhanasekaran, A.~Block, S.~Tong, A.~C.~S. Costa, M.~Stasko,
  and K.~J. Gardiner (2015).
\newblock Protein dynamics associated with failed and rescued learning in the
  ts65dn mouse model of down syndrome.
\newblock {\em PloS one\/}~{\em 10\/}(3), e0119491--.

\bibitem[\protect\citeauthoryear{Breheny and Huang}{Breheny and
  Huang}{2011}]{BrehenyPatrick2011CDAF}
Breheny, P. and J.~Huang (2011).
\newblock Coordinate descent algorithms for nonconvex penalized regression,
  with applications to biological feature selection.
\newblock {\em The annals of applied statistics\/}~{\em 5\/}(1), 232--253.

\bibitem[\protect\citeauthoryear{Chen and Chen}{Chen and
  Chen}{2008}]{ChenJiahua2008EBic}
Chen, J. and Z.~Chen (2008).
\newblock Extended bayesian information criteria for model selection with large
  model spaces.
\newblock {\em Biometrika\/}~{\em 95\/}(3), 759--771.

\bibitem[\protect\citeauthoryear{Chen and Chen}{Chen and
  Chen}{2012}]{ChenJiahua2012EBFS}
Chen, J. and Z.~Chen (2012).
\newblock Extended bic for small-n-large-p sparse glm.
\newblock {\em Statistica Sinica\/}~{\em 22\/}(2), 555--574.

\bibitem[\protect\citeauthoryear{Cox and Reid}{Cox and
  Reid}{2004}]{CoxD.R.2004Anop}
Cox, D.~R. and N.~Reid (2004).
\newblock A note on pseudolikelihood constructed from marginal densities.
\newblock {\em Biometrika\/}~{\em 91\/}(3), 729--737.

\bibitem[\protect\citeauthoryear{Danaher, Wang, and Witten}{Danaher
  et~al.}{2014}]{DanaherPatrick2014jglf}
Danaher, P., P.~Wang, and D.~M. Witten (2014).
\newblock Joint graphical lasso for inverse covariance estimation across
  multiple classes.
\newblock {\em Journal of the Royal Statistical Society. Series B, Statistical
  methodology\/}~{\em 76\/}(2), 373--397.

\bibitem[\protect\citeauthoryear{Dobra, Lenkoski, and Rodriguez}{Dobra
  et~al.}{2011}]{DobraAdrian2011BIfG}
Dobra, A., A.~Lenkoski, and A.~Rodriguez (2011).
\newblock Bayesian inference for general gaussian graphical models with
  application to multivariate lattice data.
\newblock {\em Journal of the American Statistical Association\/}~{\em
  106\/}(496), 1418--1433.

\bibitem[\protect\citeauthoryear{Drton and Maathuis}{Drton and
  Maathuis}{2017}]{DrtonMathias2017SLiG}
Drton, M. and M.~H. Maathuis (2017).
\newblock Structure learning in graphical modeling.
\newblock {\em Annual review of statistics and its application\/}~{\em 4\/}(1),
  365--393.

\bibitem[\protect\citeauthoryear{Du, Kibbe, and Lin}{Du
  et~al.}{2008}]{DuPan2008lapf}
Du, P., W.~A. Kibbe, and S.~M. Lin (2008).
\newblock lumi: a pipeline for processing illumina microarray.
\newblock {\em Bioinformatics\/}~{\em 24\/}(13), 1547--1548.

\bibitem[\protect\citeauthoryear{Fan, Feng, and Wu}{Fan
  et~al.}{2009}]{Fan_2009}
Fan, J., Y.~Feng, and Y.~Wu (2009, jun).
\newblock Network exploration via the adaptive {LASSO} and {SCAD} penalties.
\newblock {\em The Annals of Applied Statistics\/}~{\em 3\/}(2).

\bibitem[\protect\citeauthoryear{Fearnhead and Donnelly}{Fearnhead and
  Donnelly}{2002}]{FearnheadPaul2002Almf}
Fearnhead, P. and P.~Donnelly (2002).
\newblock Approximate likelihood methods for estimating local recombination
  rates.
\newblock {\em Journal of the Royal Statistical Society. Series B, Statistical
  methodology\/}~{\em 64\/}(4), 657--680.

\bibitem[\protect\citeauthoryear{Foygel and Drton}{Foygel and
  Drton}{2010}]{Foygel2010}
Foygel, R. and M.~Drton (2010, 11).
\newblock Extended $\text{B}$ayesian information criteria for $\text{G}$aussian
  graphical models.
\newblock {\em Advances in Neural Information Processing Systems\/}~{\em 23},
  20200--22028.

\bibitem[\protect\citeauthoryear{Franco, Bucasas, Wells, Niño, Wang, Zapata,
  Arden, Renwick, Yu, Quarles, Bray, Couch, Belmont, and Shaw}{Franco
  et~al.}{2013}]{FrancoLuisM.2013Igao}
Franco, L.~M., K.~L. Bucasas, J.~M. Wells, D.~Niño, X.~Wang, G.~E. Zapata,
  N.~Arden, A.~Renwick, P.~Yu, J.~M. Quarles, M.~S. Bray, R.~B. Couch, J.~W.
  Belmont, and C.~A. Shaw (2013).
\newblock Integrative genomic analysis of the human immune response to
  influenza vaccination.
\newblock {\em eLife\/}~{\em 2013\/}(2), e00299--e00299.

\bibitem[\protect\citeauthoryear{Friedman, Hastie, Höfling, and
  Tibshirani}{Friedman et~al.}{2007}]{FriedmanJerome2007PCO}
Friedman, J., T.~Hastie, H.~Höfling, and R.~Tibshirani (2007).
\newblock Pathwise coordinate optimization.
\newblock {\em The annals of applied statistics\/}~{\em 1\/}(2), 302--332.

\bibitem[\protect\citeauthoryear{Gao and Massam}{Gao and
  Massam}{2015}]{GaoXin2015EoSG}
Gao, X. and H.~Massam (2015).
\newblock Estimation of symmetry-constrained $\text{G}$aussian graphical
  models: Application to clustered dense networks.
\newblock {\em Journal of computational and graphical statistics\/}~{\em
  24\/}(4), 909--929.

\bibitem[\protect\citeauthoryear{Gao and Song}{Gao and
  Song}{2010}]{GaoXin2010CLBI}
Gao, X. and P.~X.-K. Song (2010).
\newblock Composite likelihood bayesian information criteria for model
  selection in high-dimensional data.
\newblock {\em Journal of the American Statistical Association\/}~{\em
  105\/}(492), 1531--1540.

\bibitem[\protect\citeauthoryear{Gibberd and Nelson}{Gibberd and
  Nelson}{2017}]{GibberdAlexanderJ.2017REoP}
Gibberd, A.~J. and J.~D.~B. Nelson (2017).
\newblock Regularized estimation of piecewise constant $\text{G}$aussian
  graphical models: The group-fused graphical lasso.
\newblock {\em Journal of computational and graphical statistics\/}~{\em
  26\/}(3), 623--634.

\bibitem[\protect\citeauthoryear{Graczyk, Ishi, Kołodziejek, and
  Massam}{Graczyk et~al.}{2022}]{GraczykPiotr2022Msit}
Graczyk, P., H.~Ishi, B.~Kołodziejek, and H.~Massam (2022).
\newblock Model selection in the space of gaussian models invariant by
  symmetry.
\newblock {\em The Annals of statistics\/}~{\em 50\/}(3), 1747--.

\bibitem[\protect\citeauthoryear{Guo, Levina, Michailidis, and Zhu}{Guo
  et~al.}{2011}]{GUOJIAN2011Jeom}
Guo, J., E.~Levina, G.~Michailidis, and J.~Zhu (2011).
\newblock Joint estimation of multiple graphical models.
\newblock {\em Biometrika\/}~{\em 98\/}(1), 1--15.

\bibitem[\protect\citeauthoryear{Higuera, Gardiner, and Cios}{Higuera
  et~al.}{2015}]{HigueraClara2015Sfmi}
Higuera, C., K.~J. Gardiner, and K.~J. Cios (2015).
\newblock Self-organizing feature maps identify proteins critical to learning
  in a mouse model of down syndrome.
\newblock {\em PloS one\/}~{\em 10\/}(6), e0129126.

\bibitem[\protect\citeauthoryear{Højsgaard, Edwards, and Lauritzen}{Højsgaard
  et~al.}{2012}]{alma991036188596505164}
Højsgaard, S., D.~Edwards, and S.~Lauritzen (2012).
\newblock {\em Graphical Models with R\/} (1st ed. 2012. ed.).
\newblock Use R! New York, NY: Springer New York.

\bibitem[\protect\citeauthoryear{Højsgaard and Lauritzen}{Højsgaard and
  Lauritzen}{2007}]{H2007IigG}
Højsgaard, S. and S.~L. Lauritzen (2007).
\newblock Inference in graphical gaussian models with edge and vertex
  symmetries with the grc package for r.
\newblock {\em Journal of statistical software\/}~{\em 23\/}(6), 1--26.

\bibitem[\protect\citeauthoryear{Jensen, Johansen, and Lauritzen}{Jensen
  et~al.}{1991}]{JensenSorenTolver1991GCAf}
Jensen, S.~T., S.~Johansen, and S.~L. Lauritzen (1991).
\newblock Globally convergent algorithms for maximizing likelihood function.
\newblock {\em Biometrika\/}~{\em 78\/}(4), 867--877.

\bibitem[\protect\citeauthoryear{Kolar, Parikh, and Xing}{Kolar
  et~al.}{2010}]{MladenKolar2010Osnc}
Kolar, M., A.~P. Parikh, and E.~P. Xing (2010).
\newblock On sparse nonparametric conditional covariance selection.
\newblock In {\em ICML 2010 - Proceedings, 27th International Conference on
  Machine Learning}, pp.\  559--566.

\bibitem[\protect\citeauthoryear{Kolar, Song, Ahmed, and Xing}{Kolar
  et~al.}{2010}]{KolarMladen2010ETN}
Kolar, M., L.~Song, A.~Ahmed, and E.~P. Xing (2010).
\newblock Estimating time-varying networks.
\newblock {\em The annals of applied statistics\/}~{\em 4\/}(1), 94--123.

\bibitem[\protect\citeauthoryear{Kolar and Xing}{Kolar and
  Xing}{2012}]{KolarMladen2012Enwj}
Kolar, M. and E.~P. Xing (2012).
\newblock Estimating networks with jumps.
\newblock {\em Electronic journal of statistics\/}~{\em 6\/}(none), 2069--2106.

\bibitem[\protect\citeauthoryear{Larribe and Fearnhead}{Larribe and
  Fearnhead}{2011}]{LarribeF.2011OCLI}
Larribe, F. and P.~Fearnhead (2011).
\newblock On composite likelihoods in statistical genetics.
\newblock {\em Statistica Sinica\/}~{\em 21\/}(1), 43--69.

\bibitem[\protect\citeauthoryear{Lauritzen, Uhler, and Zwiernik}{Lauritzen
  et~al.}{2019}]{PiotrZwiernik2}
Lauritzen, S., C.~Uhler, and P.~Zwiernik (2019).
\newblock {Maximum likelihood estimation in Gaussian models under total
  positivity}.
\newblock {\em The Annals of Statistics\/}~{\em 47\/}(4), 1835 -- 1863.

\bibitem[\protect\citeauthoryear{Lauritzen}{Lauritzen}{1996}]{alma991011687879705164}
Lauritzen, S.~L. (1996).
\newblock {\em Graphical models}.
\newblock Oxford statistical science series ; 17. Oxford: Clarendon Press.

\bibitem[\protect\citeauthoryear{Lin, Du, Huber, and Kibbe}{Lin
  et~al.}{2008}]{LinSimonM.2008Mvtf}
Lin, S.~M., P.~Du, W.~Huber, and W.~A. Kibbe (2008).
\newblock Model-based variance-stabilizing transformation for illumina
  microarray data.
\newblock {\em Nucleic acids research\/}~{\em 36\/}(2), e11--e11.

\bibitem[\protect\citeauthoryear{Lindsay}{Lindsay}{1988}]{Lindsay}
Lindsay, B. (1988, 01).
\newblock Composite likelihood.
\newblock {\em Contemporary Mathematics\/}~{\em 80}, 221--239.

\bibitem[\protect\citeauthoryear{Maathuis, Drton, Lauritzen, and
  Wainwright}{Maathuis et~al.}{2018}]{alma991029818979705164}
Maathuis, M., M.~Drton, S.~Lauritzen, and M.~Wainwright (2018).
\newblock {\em Handbook of graphical models}.
\newblock Chapman \& Hall/CRC handbooks of modern statistical methods. Boca
  Raton, FL: CRC Press.

\bibitem[\protect\citeauthoryear{Malakooti, Pritchard, Chen, Yu, Sgambelloni,
  Adlard, and Finkelstein}{Malakooti et~al.}{2020}]{MalakootiNakisa2020TLIo}
Malakooti, N., M.~A. Pritchard, F.~Chen, Y.~Yu, C.~Sgambelloni, P.~A. Adlard,
  and D.~I. Finkelstein (2020).
\newblock The long isoform of intersectin-1 has a role in learning and memory.
\newblock {\em Frontiers in behavioral neuroscience\/}~{\em 14}, 24--24.

\bibitem[\protect\citeauthoryear{Mazumder, Friedman, and Hastie}{Mazumder
  et~al.}{2011}]{MazumderRahul2011SCDW}
Mazumder, R., J.~H. Friedman, and T.~Hastie (2011).
\newblock Sparsenet: Coordinate descent with nonconvex penalties.
\newblock {\em Journal of the American Statistical Association\/}~{\em
  106\/}(495), 1125--1138.

\bibitem[\protect\citeauthoryear{Meinshausen and Bühlmann}{Meinshausen and
  Bühlmann}{2006}]{MeinshausenNicolai2006HGaV}
Meinshausen, N. and P.~Bühlmann (2006).
\newblock High-dimensional graphs and variable selection with the lasso.
\newblock {\em The Annals of statistics\/}~{\em 34\/}(3), 1436--1462.

\bibitem[\protect\citeauthoryear{Mohammadi and Wit}{Mohammadi and
  Wit}{2015}]{10.1214/14-BA889}
Mohammadi, A. and E.~C. Wit (2015).
\newblock {Bayesian Structure Learning in Sparse Gaussian Graphical Models}.
\newblock {\em Bayesian Analysis\/}~{\em 10\/}(1), 109 -- 138.

\bibitem[\protect\citeauthoryear{Monti, Hellyer, Sharp, Leech, Anagnostopoulos,
  and Montana}{Monti et~al.}{2014}]{MontiRicardoPio2014Etbc}
Monti, R.~P., P.~Hellyer, D.~Sharp, R.~Leech, C.~Anagnostopoulos, and
  G.~Montana (2014).
\newblock Estimating time-varying brain connectivity networks from functional
  mri time series.
\newblock {\em NeuroImage (Orlando, Fla.)\/}~{\em 103}, 427--443.

\bibitem[\protect\citeauthoryear{Ni, Stingo, and Baladandayuthapani}{Ni
  et~al.}{2019}]{NiYang2019BGR}
Ni, Y., F.~C. Stingo, and V.~Baladandayuthapani (2019).
\newblock Bayesian graphical regression.
\newblock {\em Journal of the American Statistical Association\/}~{\em
  114\/}(525), 184--197.

\bibitem[\protect\citeauthoryear{Ni, Stingo, and Baladandayuthapani}{Ni
  et~al.}{2022}]{JMLR:v23:21-0102}
Ni, Y., F.~C. Stingo, and V.~Baladandayuthapani (2022).
\newblock Bayesian covariate-dependent gaussian graphical models with varying
  structure.
\newblock {\em Journal of Machine Learning Research\/}~{\em 23\/}(242), 1--29.

\bibitem[\protect\citeauthoryear{Nica and Dermitzakis}{Nica and
  Dermitzakis}{2013}]{NicaAlexandraC.2013Eqtl}
Nica, A.~C. and E.~T. Dermitzakis (2013).
\newblock Expression quantitative trait loci: present and future.
\newblock {\em Philosophical Transactions of the Royal Society B: Biological
  Sciences\/}~{\em 368\/}(1620), 20120362--20120362.

\bibitem[\protect\citeauthoryear{Niu, Ni, Pati, and Mallick}{Niu
  et~al.}{2023}]{Niu2}
Niu, Y., Y.~Ni, D.~Pati, and B.~Mallick (2023, 07).
\newblock Covariate-assisted bayesian graph learning for heterogeneous data.
\newblock {\em Journal of the American Statistical Association\/}, 1--25.

\bibitem[\protect\citeauthoryear{Nocedal and Wright}{Nocedal and
  Wright}{2006}]{NocedalJorge2006NO}
Nocedal, J. and S.~Wright (2006).
\newblock {\em Numerical Optimization\/} (Second Edition ed.).
\newblock Springer Series in Operations Research and Financial Engineering. New
  York, NY: Springer Nature.

\bibitem[\protect\citeauthoryear{Raskutti, Yu, Wainwright, and
  Ravikumar}{Raskutti et~al.}{2008}]{NIPS2008_61f2585b}
Raskutti, G., B.~Yu, M.~J. Wainwright, and P.~Ravikumar (2008).
\newblock Model selection in gaussian graphical models: High-dimensional
  consistency of $\ell_1$-regularized $\text{MLE}$.
\newblock ~{\em 21}.

\bibitem[\protect\citeauthoryear{Ravikumar, Wainwright, Raskutti, and
  Yu}{Ravikumar et~al.}{2011}]{RavikumarPradeep2011Hceb}
Ravikumar, P., M.~J. Wainwright, G.~Raskutti, and B.~Yu (2011).
\newblock High-dimensional covariance estimation by minimizing
  $\ell_1$-penalized log-determinant divergence.
\newblock {\em Electronic journal of statistics\/}~{\em 5}, 935--980.

\bibitem[\protect\citeauthoryear{Ribatet, Cooley, and Davison}{Ribatet
  et~al.}{2012}]{RibatetMathieu2012BIFC}
Ribatet, M., D.~Cooley, and A.~C. Davison (2012).
\newblock Bayesian inference from composite likelihoods, with an application to
  spatial extremes.
\newblock {\em Statistica Sinica\/}~{\em 22\/}(2), 813--845.

\bibitem[\protect\citeauthoryear{Srivastava, Błazejewska, Heßmann, Bruder,
  Geffers, Manuel, Gruber, and Schughart}{Srivastava
  et~al.}{2009}]{SrivastavaBarkha2009Hgbs}
Srivastava, B., P.~Błazejewska, M.~Heßmann, D.~Bruder, R.~Geffers, S.~Manuel,
  A.~D. Gruber, and K.~Schughart (2009).
\newblock Host genetic background strongly influences the response to influenza
  a virus infections.
\newblock {\em PloS one\/}~{\em 4\/}(3), e4857--e4857.

\bibitem[\protect\citeauthoryear{Trammell and Toth}{Trammell and
  Toth}{2008}]{TrammellRitaA2008Gsar}
Trammell, R.~A. and L.~A. Toth (2008).
\newblock Genetic susceptibility and resistance to influenza infection and
  disease in humans and mice.
\newblock {\em Expert review of molecular diagnostics\/}~{\em 8\/}(4),
  515--529.

\bibitem[\protect\citeauthoryear{Tseng}{Tseng}{2001}]{TsengP.2001Coab}
Tseng, P. (2001).
\newblock Convergence of a block coordinate descent method for
  nondifferentiable minimization.
\newblock {\em Journal of optimization theory and applications\/}~{\em
  109\/}(3), 475--494.

\bibitem[\protect\citeauthoryear{Varin}{Varin}{2008}]{VarinCristiano2008Ocml}
Varin, C. (2008).
\newblock On composite marginal likelihoods.
\newblock {\em Advances in statistical analysis : AStA : a journal of the
  German Statistical Society\/}~{\em 92\/}(1), 1--28.

\bibitem[\protect\citeauthoryear{Wang and Kolar}{Wang and
  Kolar}{2014}]{wang2014inference}
Wang, J. and M.~Kolar (2014).
\newblock Inference for sparse conditional precision matrices. arxiv preprint
  arxiv:1412.7638.

\bibitem[\protect\citeauthoryear{Weighill, Ben~Guebila, Glass, Quackenbush, and
  Platig}{Weighill et~al.}{2022}]{WeighillDeborah2022Pggr}
Weighill, D., M.~Ben~Guebila, K.~Glass, J.~Quackenbush, and J.~Platig (2022).
\newblock Predicting genotype-specific gene regulatory networks.
\newblock {\em Genome research\/}~{\em 32\/}(3), 524--533.

\bibitem[\protect\citeauthoryear{Yang and Peng}{Yang and
  Peng}{2020}]{YangJilei2020ETGM}
Yang, J. and J.~Peng (2020).
\newblock Estimating time-varying graphical models.
\newblock {\em Journal of computational and graphical statistics\/}~{\em
  29\/}(1), 191--202.

\bibitem[\protect\citeauthoryear{Yin, Geng, Li, and Wang}{Yin
  et~al.}{2010}]{YinJianxin2010NCM}
Yin, J., Z.~Geng, R.~Li, and H.~Wang (2010).
\newblock Nonparametric covariance model.
\newblock {\em Statistica Sinica\/}~{\em 20\/}(1), 469--479.

\bibitem[\protect\citeauthoryear{Zhang and Li}{Zhang and
  Li}{2022}]{ZhangJingfei2022HGGR}
Zhang, J. and Y.~Li (2022).
\newblock High-dimensional $\text{G}$aussian graphical regression models with
  covariates.
\newblock {\em Journal of the American Statistical Association\/}~{\em
  ahead-of-print\/}(ahead-of-print), 1--13.

\bibitem[\protect\citeauthoryear{Zhang, Liu, Skogerbø, Zhu, Lu, Chen, Shi,
  Zhang, Wang, Wu, and Chen}{Zhang et~al.}{2006}]{ZhangZhihua2006Dcis}
Zhang, Z., C.~Liu, G.~Skogerbø, X.~Zhu, H.~Lu, L.~Chen, B.~Shi, Y.~Zhang,
  J.~Wang, T.~Wu, and R.~Chen (2006).
\newblock Dynamic changes in subgraph preference profiles of crucial
  transcription factors.
\newblock {\em PLoS computational biology\/}~{\em 2\/}(5), 383--391.

\bibitem[\protect\citeauthoryear{Zhao and Yu}{Zhao and
  Yu}{2006}]{ZhaoPeng2006Omsc}
Zhao, P. and B.~Yu (2006).
\newblock On model selection consistency of lasso.
\newblock {\em Journal of machine learning research\/}~{\em 7}, 2541--2563.

\bibitem[\protect\citeauthoryear{Zhou, Lafferty, and Wasserman}{Zhou
  et~al.}{2010}]{ZhouShuheng2010Tvug}
Zhou, S., J.~Lafferty, and L.~Wasserman (2010).
\newblock Time varying undirected graphs.
\newblock {\em Machine learning\/}~{\em 80\/}(2-3), 295--319.

\bibitem[\protect\citeauthoryear{Zwiernik, Uhler, and Richards}{Zwiernik
  et~al.}{2017}]{PiotrZwiernik1}
Zwiernik, P., C.~Uhler, and D.~Richards (2017, 11).
\newblock {Maximum Likelihood Estimation for Linear Gaussian Covariance
  Models}.
\newblock {\em Journal of the Royal Statistical Society Series B: Statistical
  Methodology\/}~{\em 79\/}(4), 1269--1292.

\end{thebibliography}

\clearpage

\begin{titlepage}
    \centering
    {\LARGE \bf Supplementary Materials for ``High-Dimensional Covariate-Dependent Gaussian Graphical Models" \par}
    \vspace{1cm}
    {\large JIACHENG WANG \\
    Department of Mathematics and Statistics, York University, \\
    4700 Keele Street, Toronto M3J 1P3, Canada. \\
    and \\
    XIN GAO$^{\ast}$ \\
    Department of Mathematics and Statistics, York University, \\
    4700 Keele Street, Toronto M3J 1P3, Canada. \\
    Email: xingao@yorku.ca
    \par}
    \vspace{1cm}
    {\large \today \par}
\end{titlepage}

\renewcommand{\thesection}{S\arabic{section}}
\renewcommand{\thesubsection}{S\arabic{section}.\arabic{subsection}}


\setcounter{section}{0}

\section{Notations Used in Section 3.2}
\label{Notations}
In \textbf{Section 3.2}, the notations $\mathcal{I}_{1.1}, \mathcal{I}_{1.2}, \cdots, \mathcal{I}_{2.3}$ are defined as
\begin{align*}
\mathcal{I}_{1.1} &= 
\frac{1}{n}
\sum_{m=1}^{n}
\Big[ 
\alpha_{aa} + \sum_{h=1}^{H}x_m^{(h)} 
[\theta_{aa}]_h 
\Big]^{-1}
\Big(
\sum_{i=1, i \neq a}^{p_n}
\sum_{h=1}^{H}
x_m^{(h)}[\theta_{ai}]_h
y_{mi}
\Big)
y_{mb}
\\
& +
\frac{1}{n}
\sum_{m=1}^{n}
\Big[ 
\alpha_{bb} + \sum_{h=1}^{H}x_m^{(h)} 
[\theta_{bb}]_h 
\Big]^{-1}
\Big(
\sum_{i=1, i \neq b}^{p_n}
\sum_{h=1}^{H}
x_m^{(h)}[\theta_{bi}]_h
y_{mi}
\Big)
y_{ma},
\\
\mathcal{I}_{1.2} &=
\frac{1}{n}
\sum_{m=1}^{n}
\Big[ 
\alpha_{aa} + \sum_{h=1}^{H}x_m^{(h)} 
[\theta_{aa}]_h 
\Big]^{-1}
\Big(
\sum_{
\substack{i=1, i\neq a, \\ i \neq b}}^{p_n}
\alpha_{ai}
y_{mi}
\Big)
y_{mb} 
\\
& +
\frac{1}{n}
\sum_{m=1}^{n}
\Big[ 
\alpha_{bb} + \sum_{h=1}^{H}x_m^{(h)} 
[\theta_{bb}]_h 
\Big]^{-1}
\Big(
\sum_{
\substack{i=1, i\neq b, \\ i \neq a}}^{p_n}
\alpha_{bi}
y_{mi}
\Big)
y_{ma},
\\
\mathcal{I}_{1.3} &=
\frac{1}{n}
\sum_{m=1}^{n}
\Big[ 
\alpha_{aa} + \sum_{h=1}^{H}x_m^{(h)} 
[\theta_{aa}]_h 
\Big]^{-1} y_{mb}^2
+ 
\frac{1}{n}
\sum_{m=1}^{n}
\Big[ 
\alpha_{bb} + \sum_{h=1}^{H}x_m^{(h)} 
[\theta_{bb}]_h 
\Big]^{-1} y_{ma}^2,
\\
\mathcal{I}_{2.1} &=
\frac{1}{n}
\sum_{m=1}^{n}
\Big[ 
\alpha_{aa} + \sum_{h=1}^{H}x_m^{(h)} 
[\theta_{aa}]_h 
\Big]^{-1}
\Big(
\sum_{
\substack{i=1, i\neq a}}^{p_n}
\alpha_{ai}
y_{mi}
\Big)
x_m^{(h_1)}y_{mb}
\\
& +
\frac{1}{n}
\sum_{m=1}^{n}
\Big[ 
\alpha_{bb} + \sum_{h=1}^{H}x_m^{(h)} 
[\theta_{bb}]_h 
\Big]^{-1}
\Big(
\sum_{
\substack{i=1, i\neq b}}^{p_n}
\alpha_{bi}
y_{mi}
\Big)
x_m^{(h_1)}y_{ma},
\\
\mathcal{I}_{2.2} &=
\frac{1}{n}
\sum_{m=1}^{n}
\Big[ 
\alpha_{aa} + \sum_{h=1}^{H}x_m^{(h)} 
[\theta_{aa}]_h 
\Big]^{-1}
\Big(
\sum_{
\substack{i=1, \\ i\neq a,}}^{p_n}
\sum_{\substack{h=1, \\
(i, h) \neq (b, h_1)}}^{H}
x_m^{(h)}
[\theta_{ai}]_{h}
y_{mi}
\Big)
x_m^{(h_1)}y_{mb}
\\
& +
\frac{1}{n}
\sum_{m=1}^{n}
\Big[ 
\alpha_{bb} + \sum_{h=1}^{H}x_m^{(h)} 
[\theta_{bb}]_h 
\Big]^{-1}
\Big(
\sum_{
\substack{i=1, \\ i\neq b,}}^{p_n}
\sum_{\substack{h=1, \\
(i, h) \neq (a, h_1)}}^{H}
x_m^{(h)}
[\theta_{bi}]_{h}
y_{mi}
\Big)
x_m^{(h_1)}y_{ma},
\\
\mathcal{I}_{2.3} &=
\frac{1}{n}
\sum_{m=1}^{n}
\Big[ 
\alpha_{aa} + \sum_{h=1}^{H}x_m^{(h)} 
[\theta_{aa}]_h 
\Big]^{-1} [x_m^{(h_1)}]^2y_{mb}^2
+
\frac{1}{n}
\sum_{m=1}^{n}
\Big[ 
\alpha_{bb} + \sum_{h=1}^{H}x_m^{(h)} 
[\theta_{bb}]_h 
\Big]^{-1} [x_m^{(h_1)}]^2y_{ma}^2.
\end{align*}

\section{Special Case of Updating Equations}
\label{Supplementary Section 1}

Following \textbf{Section 3.2}, consider a special case where the conditional variances remain constant across all covariates. For $m=1,\cdots,n$ and $ j=1,\cdots, p_n,$ the joint negative composite log-likelihood function is simplified as
\begin{align*}
-l_c(\bm{\beta}) & = 
-\sum_{j=1}^{p_n}
\sum_{m=1}^{n}
\log{\alpha_{jj}}
+
\sum_{j=1}^{p_n}
\sum_{m=1}^{p_n}
\alpha_{jj}
\bigg\{
y_{mj} + 
\alpha_{jj}^{-1}
\sum_{i=1, i\neq j}^{p_n}
\Big[
\alpha_{ji}
+ \sum_{h=1}^{H}
x_m^{(h)}[\theta_{ji}]_{h}
\Big]
y_{mi}
\bigg\}^{2}.
\end{align*}
Let $\mathcal{E} = \{ (i,j): \; i < j; \; i,j = 1,\cdots, p_n\}$ denote the ordered off-diagonal index set. For $\forall\; s,w \in \mathcal{E}$ and $\forall \; j=1,\cdots, p_n, h=1,\cdots,H$, the updating equations for the diagonal parameters can be obtained by taking the first derivatives of $Q(\bm{\beta})$, as presented below
\small
\begin{align*}
[\hat{\theta}_{s}]_h &= 
\frac{S
\bigg(
-tr\Big(\bm{T}^{s}\bm{C}_{\theta}^{(h)} \Big)
+
tr\Big(\bm{T}^{s}\big(\bm{T}^{s^c} \odot \big[\bm{B}_{\theta}^{(h)}\big]^T\big)\bm{D}_{\theta}^{(h)}\Big)
+
tr\Big(\bm{T}^{s}\big(\bm{T}^{s^1}
\odot
\bm{B}_{\alpha}^T\big)\bm{C}_{\theta}^{(h)}\Big)
+ 
\Delta,
\lambda
\bigg)}
{tr\Big(\bm{T}^{s}\bm{\Sigma}  \bm{T}^{s} \bm{D}_{\theta}^{(h)}\Big)},
 \\
    \hat{\alpha}_{w} & = 
\frac{S
\bigg(
-tr\Big(\bm{T}^{w}\bm{C}_{\alpha}\Big) 
+ 
tr\Big(\bm{T}^{w}\big(\bm{T}^{w^c} \odot  \bm{B}_{\alpha}^T\big)\bm{C}_{\alpha}\Big)
+
\sum_{h=1}^{H}
tr\Big(\bm{T}^{w}\big(\bm{T}^{w^1}\odot \big[\bm{B}_{\theta}^{(h)}\big]^T\big)\bm{C}_{\theta}^{(h)}\Big)
, 
\lambda
\bigg)}
{tr\Big(\bm{T}^{w}\bm{\Sigma} \bm{T}^{w} \bm{C}_{\alpha}\Big)} ,
\\
    \hat{\alpha}_{jj} &=  
\frac{2\sum_{m=1}^{n}
\Big(
\sum_{i=1,i\neq j}^{p_n}
\left[
\alpha_{ji} +
\sum_{h=1}^{H}
x_m^{(h)}[\theta_{ji}]_{h}
\right]y_{mi}
\Big)^2}
{
-n+ \sqrt{
n^2+4
\Big(
\sum_{m=1}^{n}y_{mj}^{2}
\Big)
\bigg\{
\sum_{m=1}^{n}
\Big(
\sum_{i=1,i\neq j}^{p_n}
\left[
\alpha_{ji} +
\sum_{h=1}^{H} 
x_m^{(h)} [\theta_{ji}]_{h}
\right] y_{mi}
\Big)^2
\bigg\}}},
\end{align*}
\normalsize 
where $s^c$ and $w^c$ denote the complement of $s$ and $w$ respectively, and 
\begin{align*}
& \bm{C}_{\theta}^{(h)}=
\frac{1}{n}
\Big(\bm{X}^{(h)}\odot \bm{Y}\Big)^T \bm{Y} ,
&& 
\bm{D}_{\theta}^{(h)}=
\frac{1}{n}\Big(\bm{X}^{(h)} \odot \bm{Y}\Big)^T \Big(\bm{X}^{(h)}\odot \bm{Y}\Big) ,
\\
&	\bm{G}_{\theta}^{(h,l)} = 
\frac{1}{n}
\Big(\bm{X}^{(l)} \odot \bm{Y}\Big)^T
\Big(\bm{X}^{(h)}\odot \bm{Y}\Big) ,
 && 	
\bm{C}_{\alpha}=\frac{1}{n}\bm{Y}^T\bm{Y} ,
 \\
& \Delta = 
 \sum_{l=1, l\neq h}^{H}
tr \Big(\bm{T}^{s}\big(\bm{T}^{s^1}
\odot
\big[\bm{B}_{\theta}^{(l)}\big]^T\big)\bm{G}_{\theta}^{(h,l)}\Big).
\end{align*}
The $n \times p_n$ matrix $\bm{X}^{(h)}$ is defined as 
\begin{align}
\bm{X}^{(h)} = 
\begin{bmatrix}
    x_{1}^{(h)}       & x_{1}^{(h)} & \dots & x_{1}^{(h)} 
    \\
    x_{2}^{(h)}       & x_{2}^{(h)} & \dots & x_{2}^{(h)} \\
    \vdots & \vdots & \ddots & \vdots \\
    x_{n}^{(h)}       & x_{n}^{(h)} & \cdots & x_{n}^{(h)}
\end{bmatrix}_{n \times p_n} .
\end{align}
Let $\bm{\Sigma}$  denote a $p_n\times p_n$ diagonal matrix $ \text{diag}(\alpha_{11}^{-1},...,\alpha_{p_np_n}^{-1})$, and let $\odot$ denote the component-wise product. Each entry of $\bm{B}_{\theta}^{(h)}$ and $\bm{B}_{\alpha}$ takes the following form
\begin{align}       
\big(\bm{B}_{\theta}^{(h)}\big)_{ij}& = 
\left\{                
\begin{array}{ll}   
-[\theta_{ij}]_h\alpha_{jj}^{-1} \quad & {i \neq j;\quad i,j = 1,...,p_n},
\\
0 & {i=j,}
\end{array}\right. 
\\
\big(\bm{B}_{\alpha}\big)_{ij}& = \left\{
\begin{array}{ll}   
-
\alpha_{ij}\alpha_{jj}^{-1} \quad & {i \neq j;\quad i,j = 1,...,p_n,}\\
    	0 & {i=j,}
    \end{array}\right.
\end{align}
Let $\bm{T}^{s}$ and $\bm{T}^{w}$ denote the edge adjacency matrix, e.g. $(\bm{T}^{s})_{ij}=1$ if $(i,j) = s$ or $(j,i) = s$, and $(\bm{T}^{s})_{ij}=0$ otherwise. Let  $\bm{T}^{s^1}$ be defined as: if $(i,j) = s$, all of the elements in the $ith$ and $jth$ row of the matrix $\bm{T}^{s^1}$ are equal to 1 except the diagonal entries, and all other entries are equal to zero. The matrix $\bm{T}^{w^1}$ is defined similarly.

\section{Sensitivity Analysis with Real Dataset}
We perform a sensitivity analysis on the influenza vaccination dataset to investigate the impact of $\gamma$ on the estimation results. Specifically, we conduct the analysis across 5 levels of $\gamma$, namely $\gamma = 0, 0.25, 0.50, 0.75, 1$, with the tuning parameter $\lambda$ determined based on minimum EBIC \citep{ChenJiahua2008EBic, Foygel2010, GaoXin2010CLBI}. 
Table~\ref{Supplementary: Table 1} summarizes the number of estimated nonzero parameters under each scenario and Figure~\ref{Supplementary: Figure 1} provides the sparsity structures across 4 levels of $\gamma$. Note that the estimation results are identical for $\gamma = 0$ and $\gamma = 0.25$ (e.g. the optimal $\lambda$ selected by EBIC is the same for both cases). Therefore, we just present the sparsity structure for $\gamma = 0.25$, omitting $\gamma = 0$ for simplicity. Table~\ref{Supplementary: Table 1} and Figure~\ref{Supplementary: Figure 1} indicate that the number of estimated nonzero parameters declines, leading to greater sparsity in the graph structure as $\gamma$ increases from 0 to 1. This result reflects the expected property of EBIC as reported
in \citet{ChenJiahua2008EBic, GaoXin2010CLBI}. The following \eqref{Supplementary: formula 1} gives the EBIC for Gaussian graphical models,
\begin{align} \label{Supplementary: formula 1}
EBIC_{\gamma} 
&= 
- 2l_c(\hat{\bm{\beta}_c}) + df\log{n}
+ 4 df \gamma \log{p_n}.
\end{align}
In this formulation, the first term corresponds to minus twice the composite log-likelihood, capturing the goodness-of-fit for the current model. The second component penalizes model complexity, while the third term promotes sparsity.  More specifically, the parameter $\gamma \in [0, 1]$ determines the emphasis on sparsity, with higher values leading to smaller, simpler models. Previous studies \citep{ChenJiahua2008EBic, ChenJiahua2012EBFS} have demonstrated that the EBIC sacrifices a small degree of sensitivity but effectively manages the false positive instances. This feature makes it especially valuable for variable selection tasks involving moderate sample sizes with a high-dimensional covariate space, as often encountered in genomic studies. The consistency of the composite likelihood-based EBIC in selecting the true underlying model has been established by \citet{GaoXin2010CLBI}. In general, as $\gamma$ becomes larger, some moderate and weak interactions are filtered out, whereas strong interactions between genes can still be captured.

\begin{table}
\caption{Sensitivity analysis regarding $\gamma$ applied to the influenza dataset}  
\begin{tabular}{>{\centering\arraybackslash}p{2.3cm}>{\centering\arraybackslash}p{2.3cm}>{\centering\arraybackslash}p{2.3cm}>{\centering\arraybackslash}p{2.3cm}>
{\centering\arraybackslash}p{2.3cm}>
{\centering\arraybackslash}p{2.3cm}
}
\toprule
Matrix & $\gamma = 0$ 
&  $\gamma = 0.25$
&  $\gamma = 0.50$
&  $\gamma = 0.75$
&  $\gamma = 1$\\
\midrule
$Q_0$  & 604 & 604  & 500  & 380 & 306  \\
\addlinespace
$Q_1$  & 404 & 404  & 312  & 239 & 178\\
\bottomrule
\end{tabular}
\label{Supplementary: Table 1}
\begin{tablenotes}
\item This table shows the number of nonzero estimated parameters across different levels of $\gamma$ in EBIC.  
\end{tablenotes}
\end{table}

\begin{figure}
\begin{center}
\centerline{\includegraphics[width=15cm, height=7.5cm]{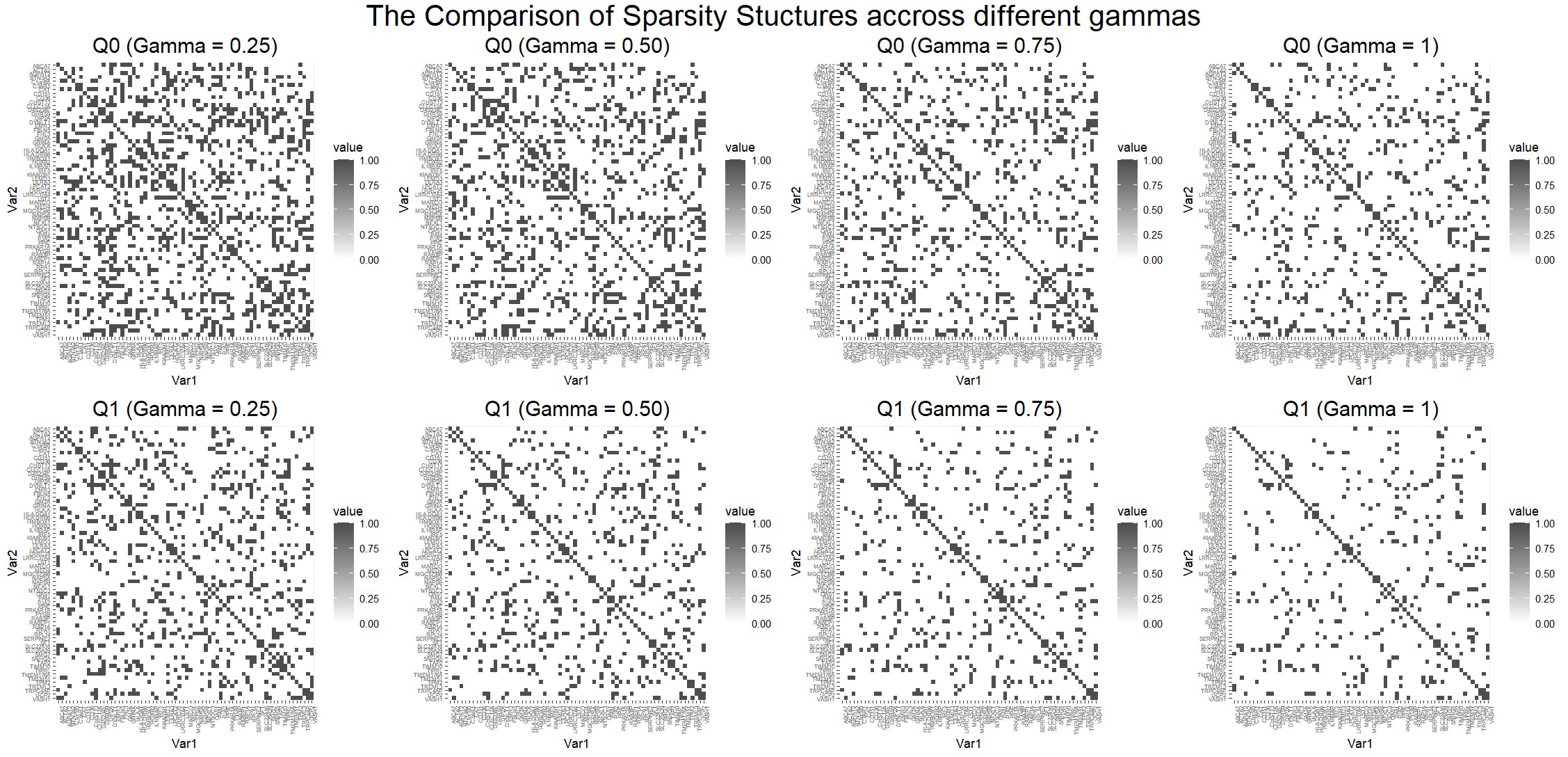}}
\end{center}
\caption{The comparison of estimated sparsity patterns across 4 levels of $\gamma$. Each subplot has 68 rows and 68 columns, representing the estimated elements of $\bm{Q}_0, \bm{Q}_1$ respectively.  The corresponding cell is marked in black if the estimated entry is nonzero. Otherwise, it is white.}
\label{Supplementary: Figure 1}
\end{figure}

\section{Simulation}

\subsection{Simulation Procedure for Matrices Discussed in Section 4.3}

The simulation procedure for $\bm{Q}_0, \bm{P}_1, \bm{P}_2$ referenced in Section 4.3 of our manuscript is provided below.  
\begin{enumerate}
  \item Generate $\bm Q_0, \bm Q_1, \bm Q_2$ as random sparse positive definite matrices as outlined in the previous simulation. Let 
  $(\bm Q_0)_{ij} = \alpha_{ij},$ 
  $(\bm Q_1)_{ij} = \eta_{ij},$ 
  and $(\bm Q_2)_{ij} = \pi_{ij} $ represent the $(i,j)th$ element of $\bm Q_0,$ $\bm Q_1,$ and $\bm Q_2$ respectively.
  \item Define 
  $\bm N_1^{\ast} = 
  \Lambda_1 \bm Q_1 \Lambda_1
  $,  
  $ \bm N_1^{\ast\ast} = 
  \Lambda_0 \bm N_1^{\ast} 
  \Lambda_0
  $,
  and 
  $\bm N_2^{\ast} = 
  \Lambda_2 \bm Q_2 \Lambda_2
  $, 
  $ \bm N_2^{\ast\ast} = 
  \Lambda_0 \bm N_2^{\ast} 
  \Lambda_0
  $,
  where 
  $
  \Lambda_0 = diag\Big(
  \sqrt{
  \frac{1}{2}\alpha_{11}
  },\cdots,
  \sqrt{
  \frac{1}{2}
  \alpha_{p_n p_n}}
  \Big)
  $,
  $
  \Lambda_1 =
  diag\Big(\sqrt{\eta_{11}^{-1}}, 
  \cdots,
  \sqrt{\eta_{p_n p_n}^{-1}}
  \Big),
  $
  $
  \Lambda_2 =
  diag\Big(
  \sqrt{\pi_{11}^{-1}}, 
  \cdots,
  \sqrt{\pi_{p_n p_n}^{-1}}
  \Big)
  $.
  \item Generate $\bm P_1, \bm P_2$ by 
  $
  \bm P_1 = 
  \bm N_1^{\ast \ast} - 
  \frac{1}{2} \bm Q_0, \;
  \bm P_2 = 
  \bm N_2^{\ast \ast} -
  \frac{1}{2} \bm Q_0
  $.
\end{enumerate}

\subsection{Simulation Results for $p_n=100$}
\begin{table}
\caption{\label{Supplementary: Table 2}Simulation results of penalized composite likelihood estimator for high-dimensional covariate dependent GGMs: $p_n = 100$}
\centering 
\begin{threeparttable}
\begin{tabular}{>{\centering\arraybackslash}p{1.0cm}>{\centering\arraybackslash}p{1.5cm}>{\centering\arraybackslash}p{1.5cm}>{\centering\arraybackslash}p{2.4cm}>{\centering\arraybackslash}p{2.5cm}>
{\centering\arraybackslash}p{2.5cm}>
{\centering\arraybackslash}p{2.5cm}} \toprule
$p_n$ & Matrix &  \# of Para. &
  \# of Nonzero Para. &
  Sensitivity & Specificity & MCC   \\
    \midrule
    \multirow{2}{*}{$100$} & 
     $\bm{Q}_0$ & 5050 & 165 &
    0.7403(0.1248) & 0.8701(0.0537) & 0.2084(0.0281) \\
    \addlinespace
    & $\bm{Q}_1$ & 5050  & 163 & 
         0.5473(0.1593) & 0.8856(0.0308) & 0.1484(0.0429)  \\
    \bottomrule  
\end{tabular} 
\begin{tablenotes} \item In this experiment, the sample size is $3000$ and the standard errors for each evaluation metric are shown in parentheses. 
\end{tablenotes}
\end{threeparttable}
\end{table}

We increase the dimension $p_n$ to 100 to see the performance of our proposed method. The simulation results over 100 repetitions are displayed in  Table~\ref{Supplementary: Table 2}. To reduce the computational workload, we simulate the special case where the diagonal elements of the precision matrix are constant across all values of the covariate. Despite the challenges of high-dimensional data, the proposed method still demonstrates acceptable performance.

\section{Technical Lemmas}
Without loss of generality, the proofs and lemmas presented in this section are derived based on one covariate for simplicity. Nonetheless, extending these arguments to scenarios with multiple covariates can be justified in the same way. We begin with some notations shown below.
\begin{definition}
\label{Definition 3.1} 
A random variable $Y$ with mean $\mu = E\{Y\}$ is sub-exponential if there are non-negative parameters $(v^2, c)$ such that 
\begin{align*}
E\Big\{\exp\big\{\lambda(Y-\mu)\big\}
\Big\} \leq 
\exp\Big\{ 
\frac{v^2\lambda^2}{2}
\Big\} \;\;\;\;\; 
\text{for all } \abs{\lambda} < \frac{1}{c}.
\end{align*}
\end{definition}
\justify \textbf{Notation 3.1}: If the random vector $\bm{Y}_m=(Y_{m1},\dots,Y_{mp_n})^T,$ $m=1,\dots,n,$ follows  multivariate normal distribution $N(\bm{0},\bm{\Sigma}_m)$, then we have 
\begin{align*}
Y_{mj}^2  \sim \; \text{Sub-Exponential} \; (v_{m,jj}^2, c_{m,jj}),
\;\;\;
Y_{mi}Y_{mj} \sim \; \text{Sub-Exponential} \; (v_{m,ij}^2, c_{m,ij}) ,
\end{align*}
where $ m = 1,...,n, i = 1,...,p_n, j = 1,...,p_n$ and $i \neq j.$  

\justify \textbf{Notation 3.2}: The proofs are based on the following parameterization $
\bm \Sigma_{m}^{-1} = \bm Q_0 + x_m (\bm Q_1 - \bm Q_0) = x_m \bm Q_1 + (1-x_m) \bm Q_0
$,
where
$(\bm Q_0)_{ij}=\alpha_{ij}, (\bm Q_1)_{ij}=\theta_{ij}$. However, one may replace $x_m \theta_{ij} + (1-x_m) \alpha_{ij}$ with $\alpha_{ij} + \sum_{h=1}^{H} x_m^{(h)} [\theta_{ij}]_h$ for multiple covariates scenarios, where $i,j=1,\cdots,p_n$. Note that the notation $\theta_{ij}$ is used under one-covariate scenario to denote $(\bm Q_1)_{ij}$, whereas $[\theta_{ij}]_h$ is used to denote $(\bm P_h)_{ij}$ for multiple covariate situation. 
\justify 
\textbf{Notation 3.3}: We denote
$\abs{\mathcal{S}_1}, \abs{\mathcal{S}_2} and \abs{\bm{\mathcal{S}}}$ as the cardinality of $\mathcal{S}_1, \mathcal{S}_2, $ and $\bm{\mathcal{S}}$ respectively. 
In particular, 
$
|\bm{\mathcal{S}}| =
\text{total \# of true nonzero parameters} = q_n + p_n(H+1).
$
Furthermore, we write 
$-l_c(\bm{Y}, \bm \beta) = -l_c(\bm \beta)$ when there is no ambiguity. 
\justify 
\textbf{Notation 3.4}: We define a vector $\bm v = (v_1,\cdots, v_d) \in \mathbb{R}^{d}$ and write its $\ell_q$ norm as $\lVert \bm v \rVert_q = \big(\sum_{i=1}^{d}\abs{v_i}^q\big)^{\frac{1}{q}}$. In particular, the $\ell_{\infty}$ norm of a vector is defined as $\lVert \bm v \rVert_{\infty} = sup_{1 \leq i \leq d}\abs{v_i}$. 
\justify 
\textbf{Notation 3.5}: For any matrix $\bm A \in \mathbb{R}^{d_1 \times d_2}$, we write its spectral norm, infinity norm, max norm, and Frobenius norm as $\vertiii{\bm A}_2 = \Lambda_{\max} (\bm A)$, $\vertiii{\bm A}_{\infty} = \max_{1\leq i \leq d_1} \sum_{j=1}^{d_2} \abs{a_{ij}}
$, $\vertiii{\bm A}_{\max} = 
\max_{ij}\abs{a_{ij}}$, $\vertiii{\bm A}_{F} = \sqrt{\sum_{i=1}^{d_1}\sum_{j=1}^{d_2} \abs{a_{ij}}^2}$, where $\Lambda_{\max}(\cdot)$ denotes the largest singular value of a matrix and $a_{ij}$ is the $(i,j)$th element of $\bm A$. Moreover, we let $\Lambda_{\min}(\bm A), \lambda_{\min}(\bm A), \lambda_{\max}(\bm A)$ represent the smallest singular value, and smallest and largest eigenvalues of $\bm{A}$ respectively.

\begin{lemma} \label{lemma 2.1}
If the random vector $\bm{Y}=(Y_1,\dots,Y_{p_n})^T$ follows  a multivariate normal distribution $N(\bm{0},\bm{\Sigma})$, 
with $\Sigma_{jj} =\sigma_{jj}^2< \infty$, $j = 1,...,p_n$, 
then there exists some constant $0 < \mathcal{M}_1 < \infty$ such that
\begin{align*}
&
\abs{E\{Y_{i}^{h_1} Y_{j}^{h_2} \}} 
< \mathcal{M}_1,
\quad
\abs{E\{Y_{i}^{h_1} Y_{j}^{h_2} Y_{l}^{h_3} \} 
} 
< \mathcal{M}_1,
\quad
\abs{E\{Y_{i}^{h_1} Y_{j}^{h_2} Y_{k}^{h_3} Y_{l}^{h_4} \}}  
< \mathcal{M}_1, 
\end{align*}
where $h_1, h_2, h_3, h_4$ are finite positive integers and $i,j,k,l = 1,...p_n.$ 
\end{lemma}

\begin{proof}
It is a well-established fact that the central absolute moments for all even orders of normal random variables with mean zero are given by $
E\{ \abs{Y_j}^{r}\} = \sigma_{jj}^{r} (r-1)!!,
$
where 
$
(r-1)!!$ denotes the double factorial, that is, $(r-1)!! = (r-1)\cdot (r-3)\cdots 1$ and $r$ is some even number. By Jensen's inequality and H\"older's inequality, $\forall \; i,j,k,l = 1,...p,$ we have 
\begin{align*}
& \vert E\{Y_{i}^{h_1}Y_{j}^{h_2} \} \vert \leq E\{\vert Y_{i}^{h_1}Y_{j}^{h_2} \vert \} \leq 
\Big\{ E\{ \vert  Y_{i}^{h_1} \vert ^2 \} \Big\} ^{\frac{1}{2}} \Big\{E\{ \vert  Y_{j}^{h_2} \vert ^2 \} \Big\}^{\frac{1}{2}} 
\leq \mathcal{M}_1,
\\
&
\abs{E\{Y_{i}^{h_1} Y_{j}^{h_2} Y_{l}^{h_3} \}    
}  
\leq 
E\{
\abs{Y_i^{h_1} Y_j^{h_2} Y_l^{h_3}}
\}
\leq
\Big\{ E\{ \vert  Y_{i}^{h_1} Y_{j}^{h_2} \vert ^2 \} \Big\} ^{\frac{1}{2}} \Big\{ E\{ \vert  Y_{l}^{h_3} \vert ^2 \} \Big\} ^{\frac{1}{2}} 
\leq \mathcal{M}_1,
\\
&   \abs{E\{Y_{i}^{h_1} Y_{j}^{h_2} Y_{k}^{h_3} Y_{l}^{h_4} \}}  \leq 
E\{
\abs{Y_{i}^{h_1} Y_{j}^{h_2} Y_{k}^{h_3} Y_{l}^{h_4}}
\}
\leq 
\Big\{ E\{
\vert Y_{i}^{h_1} Y_{j}^{h_2} \vert ^2 
\}
\Big\}^{\frac{1}{2}}
\Big\{ E\{
\vert Y_{k}^{h_3} Y_{l}^{h_4} \vert ^2 
\}
\Big\}^{\frac{1}{2}}
< \mathcal{M}_1.
\end{align*}
\end{proof}

\begin{lemma} \label{lemma 2.2}
For $\forall \;
s,w = \{(a,b)\} \in \mathcal{E}
$ and $j=1,\cdots, p_n$, 
\begin{align*}
\max_{s}
\abs{
\frac{\partial - {l}_c(\bm{\beta^0})}{\partial \theta_{s}}
} &=  \mathcal{O}_p
\bigg\{
(ns_n^2\log{p_n})^{\frac{1}{2}} 
\bigg\} ,
&
\max_{w}
\abs{
\frac{\partial - {l}_c(\bm{\beta^0})}{\partial \alpha_{w}}
} &=  \mathcal{O}_p
\bigg\{
(ns_n^2\log{p_n})^{\frac{1}{2}} 
\bigg\} ,
\\
\max_{1\leq j \leq p_n}
\abs{\frac{\partial - {l}_c(\bm{\beta^0})}{\partial \theta_{jj}}} &=  \mathcal{O}_p
\bigg\{
(ns_n^4\log{p_n})^{\frac{1}{2}}
\bigg\} ,
&
\max_{1\leq j \leq p_n}
\abs{\frac{\partial - {l}_c(\bm{\beta^0})}{\partial \alpha_{jj}}} &=  \mathcal{O}_p
\bigg\{
(ns_n^4\log{p_n})^{\frac{1}{2}}
\bigg\}.
\end{align*}
Furthermore, for $\forall \; \epsilon > 0,$ there exist some positive constants $\mathcal{K}_1, C_1, c_{z,max}$ and $ c_{\mathcal{I}_2, max}$ such that
\begin{align*}
P\bigg\{
\Big\|
\frac{1}{n} l_c^{(1)} (\bm \beta^0)_{\mathcal{S}^c}
\Big\|_{\infty}
\geq
\epsilon
\bigg\}
&
\leq
8\exp\Bigg\{
-\min \bigg\{
\frac{C_1 n \epsilon^2}{72s_n^2\mathcal{K}_1},
\frac{n\epsilon}{12s_n c_{z,max}}
\bigg\}
+ \log{s_n}
+ \log{(p_n^2 - p_n - q_n)}
\Bigg\}
\\
&
+
2\exp\Bigg\{
-\min \bigg\{
\frac{C_1 n \epsilon^2}{18\mathcal{K}_1},
\frac{n\epsilon}{6c_{\mathcal{I}_2, max}}
\bigg\}
+ \log{(p_n^2 - p_n - q_n)}
\Bigg\},
\\
P\bigg\{
\Big\|
\frac{1}{n} l_c^{(1)} (\bm \beta^0)_{\mathcal{S}_1}
\Big\|_{\infty}
\geq
\epsilon
\bigg\}
&
\leq
8\exp\Bigg\{
-\min \bigg\{
\frac{C_1n\epsilon^2}{72s_n^2\mathcal{K}_1},
\frac{n\epsilon}{12s_nc_{z,max}}
\bigg\}
+ \log{s_n}
+ \log{q_n}
\Bigg\}
\\
&
+
2\exp\Bigg\{
-\min \bigg\{
\frac{C_1n\epsilon^2}{18\mathcal{K}_1},
\frac{n\epsilon}{6c_{\mathcal{I}_2, max}}
\bigg\}
+ \log{q_n}
\Bigg\} .
\end{align*}
\end{lemma}
\begin{proof}
We start with $\frac{\partial - {l}_c(\bm{\beta^0})}{\partial \theta_{s}}$ as follows. For $\forall \; s=\{(a, b) \} \in \mathcal{E}$, we show that 
\begin{align*}
\frac{\partial - {l}_c(\bm{\beta^0})}{\partial \theta_{s}}
& =
\underbrace{
\sum_{m=1}^{n}
\left[
x_m \theta_{aa} + (1-x_m)\alpha_{aa}
\right]^{-1}
\bigg(
\sum_{i=1,i\neq a}^{p_n}
\left[
x_m\theta_{ai} + (1-x_m)\alpha_{ai}
\right]
y_{mi}
\bigg)
x_m y_{mb}
}_{\mathcal{I}_{1.1}} 
\\
& +
\underbrace{
\sum_{m=1}^{n}
\left[
x_m \theta_{bb} + (1-x_m)\alpha_{bb}
\right]^{-1}
\bigg(
\sum_{i=1,i\neq b}^{p_n}
\left[
x_m\theta_{bi} + (1-x_m)\alpha_{bi}
\right]
y_{mi}
\bigg)
x_m y_{ma}
}_{\mathcal{I}_{1.2}}
\\
& +
\underbrace{
2\sum_{m=1}^{n}
x_m y_{ma} y_{mb}
}_{\mathcal{I}_2}.
\end{align*}
The first term $\mathcal{I}_{1.1}$ can be rewritten as
\begin{align*}
\mathcal{I}_{1.1} 
&=
\sum_{m=1}^{n}
\left[
x_m \theta_{aa} + (1-x_m)\alpha_{aa}
\right]^{-1}
\bigg(
\sum_{i=1,i\neq a}^{p_n}
x_m\theta_{ai}
y_{mi}
\bigg)
x_m y_{mb} \\
&
+
\sum_{m=1}^{n}
\left[
x_m \theta_{aa} + (1-x_m)\alpha_{aa}
\right]^{-1}
\bigg(
\sum_{i=1,i\neq a}^{p_n}
(1-x_m)\alpha_{ai}
y_{mi}
\bigg)
x_m y_{mb}.
\end{align*}
It can be shown that
\begin{align*}
&
\sum_{m=1}^{n}
\left[
x_m \theta_{aa} + (1-x_m)\alpha_{aa}
\right]^{-1}
\bigg(
\sum_{i=1,i\neq a}^{p_n}
x_m\theta_{ai}
y_{mi}
\bigg)
x_m y_{mb}
\\
=&
\sum_{i=1, i\neq a}^{p_n}
\underbrace{
\bigg(
\sum_{m=1}^{n}
x_m^2
\left[
x_m \theta_{aa} + (1-x_m)\alpha_{aa}
\right]^{-1}
\theta_{ai} y_{mi} y_{mb}
\bigg)
}_{Z_{i,ab}^{(1)}}
=
\sum_{i=1, i\neq a}^{p_n}
Z_{i,ab}^{(1)},
\\
&
\sum_{m=1}^{n}
\left[
x_m \theta_{aa} + (1-x_m)\alpha_{aa}
\right]^{-1}
\bigg(
\sum_{i=1,i\neq a}^{p_n}
(1-x_m)\alpha_{ai}
y_{mi}
\bigg)
x_m y_{mb}
\\
=&
\sum_{i=1, i\neq a}^{p_n}
\underbrace{
\bigg(
\sum_{m=1}^{n}
x_m(1-x_m)
\left[
x_m \theta_{aa} + (1-x_m)\alpha_{aa}
\right]^{-1}
\alpha_{ai} y_{mi} y_{mb}
\bigg)
}_{Z_{i,ab}^{(2)}}
=
\sum_{i=1, i\neq a}^{p_n}
Z_{i,ab}^{(2)},
\end{align*}
where each $Z_{i,ab}^{(1)}$ and $Z_{i,ab}^{(2)}$ follows sub-exponential distribution with parameters $\Big(\big[v_{z_{i,ab}}^{(1)}\big]^2, c_{z_{i,ab}}^{(1)}\Big)$ and $
\Big(\big[v_{z_{i,ab}}^{(2)}\big]^2, c_{z_{i,ab}}^{(2)}\Big),$ respectively. The parameters take the following form
\begin{align*}
&
\big[v_{z_{i,ab}}^{(1)}\big]^2 =
\sum_{m=1}^{n}x_m^4
\left[
x_m \theta_{aa} + (1-x_m)\alpha_{aa}
\right]^{-2}
\theta_{ai}^2
v_{m, ib}^2,
\\
&
c_{z_{i,ab}}^{(1)} =
\max_{1\leq m \leq n}
\Big\{
x_m^2
\left[
x_m \theta_{aa} + (1-x_m)\alpha_{aa}
\right]^{-1}
\abs{\theta_{ai}}
c_{m, ib}
\Big\},
\\
&
\big[v_{z_{i,ab}}^{(2)}\big]^2 = \sum_{m=1}^{n}x_m^2
(1-x_m)^2 
\left[
x_m \theta_{aa} + (1-x_m)\alpha_{aa}
\right]^{-2}
\alpha_{ai}^2
v_{m, ib}^2,
\\
&
c_{z_{i,ab}}^{(2)} =
\max_{1\leq m \leq n}
\Big\{x_m (1-x_m)
\left[
x_m \theta_{aa} + (1-x_m)\alpha_{aa}
\right]^{-1}
\abs{\alpha_{ai}}
c_{m, ib}
\Big\}.
\end{align*}
According to the property of sub-exponential distribution, we have 
\begin{align*}
P\Big\{\abs{Z_{i,ab}^{(1)} - E\big\{Z_{i,ab}^{(1)}\big\}} \geq t \Big\} & \leq 
\begin{cases}
2\exp\Big\{
-\frac{t^2}{2\big[v_{z_{i,ab}}^{(1)}\big]^2}
\Big\} 
& \text{if}\; 0 \leq t \leq 
\frac{\big[v_{z_{i,ab}}^{(1)}\big]^2}{c_{z_{i,ab}}^{(1)}}, 
\\ 
2\exp\Big\{
-\frac{t}{2c_{z_{i,ab}}^{(1)}}
\Big\} & \text{if}\; 
t > \frac{\big[v_{z_{i,ab}}^{(1)}\big]^2}{c_{z_{i,ab}}^{(1)}} .
\end{cases}
\end{align*}
Using Bonferroni inequalities, we show that
\begin{align*}
& P\Big\{\abs{\mathcal{I}_{1.1} - E\big\{
\mathcal{I}_{1.1}
\big\}} \geq t \Big\} 
=
P\Bigg\{
\abs{
\sum_{i=1,i\neq a}^{p_n}
\left[
Z_{i,ab}^{(1)} - E\big\{Z_{i,ab}^{(1)}\big\}
\right]
+  
\sum_{i=1,i\neq a}^{p_n}
\left[
Z_{i,ab}^{(2)} - E\big\{Z_{i,ab}^{(2)}\big\}
\right]
}
\geq t
\Bigg\} 
\\
&
\leq 
P\Bigg\{
\abs{
\sum_{i=1,i\neq a}^{p_n}
\left[
Z_{i,ab}^{(1)} - E\big\{Z_{i,ab}^{(1)}\big\}
\right]
}
\geq \frac{t}{2}
\Bigg\} 
+
P\Bigg\{
\abs{
\sum_{i=1,i\neq a}^{p_n}
\left[
Z_{i,ab}^{(2)} - E\big\{Z_{i,ab}^{(2)}\big\}
\right]
}
\geq \frac{t}{2}
\Bigg\} 
\\
&
\leq 
\sum_{i=1,i\neq a}^{p_n}
P\Big\{
\abs{Z_{i,ab}^{(1)} - E\big\{Z_{i,ab}^{(1)}\big\}}
\geq 
\frac{t}{2s_n}
\Big\}
+ 
\sum_{i=1,i\neq a}^{p_n}
P\Big\{
\abs{Z_{i,ab}^{(2)} - E\big\{Z_{i,ab}^{(2)}\big\}}
\geq 
\frac{t}{2s_n}
\Big\}
\\
& 
\leq 
4\exp\Bigg\{
-\min \bigg\{
\frac{t^2}{8s_n^2v_{z \backslash a,max}^2},
\frac{t}{4s_n c_{z \backslash a,max}}
\bigg\}
+ \log{s_n}
\Bigg\} ,
\end{align*}
where 
\begin{align*}
&
v_{z \backslash a,max}^2 = \max_{\substack{ 1\leq i \leq p_n, \\ i \neq a}}\Big\{ \big[ v_{z_{i,ab}}^{(1)}\big]^2,
\big[ v_{z_{i,ab}}^{(2)}\big]^2
\Big\}, &
&
c_{z \backslash a,max} = \max_{\substack{ 1\leq i \leq p_n, \\ i \neq a}} \Big\{
c_{z_{i, ab}}^{(1)},
c_{z_{i, ab}}^{(2)}
\Big\}.
\end{align*}
Likewise, for the second term $\mathcal{I}_{1.2}$, we can show 
\begin{align*}
\mathcal{I}_{1.2}
&=
\sum_{m=1}^{n}
\left[
x_m \theta_{bb} + (1-x_m)\alpha_{bb}
\right]^{-1}
\bigg(
\sum_{i=1,i\neq b}^{p_n}
x_m\theta_{bi}
y_{mi}
\bigg)
x_m y_{ma} \\
&
+
\sum_{m=1}^{n}
\left[
x_m \theta_{bb} + (1-x_m)\alpha_{bb}
\right]^{-1}
\bigg(
\sum_{i=1,i\neq b}^{p_n}
(1-x_m)\alpha_{bi}
y_{mi}
\bigg)
x_m y_{ma}
\\
&
=
\sum_{i=1,i\neq b}^{p_n}
Z_{i, ba}^{(1)}
+ \sum_{i=1,i\neq b}^{p_n}
Z_{i, ba}^{(2)},
\end{align*}
where each $Z_{i, ba}^{(1)}$ and $ Z_{i, ba}^{(2)}$ follows sub-exponential distribution with parameters $\Big(\big[v_{z_{i,ba}}^{(1)}\big]^2, c_{z_{i,ba}}^{(1)}\Big),$ and $
\Big(\big[v_{z_{i,ba}}^{(2)}\big]^2, c_{z_{i,ba}}^{(2)}\Big)
,$ respectively. The parameters take the following form
\begin{align*}
&
\big[v_{z_{i,ba}}^{(1)}\big]^2 =
\sum_{m=1}^{n}x_m^4
\left[
x_m \theta_{bb} + (1-x_m)\alpha_{bb}
\right]^{-2}
\theta_{bi}^2
v_{m, ia}^2,
\\
&
c_{z_{i,ba}}^{(1)} =
\max_{1\leq m \leq n}
\Big\{
x_m^2
\left[
x_m \theta_{bb} + (1-x_m)\alpha_{bb}
\right]^{-1}
\abs{\theta_{bi}}
c_{m, ia}
\Big\},
\\
&
\big[v_{z_{i,ba}}^{(2)}\big]^2 = \sum_{m=1}^{n}x_m^2
(1-x_m)^2 
\left[
x_m \theta_{bb} + (1-x_m)\alpha_{bb}
\right]^{-2}
\alpha_{bi}^2
v_{m, ia}^2,
\\
&
c_{z_{i,ba}}^{(2)} =
\max_{1\leq m \leq n}
\Big\{x_m (1-x_m)
\left[
x_m \theta_{bb} + (1-x_m)\alpha_{bb}
\right]^{-1}
\abs{\alpha_{bi}}
c_{m, ia}
\Big\}.
\end{align*}
Therefore, it leads to similar result
\begin{align*}
& P\Big\{\abs{\mathcal{I}_{1.2} - E\big\{
\mathcal{I}_{1.2}
\big\}} \geq t \Big\} 
=
P\Bigg\{
\abs{
\sum_{i=1,i\neq b}^{p_n}
\left[
Z_{i,ba}^{(1)} - E\big\{Z_{i,ba}^{(1)}\big\}
\right]
+  
\sum_{i=1,i\neq b}^{p_n}
\left[
Z_{i,ba}^{(2)} - E\big\{Z_{i,ba}^{(2)}\big\}
\right]
}
\geq t
\Bigg\} 
\\
&
\leq 
P\Bigg\{
\abs{
\sum_{i=1,i\neq b}^{p_n}
\left[
Z_{i,ba}^{(1)} - E\big\{Z_{i,ba}^{(1)}\big\}
\right]
}
\geq \frac{t}{2}
\Bigg\} 
+
P\Bigg\{
\abs{
\sum_{i=1,i\neq b}^{p_n}
\left[
Z_{i,ba}^{(2)} - E\big\{Z_{i,ba}^{(2)}\big\}
\right]
}
\geq \frac{t}{2}
\Bigg\} 
\\
& 
\leq 
4\exp\Bigg\{
-\min \bigg\{
\frac{t^2}{8s_n^2v_{z \backslash b,max}^2},
\frac{t}{4s_n c_{z\backslash b,max}}
\bigg\}
+ \log{s_n}
\Bigg\} ,
\end{align*}
where 
\begin{align*}
&
v_{z\backslash b,max}^2 = \max_{\substack{ 1\leq i \leq p_n, \\ i \neq b}}\Big\{ \big[ v_{z_{i,ba}}^{(1)}\big]^2,
\big[ v_{z_{i,ba}}^{(2)}\big]^2
\Big\}, &
&
c_{z \backslash b,max} = \max_{\substack{ 1\leq i \leq p_n, \\ i \neq b}} \Big\{
c_{z_{i, ba}}^{(1)},
c_{z_{i, ba}}^{(2)}
\Big\}.
\end{align*}
Lastly, the third term 
$
\mathcal{I}_2 
=
2\sum_{m=1}^{n}x_m y_{ma}y_{mb}
$
follows a sub-exponential distribution with parameters
\begin{align*}
v_{\mathcal{I}_{2,ab}}^2 = 
4\sum_{m=1}^{n}
x_m^2v_{m,ab}^2 ,
\qquad
c_{\mathcal{I}_{2,ab}} = 
\max_{1\leq m \leq n} 
\big\{ 
2x_m c_{m, ab}
\big\} .
\end{align*}
Hence, for any $t > 0,$
\begin{align*}
P\Big\{
\abs{\mathcal{I}_{2} - E\big\{\mathcal{I}_{2}\big\}} \geq t \Big\} 
\leq 
2\exp\Bigg\{
-\min \bigg\{
\frac{t^2}{2v_{\mathcal{I}_{2,ab}}^2},
\frac{t}{2c_{\mathcal{I}_{2,ab}}}
\bigg\}
\Bigg\} .
\end{align*}
As a result, 
\begin{align*}
& P\Bigg\{
\abs{
\frac{\partial -l_c(\bm{\beta^0})}{\partial \theta_s}
- 
E\bigg\{
\frac{\partial -l_c(\bm{\beta^0})}{\partial \theta_s}
\bigg\}
} 
\geq t 
\Bigg\} 
\\
&
=
P\Big\{
\abs{
\mathcal{I}_{1.1}
- 
E\{
\mathcal{I}_{1.1}
\}
+
\mathcal{I}_{1.2}
- 
E\{
\mathcal{I}_{1.2}
\}
+
\mathcal{I}_{2}
- 
E\{
\mathcal{I}_{2}
\}
} 
\geq t 
\Big\} 
\\
&
\leq
P\Big\{
\abs{
\mathcal{I}_{1.1}
- 
E\{
\mathcal{I}_{1.1}
\}
} \geq 
\frac{t}{3}
\Big\}
+
P\Big\{
\abs{
\mathcal{I}_{1.2}
- 
E\{
\mathcal{I}_{1.2}
\}
} \geq 
\frac{t}{3}
\Big\}
+
P\Big\{
\abs{
\mathcal{I}_{2}
- 
E\{
\mathcal{I}_{2}
\}
} \geq 
\frac{t}{3}
\Big\}
\\
& 
\leq 
8\exp\Bigg\{
-\min \bigg\{
\frac{t^2}{72s_n^2v_{z,max}^2},
\frac{t}{12s_n c_{z,max}}
\bigg\}
+ \log{s_n}
\Bigg\}
+
2\exp\Bigg\{
-\min \bigg\{
\frac{t^2}{18v_{I_{2,ab}}^2},
\frac{t}{6c_{\mathcal{I}_{2,ab}}}
\bigg\}
\Bigg\},
\end{align*}
where $m=1,\cdots,n; i,j,k = 1,\cdots, p_n; i \neq j, j \neq k,$ and 
\begin{align*}
v_{z,max}^2 = \max_{m,i,j,k} \Big\{
&
\sum_{m=1}^{n}
x_m^4 \left[
x_m \theta_{jj}
+ (1-x_m) \alpha_{jj}
\right]^{-2}
\theta_{ji}^2v_{m, ik}^2,
\\
&
\sum_{m=1}^{n}
x_m^2 (1-x_m)^2 \left[
x_m \theta_{jj}
+ (1-x_m) \alpha_{jj}
\right]^{-2} \alpha_{ji}^2
v_{m, ik}^2
\Big\}
\\
c_{z, max} =
\max_{m, i,j,k} \Big\{
&
x_m^2 \left[
x_m \theta_{jj}
+ (1-x_m) \alpha_{jj}
\right]^{-1}
\abs{\theta_{ji}}c_{m, ik},
\\
&
x_m(1-x_m) \left[
x_m \theta_{jj}
+ (1-x_m) \alpha_{jj}
\right]^{-1}
\abs{\alpha_{ji}}c_{m, ik}
\Big\}.
\end{align*}
Define the constants $\mathcal{W}_1, \mathcal{K}_1$ as
\begin{align*}
\mathcal{W}_1 =
\max_{m, i, j, k} \Big\{
& 4 x_m^2 v_{m, jk}^2,
x_m^4 \left[
x_m \theta_{jj}
+ (1-x_m)\alpha_{jj}
\right]^{-2}
\theta_{ji}^2 v_{m, ik}^2,
\\
&
x_m^2 (1-x_m)^2 \left[
x_m \theta_{jj}
+ (1-x_m)\alpha_{jj}
\right]^{-2} \alpha_{ji}^2 v_{m,ik}^2
\Big\},
\\
\mathcal{K}_1 = C_1 \mathcal{W}_1 \; & \; \text{for any constant } C_1 > 0, 
\end{align*}
and let $t = (\mathcal{K}_1 n s_n^2 \log{p_n})^{\frac{1}{2}}$,
the probability becomes
\begin{align*}
&
P\Bigg\{
\abs{
\frac{\partial -l_c(\bm{\beta^0})}{\partial \theta_s}
- 
E\bigg\{
\frac{\partial -l_c(\bm{\beta^0})}{\partial \theta_s}
\bigg\}
} 
\geq (\mathcal{K}_1n s_n^2 \log{p_n})^{\frac{1}{2}} 
\Bigg\} 
\\
&
\leq 
8\exp\Bigg\{
-\min \bigg\{
\frac{ \mathcal{K}_1 n s_n^2 \log{p_n} }{72s_n^2v_{z,max}^2},
\frac{(\mathcal{K}_1 n s_n^2 \log{p_n})^{\frac{1}{2}}}{12s_n c_{z,max}}
\bigg\}
+ \log{s_n}
\Bigg\}
\\
&
+
2\exp\Bigg\{
-\min \bigg\{
\frac{\mathcal{K}_1 n s_n^2 \log{p_n}}{18v_{I_{2,ab}}^2},
\frac{(\mathcal{K}_1 n s_n^2 \log{p_n})^{\frac{1}{2}}}{6c_{\mathcal{I}_{2,ab}}}
\bigg\}
\Bigg\},
\\
& \leq
8 \exp\Bigg\{
-\min \bigg\{
\frac{C_1\log{p_n}}{72},
\frac{(\mathcal{K}_1n\log{p_n})^{\frac{1}{2}}}{12 c_{z,max}}
\bigg\}
+ \log{s_n}
\Bigg\}
\\
&
+ 2 \exp\Bigg\{
-\min \bigg\{
\frac{C_1s_n^2\log{p_n}}{18},
\frac{(\mathcal{K}_1ns_n^2\log{p_n})^{\frac{1}{2}}}{6 c_{\mathcal{I}_{2,ab}}}
\bigg\}
\Bigg\}.
\end{align*}
By the Bonferroni inequality,
\begin{align*}
&
P\Bigg\{
\max_{s}
\abs{
\frac{\partial -l_c(\bm{\beta^0})}{\partial \theta_s}
} 
\geq (\mathcal{K}_1ns_n^2\log{p_n})^{\frac{1}{2}} 
\Bigg\}
\leq 
\frac{ p_n (p_n-1) }{2}
P\Bigg\{
\abs{
\frac{\partial -l_c(\bm{\beta^0})}{\partial \theta_s}
} 
\geq (\mathcal{K}_1ns_n^2\log{p_n})^{\frac{1}{2}} 
\Bigg\}
\\
& 
\leq 
8 \exp\Bigg\{
-\min \bigg\{
\frac{C_1\log{p_n}}{72},
\frac{(\mathcal{K}_1n\log{p_n})^{\frac{1}{2}}}{12 c_{z,max}}
\bigg\}
+ \log{s_n} + \log{p_n} + \log{(p_n-1)}
\Bigg\}
\\
& 
+
\exp\Bigg\{
-\min \bigg\{
\frac{C_1s_n^2\log{p_n}}{18},
\frac{(\mathcal{K}_1ns_n^2\log{p_n})^{\frac{1}{2}}}{6 c_{\mathcal{I}_{2,max}}}
\bigg\}
+ \log{p_n} + \log{(p_n-1)}
\Bigg\},
\end{align*}
where $c_{\mathcal{I}_{2,max}} = \max_{m, a, b} \Big\{
2 x_m c_{m,ab}
\Big\}$. By choosing the constant $\mathcal{C}_1 > 216,$ the upper bound probability approaches to zero as $n$ increases. The proofs of the remaining results are analogous to the one provided above and thus, are omitted.  
\end{proof}

\begin{lemma} \label{lemma 2.3}
All the second order partial derivatives of the negative composite log-likelihood $-l_c(\bm \beta)$ at $\bm{\beta} = \bm{\beta}^0$ are of the order $\mathcal{O}_p(ns_n^2)$.
\end{lemma}
\begin{proof}
For $\forall \; s=\{(a,b)\} \in \mathcal{E}$, $j=1,\cdots,p_n$ and $\bm{\beta} = \bm{\beta}^0$, we show that 
\begin{align*}
\frac{\partial^2 - {l}_c(\bm{\beta})}{\partial \theta_{s}^2}
& =
\sum_{m=1}^{n}
\left[
x_m \theta_{aa} + (1-x_m)\alpha_{aa}
\right]^{-1}
x_m^2 y_{mb}^2
+
\sum_{m=1}^{n}
\left[
x_m \theta_{bb} + (1-x_m)\alpha_{bb}
\right]^{-1}
x_m^2 y_{ma}^2 =
\mathcal{O}_p(n),
\\
\frac{\partial^2 - {l}_c(\bm{\beta})}{\partial \theta_{s} \partial \theta_{jj}} 
&=
-\sum_{m=1}^{n}
x_m
\left[
x_m \theta_{aa} + (1-x_m)\alpha_{aa}
\right]^{-2}
\Big( 
\sum_{i=1,i\neq a}^{p_n}
\left[x_m\theta_{ai}+(1-x_m)\alpha_{ai} \right]
y_{mi}
\Big)
x_my_{mb} 
\\
& =
\mathcal{O}_p(ns_n), \; \text{ when } j = a,
\\
\frac{\partial^2 - {l}_c(\bm{\beta})}{\partial \theta_{s} \partial \theta_{jj}} 
&=
-\sum_{m=1}^{n}
x_m
\left[
x_m \theta_{bb} + (1-x_m)\alpha_{bb}
\right]^{-2}
\Big( 
\sum_{i=1,i\neq b}^{p_n}
\left[x_m\theta_{bi}+(1-x_m)\alpha_{bi} \right]
y_{mi}
\Big)
x_my_{ma} 
\\
& =
\mathcal{O}_p(ns_n),
\; \text{ when } j = b.
\end{align*}
If $j\neq a$ and $j \neq b,$ we have $\frac{\partial^2 - {l}_c(\bm{\beta})}{\partial \theta_{s} \partial \theta_{jj}}=0.$ Lastly,
\begin{align*}
\frac{\partial^2 - {l}_c(\bm{\beta})}{\partial \theta_{jj}^2}
&=
\frac{1}{2}\sum_{m=1}^{n}
x_m^2
\left[
x_m \theta_{jj} + (1-x_m)\alpha_{jj}
\right]^{-2}
\\
& + 
\sum_{m=1}^{n}
x_m^2
\left[
x_m \theta_{jj} + (1-x_m)\alpha_{jj}
\right]^{-3}
\Big( 
\sum_{i=1,i\neq j}^{p_n}
\left[x_m\theta_{ji}+(1-x_m)\alpha_{ji} \right]
y_{mi}
\Big)^2
\\
&= \mathcal{O}_p(ns_n^2).
\end{align*}
The same principle applies to other terms, as illustrated above.
\end{proof}

\begin{lemma} \label{lemma 2.4}
(\textbf{Incoherence Condition}) Under \textbf{Assumptions 3-5} and $q_n^2p_n^{d}\log{p_n} = o(n)$,  the following condition
$$
\vertiii{  
\frac{1}{n}
l_c ^{(2)}(\bm{\beta^0})_{\mathcal{S}^{c}\mathcal{S}_1}
\bigg(
\frac{1}{n}
l_c ^{(2)}(\bm{\beta^0})_{\mathcal{S}_1\mathcal{S}_1}
\bigg)^{-1}
}_{\infty} \leq 1-\frac{1}{3}\xi^2 
$$
holds with probability
tending to one.
\end{lemma} 

\begin{proof}
First we show
\begin{align*}
&\frac{1}{n}
l_c ^{(2)}(\bm{\beta^0})_{\mathcal{S}^{c}\mathcal{S}_1}
\bigg(
\frac{1}{n}
l_c ^{(2)}(\bm{\beta^0})_{\mathcal{S}_1\mathcal{S}_1}
\bigg)^{-1} 
\\
&= \bigg[
H_{\mathcal{S}^{c}\mathcal{S}_1}^{0} 
-
\frac{1}{n}
l_c ^{(2)}(\bm{\beta^0})_{\mathcal{S}^{c}\mathcal{S}_1}
- H_{\mathcal{S}^{c}\mathcal{S}_1}^{0} 
\bigg] 
\cdot
\bigg[
\big(
H_{\mathcal{S}_1\mathcal{S}_1}^{0}\big)^{-1} 
-
\bigg(
\frac{1}{n}
l_c ^{(2)}(\bm{\beta^0})_{\mathcal{S}_1\mathcal{S}_1}
\bigg)^{-1} 
- \big(H_{\mathcal{S}_1\mathcal{S}_1}^{0}\big)^{-1} 
\bigg] 
\\
& 
= 
\underbrace{H_{\mathcal{S}^{c}\mathcal{S}_1}^{0} 
\big(H_{\mathcal{S}_1\mathcal{S}_1}^{0}\big)^{-1}
}_{\mathcal{I}_1} 
+
\underbrace{H_{\mathcal{S}^{c}\mathcal{S}_1}^{0}
\bigg[
\bigg(
-\frac{1}{n}
l_c ^{(2)}(\bm{\beta^0})_{\mathcal{S}_1\mathcal{S}_1}
\bigg)^{-1} 
- \big(H_{\mathcal{S}_1\mathcal{S}_1}^{0}\big)^{-1} 
\bigg]}_{\mathcal{I}_2}
\\
& 
+
\underbrace{
\bigg[
-\frac{1}{n}
l_c ^{(2)}(\bm{\beta^0})_{\mathcal{S}^{c}\mathcal{S}_1}
- H_{\mathcal{S}^{c}\mathcal{S}_1}^{0} 
\bigg]\big(H_{\mathcal{S}_1\mathcal{S}_1}^{0}\big)^{-1}}_{\mathcal{I}_3}
\\
&
+ 
\underbrace{
\bigg[
-\frac{1}{n}
l_c ^{(2)}(\bm{\beta^0})_{\mathcal{S}^{c}\mathcal{S}_1}
- H_{\mathcal{S}^{c}\mathcal{S}_1}^{0} 
\bigg]
\bigg[
\bigg(
-\frac{1}{n}
l_c ^{(2)}(\bm{\beta^0})_{\mathcal{S}_1\mathcal{S}_1}
\bigg)^{-1} 
- \big(H_{\mathcal{S}_1\mathcal{S}_1}^{0}\big)^{-1} 
\bigg]}_{\mathcal{I}_4}.
\end{align*}
\justifying
Therefore, $\vertiii{\frac{1}{n}
l_c ^{(2)}(\bm{\beta^0})_{\mathcal{S}^{c}\mathcal{S}_1}
\big(
\frac{1}{n}
l_c ^{(2)}(\bm{\beta^0})_{\mathcal{S}_1\mathcal{S}_1}
\big)^{-1}
}_{\infty} \leq 
\vertiii{\mathcal{I}_1}_{\infty} +
\vertiii{\mathcal{I}_2}_{\infty} +
\vertiii{\mathcal{I}_3}_{\infty} +
\vertiii{\mathcal{I}_4}_{\infty}
.$ By \textbf{Assumption 5}, $\vertiii{\mathcal{I}_1}_{\infty} \leq 1-\xi.$ The second term can be formulated as
\begin{small}
\begin{align*}
\mathcal{I}_2 & = 
H_{\mathcal{S}^{c}\mathcal{S}_1}^{0}
\bigg[
\bigg(
- \frac{1}{n}
l_c ^{(2)}(\bm{\beta^0})_{\mathcal{S}_1\mathcal{S}_1}
\bigg)^{-1} 
- \big(H_{\mathcal{S}_1\mathcal{S}_1}^{0}\big)^{-1} 
\bigg] 
\\
&
= 
H_{\mathcal{S}^{c}\mathcal{S}_1}^{0}
\big(H_{\mathcal{S}_1\mathcal{S}_1}^{0}\big)^{-1}
\bigg[
H_{\mathcal{S}_1\mathcal{S}_1}^{0}
-
\bigg(
-\frac{1}{n} l_c ^{(2)}(\bm{\beta^0})_{\mathcal{S}_1\mathcal{S}_1}
\bigg)
\bigg]
\bigg(
-
\frac{1}{n} l_c ^{(2)}(\bm{\beta^0})_{\mathcal{S}_1\mathcal{S}_1}
\bigg)^{-1}.
\end{align*}
\end{small}
By \textbf{Lemma} \ref{lemma 2.6}, \textbf{Lemma} \ref{lemma 2.8} and \textbf{Assumption 5}, we show that
\begin{align*}
\vertiii{\mathcal{I}_2}_{\infty} 
&
\leq
\vertiii{
H_{\mathcal{S}^{c}\mathcal{S}_1}^{0}
\big(
H_{\mathcal{S}_1\mathcal{S}_1}^{0}\big)^{-1}
}_{\infty} 
\vertiii{
-\frac{1}{n} l_c ^{(2)}(\bm{\beta^0})_{\mathcal{S}_1\mathcal{S}_1}
- H_{\mathcal{S}_1\mathcal{S}_1}^{0}
}_{\infty}
\vertiii{
\bigg(
-
\frac{1}{n} l_c ^{(2)}(\bm{\beta^0})_{\mathcal{S}_1\mathcal{S}_1}
\bigg)^{-1}
}_{\infty} 
\\
& 
\leq
(1-\xi)  
\vertiii{
-
\frac{1}{n} l_c ^{(2)}(\bm{\beta^0})_{\mathcal{S}_1\mathcal{S}_1}
- H_{\mathcal{S}_1\mathcal{S}_1}^{0}
}_{\infty}
\sqrt{|\mathcal{S}_1|}
\vertiii{
\bigg(
-
\frac{1}{n} l_c ^{(2)}(\bm{\beta^0})_{\mathcal{S}_1\mathcal{S}_1}
\bigg)^{-1}
}_{2} 
\\
&
\leq
(1-\xi)\frac{\tau_{-} \epsilon}{\sqrt{\abs{\mathcal{S}_1}}}
\frac{\sqrt{\abs{\mathcal{S}_1}} }{\tau_{-}}
=
(1-\xi) \epsilon ,
\end{align*}
\justifying with a probability at least
$
1- 8 \exp \Big\{ - \min 
\big\{
\frac{C_3 n (\epsilon^{\prime})^{2} }{128\mathcal{K}_3 q_n^2},
\frac{n \epsilon^{\prime}}{16 q_n c_{\ast}}
\big\}
+ \log{q_n}
\Big\}
- 8  \exp \Big\{ - \min 
\big\{
\frac{C_3 n (\epsilon^{\prime})^{2} }{128 \mathcal{K}_3 q_n^2 s_n^2 }, \\
\frac{n \epsilon^{\prime}}{16 q_n s_n c_{\ast}}
\big\}
+ \log{q_n} + \log{s_n}
\Big\}
- 8 \exp \Big\{ - \min
\big\{
\frac{C_3 n \tau_{-}^2 \epsilon^2 }{128 \mathcal{K}_3 q_n},
\frac{n \tau_{-} \epsilon }{16 \sqrt{q_n} c_{\ast} }
\big\}
+ \log{q_n}
\Big\}
- 8 \exp \Big\{ - \min
\big\{
\frac{C_3 n \tau_{-}^2 \epsilon^2 }{128 \mathcal{K}_3 s_n^2 q_n}, \\
\frac{n \tau_{-} \epsilon }{16 \sqrt{q_n} s_n c_{\ast} }
\big\}
+ \log{q_n} + \log{s_n}
\Big\}. 
$ \\
Likewise, by \textbf{Lemmas} \ref{lemma 2.6}, \ref{lemma 2.7} and \ref{lemma 2.8}, with a probability at least
$
1 - 8 \exp \Big\{ - \min 
\big\{
\frac{C_3 n (\epsilon^{\prime})^{2} }{128\mathcal{K}_3 q_n^2},
\frac{n \epsilon^{\prime}}{16 q_n c_{\ast}}
\big\} \\
+ \log{q_n}
\Big\}
- 8 \exp \Big\{ - \min 
\big\{
\frac{C_3 n (\epsilon^{\prime})^{2} }{128 \mathcal{K}_3 q_n^2 s_n^2 }, 
\frac{n \epsilon^{\prime}}{16 q_n s_n c_{\ast}}
\big\}
+ \log{q_n} + \log{s_n}
\Big\}  
-
8 \exp \Big\{ - \min
\big\{
\frac{C_3 n \tau^2 \epsilon^{2} }{128\mathcal{K}_3 q_n},
\frac{n \tau \epsilon}{16 \sqrt{q_n} c_{\ast}}
\big\} 
+ \log{(p_n^2 - p_n - q_n)}
\Big\}
- 
8 \exp \Big\{ - \min
\big\{
\frac{C_3 n \tau^2 \epsilon^{2} }{128\mathcal{K}_3 q_n s_n^2},
\frac{n \tau \epsilon}{16 \sqrt{q_n} s_n c_{\ast}}
\big\}
+ \log{(p_n^2 - p_n - q_n)} + \log{s_n}
\Big\},
$
the upper bound for $\vertiii{\mathcal{I}_3}_{\infty}$ and $\vertiii{\mathcal{I}_4}_{\infty}$ 
can be formulated as follows
\begin{align*}
\vertiii{\mathcal{I}_3}_{\infty} 
&
\leq
\vertiii{
-
\frac{1}{n} l_c ^{(2)}(\bm{\beta^0})_{\mathcal{S}^{c}\mathcal{S}_1}
- H_{\mathcal{S}^{c}\mathcal{S}_1}^{0}
}_{\infty}
\vertiii{
\big(
H_{\mathcal{S}_1\mathcal{S}_1}^{0}\big)^{-1}
}_{\infty} 
\\
&
\leq 
\vertiii{
-
\frac{1}{n} l_c ^{(2)}(\bm{\beta^0})_{\mathcal{S}^{c}\mathcal{S}_1}
- H_{\mathcal{S}^{c}\mathcal{S}_1}^{0}
}_{\infty}
\sqrt{|\mathcal{S}_1|}
\vertiii{
\big(
H_{\mathcal{S}_1\mathcal{S}_1}^{0}\big)^{-1}
}_{2} 
\leq 
\frac{\tau\epsilon}{\sqrt{\abs{\mathcal{S}_1}}}
\frac{\sqrt{\abs{\mathcal{S}_1}}}{\tau}
=
\epsilon.
\end{align*}
For $\mathcal{I}_4$, we first show that 
\begin{align*}
\bigg( 
-
\frac{1}{n} l_c ^{(2)}(\bm{\beta^0})_{\mathcal{S}_1\mathcal{S}_1} 
\bigg)^{-1}
-
\big(
H^{0}_{\mathcal{S}_1\mathcal{S}_1}\big)^{-1}
=
\big(H_{\mathcal{S}_1\mathcal{S}_1}^{0}\big)^{-1}
\bigg[
H_{\mathcal{S}_1\mathcal{S}_1}^{0}
- \bigg(
-
\frac{1}{n} l_c ^{(2)}(\bm{\beta^0})_{\mathcal{S}_1\mathcal{S}_1} 
\bigg)
\bigg]
\bigg\{
-
\frac{1}{n} l_c ^{(2)}(\bm{\beta^0})_{\mathcal{S}_1\mathcal{S}_1} 
\bigg\}^{-1}.
\end{align*}
The infinity norm becomes
\small
\begin{align*}
& \vertiii{  
\bigg\{ 
-
\frac{1}{n} l_c ^{(2)}(\bm{\beta^0})_{\mathcal{S}_1\mathcal{S}_1} 
\bigg\}^{-1}
-
\big(H_{\mathcal{S}_1\mathcal{S}_1}^{0}\big)^{-1}
}_{\infty} 
=
\vertiii{
\big(H_{\mathcal{S}_1\mathcal{S}_1}^{0}\big)^{-1}
\bigg[
H_{\mathcal{S}_1\mathcal{S}_1}^{0}
+
\frac{1}{n} l_c ^{(2)}(\bm{\beta^0})_{\mathcal{S}_1\mathcal{S}_1} 
\bigg]
\bigg\{
-
\frac{1}{n} l_c ^{(2)}(\bm{\beta^0})_{\mathcal{S}_1\mathcal{S}_1} 
\bigg\}^{-1}
}_{\infty} 
\\
&
\leq 
\sqrt{|\mathcal{S}_1|} \cdot
\vertiii{
\big(H_{\mathcal{S}_1\mathcal{S}_1}^{0}\big)^{-1}
\bigg[
H_{\mathcal{S}_1\mathcal{S}_1}^{0}
-
\bigg(-
\frac{1}{n} l_c ^{(2)}(\bm{\beta^0})_{\mathcal{S}_1\mathcal{S}_1} 
\bigg)
\bigg]
\bigg\{
-
\frac{1}{n} l_c ^{(2)}(\bm{\beta^0})_{\mathcal{S}_1\mathcal{S}_1} 
\bigg\}^{-1}
}_{2} 
\\
&
\leq 
\sqrt{|\mathcal{S}_1|} \cdot
\vertiii{
\big(H_{\mathcal{S}_1\mathcal{S}_1}^{0}\big)^{-1}
}_{2} 
\cdot
\vertiii{
H_{\mathcal{S}_1\mathcal{S}_1}^{0}
-
\bigg(
-
\frac{1}{n} l_c ^{(2)}(\bm{\beta^0})_{\mathcal{S}_1\mathcal{S}_1}
\bigg) 
}_{2}
\cdot
\vertiii{
\bigg\{
-
\frac{1}{n} l_c ^{(2)}(\bm{\beta^0})_{\mathcal{S}_1\mathcal{S}_1} 
\bigg\}^{-1}
}_{2} .
\end{align*}
\normalsize
By \textbf{Lemma} \ref{lemma 2.8}, we show that with the probability at least 
$
1- 8 \exp \Big\{ -\min \big\{
\frac{C_3 n (\epsilon^{\prime})^2}{128 \mathcal{K}_3 q_n^2},
\frac{n \epsilon^{\prime}}{16 q_n c_{\ast}}
\big\}
+ \log{q_n}
\Big\}
-
8 \exp \Big\{ -\min \big\{
\frac{C_3 n (\epsilon^{\prime})^2}{128 \mathcal{K}_3 q_n^2 s_n^2},
\frac{n \epsilon^{\prime}}{16 q_n s_n c_{\ast}}
\big\}
+ \log{q_n} + \log{s_n}
\Big\}
$
and some constant $\tau_{-} = \tau - \epsilon^{\prime} > 0,$
$$
\vertiii{
\bigg\{
-
\frac{1}{n} l_c ^{(2)}(\bm{\beta^0})_{\mathcal{S}_1\mathcal{S}_1} 
\bigg\}^{-1}
}_{2} 
\leq \frac{1}{\tau_{-}}.
$$
We also have
$
\vertiii{
\big(H_{\mathcal{S}_1\mathcal{S}_1}^{0}\big)^{-1}
}_{2} \leq \frac{1}{\tau},
$ and thus, we show that
\begin{align*}
\vertiii{  
\bigg\{ 
-
\frac{1}{n} l_c ^{(2)}(\bm{\beta^0})_{\mathcal{S}_1\mathcal{S}_1} 
\bigg\}^{-1}
-
\big(H_{\mathcal{S}_1\mathcal{S}_1}^{0}\big)^{-1}
}_{\infty} 
\leq 
\frac{\sqrt{|\mathcal{S}_1|}}{\tau \cdot \tau_{-}}
\vertiii{
H_{\mathcal{S}_1\mathcal{S}_1}^{0}
-
\bigg(
-
\frac{1}{n} l_c ^{(2)}(\bm{\beta^0})_{\mathcal{S}_1\mathcal{S}_1} 
\bigg)
}_{2}.
\end{align*}
Finally, we can calculate the upper bound of $\vertiii{\mathcal{I}_4}_{\infty}$ as
\begin{align*}
\vertiii{\mathcal{I}_4}_{\infty} 
&
\leq
\vertiii{
-
\frac{1}{n}
l_c ^{(2)}(\bm{\beta^0})_{\mathcal{S}^{c}\mathcal{S}_1}
- H_{\mathcal{S}^{c}\mathcal{S}_1}^{0} 
}_{\infty}
\vertiii{
\bigg(
-
\frac{1}{n}
l_c ^{(2)}(\bm{\beta^0})_{\mathcal{S}_1\mathcal{S}_1}
\bigg)^{-1} 
- \big(H_{\mathcal{S}_1\mathcal{S}_1}^{0}\big)^{-1}
}_{\infty}
\\
&
\leq 
\vertiii{
-
\frac{1}{n}
l_c ^{(2)}(\bm{\beta^0})_{\mathcal{S}^{c}\mathcal{S}_1}
- H_{\mathcal{S}^{c}\mathcal{S}_1}^{0} 
}_{\infty}
\frac{\sqrt{|\mathcal{S}_1|}}{\tau \cdot \tau_{-}}
\vertiii{
H_{\mathcal{S}_1\mathcal{S}_1}^{0}
-
\bigg(
-
\frac{1}{n} l_c ^{(2)}(\bm{\beta^0})_{\mathcal{S}_1\mathcal{S}_1} 
\bigg)
}_{2}
\\
&
\leq
\frac{\tau\epsilon}{\sqrt{|\mathcal{S}_1|}}
\frac{\sqrt{|\mathcal{S}_1|}}{\tau \cdot \tau_{-} }
\epsilon^{\prime}
=
\frac{\epsilon^{\prime}} {\tau_{-}}
\epsilon .
\end{align*}
We let $\epsilon^{\prime} < \tau_{-} = \tau - \epsilon^{\prime} (\text{which is equivalent to } \epsilon^{\prime} < 0.5\tau), \epsilon < \frac{\xi}{3},$ then we have 
\begin{align*}
\vertiii{\frac{1}{n}
l_c ^{(2)}(\bm{\beta^0})_{\mathcal{S}^{c}\mathcal{S}_1}
\bigg(
\frac{1}{n}
l_c ^{(2)}(\bm{\beta^0})_{\mathcal{S}_1\mathcal{S}_1}
\bigg)^{-1}}_{\infty}
&
\leq
\vertiii{\mathcal{I}_1}_{\infty}
+
\vertiii{\mathcal{I}_2}_{\infty}
+
\vertiii{\mathcal{I}_3}_{\infty}
+
\vertiii{\mathcal{I}_4}_{\infty}
\\
&
\leq 
1-\xi + (1-\xi)\epsilon + \epsilon +
\frac{\epsilon^{\prime}} {\tau_{-}}
\epsilon
\\
&
\leq
1 - \xi + \xi - \frac{1}{3} \xi^2 = 1 - \frac{1}{3}\xi^2
\end{align*}
with the probability at least
$
1 - 8 \exp \Big\{ - \min 
\big\{
\frac{C_3 n (\epsilon^{\prime})^{2} }{128\mathcal{K}_3 q_n^2},
\frac{n \epsilon^{\prime}}{16 q_n c_{\ast}}
\big\}
+ \log{q_n}
\Big\}
- 8  \exp \Big\{ - \min 
\big\{
\frac{C_3 n (\epsilon^{\prime})^{2} }{128 \mathcal{K}_3 q_n^2 s_n^2 }, \\
\frac{n \epsilon^{\prime}}{16 q_n s_n c_{\ast}}
\big\}
+ \log{q_n} + \log{s_n}
\Big\}
- 8 \exp \Big\{ - \min
\big\{
\frac{C_3 n \tau_{-}^2 \epsilon^2 }{128 \mathcal{K}_3 q_n},
\frac{n \tau_{-} \epsilon }{16 \sqrt{q_n} c_{\ast} }
\big\}
+ \log{q_n}
\Big\}
- 8 \exp \Big\{ - \min
\big\{
\frac{C_3 n \tau_{-}^2 \epsilon^2 }{128 \mathcal{K}_3 s_n^2 q_n}, \\
\frac{n \tau_{-} \epsilon }{16 \sqrt{q_n} s_n c_{\ast} }
\big\}
+ \log{q_n} + \log{s_n}
\Big\}
- 8 \exp \Big\{ - \min
\big\{
\frac{C_3 n \tau^2 \epsilon^{2} }{128\mathcal{K}_3 q_n},
\frac{n \tau \epsilon}{16 \sqrt{q_n} c_{\ast}}
\big\}
+ \log{(p_n^2 - p_n - q_n)}
\Big\}
- 
8 \exp \Big\{ - \min
\big\{
\frac{C_3 n \tau^2 \epsilon^{2} }{128\mathcal{K}_3 q_n s_n^2},
\frac{n \tau \epsilon}{16 \sqrt{q_n} s_n c_{\ast}}
\big\}
+ \log{(p_n^2 - p_n - q_n)} + \log{s_n}
\Big\}.
$
\end{proof}

\begin{lemma} \label{lemma 2.5}
For 
$
\forall \;
s, w \in \mathcal{E}
$ and 
$ \forall \; j = 1,...,p_n
$, regarding all the second order partial derivatives of $-l_c(\bm{\beta})$ at $\bm{\beta} = \bm{\beta}^0,$ we have
\begin{align*}
\max_{s} \abs{ \frac{\partial^2 - {l}_c(\bm{\beta^0})}{\partial \theta_{s} \partial \theta_{w}}
- E\{ \frac{\partial^2 - {l}_c(\bm{\beta^0})}{\partial \theta_{s}
\partial \theta_{w}} \}} 
&= \mathcal{O}_p
\big\{
(n\log{p_n})^{\frac{1}{2}}
\big\},
\\
\max_{s,j} \abs{ \frac{\partial^2 - {l}_c(\bm{\beta^0})}{\partial \theta_{s} \partial \theta_{jj}}
- E\{ \frac{\partial^2 - {l}_c(\bm{\beta^0})}{\partial \theta_{s}
\partial \theta_{jj}} \}} 
&= \mathcal{O}_p
\big\{
(ns_n^2\log{p_n})^{\frac{1}{2}}
\big\},
\\
\max_{j} \abs{ \frac{\partial^2 - {l}_c(\bm{\beta^0})}{\partial \theta_{jj}^2}
- E\{ \frac{\partial^2 - {l}_c(\bm{\beta^0})}{\partial \theta_{jj}^2} \}} 
&= \mathcal{O}_p
\big\{
(ns_n^4\log{p_n})^{\frac{1}{2}}
\big\}.
\end{align*}

\begin{proof}
Following the approach in \textbf{Lemma} \ref{lemma 2.2}, for $\forall \; \epsilon > 0$, we are able to provide the probabilities outlined below, 
\begin{align*}
&
P\bigg\{
\abs{
\frac{\partial^2 -{l}_c(\bm{\beta^0})}
{\partial \theta_{s} \partial \theta_{w}} 
-
E\big\{  
\frac{\partial^2 -{l}_c(\bm{\beta^0})}
{\partial \theta_{s} \partial \theta_{w}}  
\big\}
}
\geq \epsilon
\bigg\} 
\leq 
2 \exp \bigg\{ -\min
\Big\{
\frac{C_3\epsilon^2}{2n \mathcal{K}_3},
\frac{\epsilon}{2c_{\ast}}
\Big\}
\bigg\},
\\
& 
P\bigg\{
\abs{
\frac{\partial^2 -{l}_c(\bm{\beta^0})}
{\partial \theta_{s}^2} 
-
E\big\{  
\frac{\partial^2 -{l}_c(\bm{\beta^0})}
{\partial \theta_{s}^2}  
\big\}
}
\geq \epsilon
\bigg\} 
\leq 
4 \exp \bigg\{ -\min
\Big\{
\frac{C_3\epsilon^2}{8n \mathcal{K}_3},
\frac{\epsilon}{4c_{\ast}}
\Big\}
\bigg\},
\\
& 
P\bigg\{
\abs{
\frac{\partial^2 -{l}_c(\bm{\beta^0})}
{\partial \theta_{s} \partial \theta_{jj}} 
-
E\big\{  
\frac{\partial^2 -{l}_c(\bm{\beta^0})}
{\partial \theta_{s} \partial \theta_{jj}}  
\big\}
}
\geq \epsilon
\bigg\} 
\leq
4 \exp \bigg\{ -\min
\Big\{
\frac{C_3\epsilon^2}{8\mathcal{K}_3ns_n^2},
\frac{\epsilon}{4c_{\ast}s_n}
\Big\}
+ \log{s_n}
\bigg\},
\\
&
P\bigg\{
\abs{
\frac{\partial^2 -{l}_c(\bm{\beta^0})}
{\partial \theta_{jj}^2} 
-
E\big\{  
\frac{\partial^2 -{l}_c(\bm{\beta^0})}
{\partial \theta_{jj}^2}  
\big\}
}
\geq \epsilon
\bigg\} 
\leq
6 \exp \bigg\{ -\min
\Big\{
\frac{C_3\epsilon^2}{98\mathcal{K}_3ns_n^2},
\frac{\epsilon}{14c_{\ast}s_n}
\Big\}
+ \log{s_n}
\bigg\}
\\
& +
4 \exp \bigg\{ -\min
\Big\{
\frac{2C_3\epsilon^2}{49\mathcal{K}_3ns_n^4},
\frac{\epsilon}{7c_{\ast}s_n^2}
\Big\}
+ \log{s_n} + \log{(s_n-1)}
\bigg\},
\end{align*}
for some positive constants $C_3, \mathcal{K}_3, c_{\ast} > 0$ and $\forall \; j=1,\cdots, p_n$, and $\forall \; s, w \in \mathcal{E}$, $s \neq w$. The detailed proof is similar to \textbf{Lemma} \ref{lemma 2.2} and therefore, is omitted.
\end{proof}
\end{lemma}

\begin{lemma} \label{lemma 2.6}
For $\forall \; \epsilon > 0$, we show that 
\begin{align*}
& 
P\Bigg\{
\vertiii{ -\frac{1}{n} l_c ^{(2)}(\bm{\beta^0})_{\mathcal{S}_1\mathcal{S}_1}
- 
E\bigg\{
-
\frac{1}{n}
l_c ^{(2)}(\bm{\beta^0})_{\mathcal{S}_1\mathcal{S}_1}
\bigg\}
}_{\infty}
\geq \epsilon
\Bigg\}
\\
&
\leq 
8 \exp\Bigg\{ -\min 
\bigg\{
\frac{C_3n\epsilon^2}{128\mathcal{K}_3},
\frac{n\epsilon}{16c_{\ast}}
\bigg\}
+
\log{q_n}
\Bigg\} 
+ 
8 \exp \Bigg\{ -\min 
\bigg\{
\frac{C_3n\epsilon^2}{128\mathcal{K}_3s_n^2},
\frac{n \epsilon}{16s_nc_{\ast}}
\bigg\}
+
\log{q_n} + \log{s_n}
\Bigg\} ,
\\
&
P\Bigg\{
\vertiii{ -\frac{1}{n} l_c ^{(2)}(\bm{\beta^0})_{\mathcal{S}^c\mathcal{S}_1}
- 
E\bigg\{
-
\frac{1}{n}
l_c ^{(2)}(\bm{\beta^0})_{\mathcal{S}^c\mathcal{S}_1}
\bigg\}
}_{\infty}
\geq \epsilon
\Bigg\}
\\
&
\leq 
8 \exp\Bigg\{ -\min 
\bigg\{
\frac{C_3n\epsilon^2}{128\mathcal{K}_3},
\frac{n\epsilon}{16c_{\ast}}
\bigg\}
+
\log{(p_n^2 - p_n - q_n)}
\Bigg\} 
\\
&
+ 
8 \exp \Bigg\{ -\min 
\bigg\{
\frac{C_3n\epsilon^2}{128\mathcal{K}_3s_n^2},
\frac{n \epsilon}{16s_nc_{\ast}}
\bigg\}
+
\log{(p_n^2 - p_n - q_n)}
+ 
\log{s_n}
\Bigg\} .
\end{align*}
\end{lemma}
\begin{proof}
For any matrix $A$, the $\infty$ norm is defined as 
$$
\vertiii{A}_{\infty} = 
\max_{1\leq i \leq p_n} 
\sum_{j=1}^{p_n}
|
A_{ij} 
|.
$$
Let $l_{kl}^{(2)}(\bm{\beta^0})$ denote the $(k,l)th$ entry of $l_c^{(2)}(\bm{\beta^0})$. The probability can be formulated as
\begin{align*}
&
P\Bigg\{
\vertiii{ -\frac{1}{n} l_c ^{(2)}(\bm{\beta^0})_{\mathcal{S}_1\mathcal{S}_1}
- 
E\bigg\{
-
\frac{1}{n}
l_c ^{(2)}(\bm{\beta^0})_{\mathcal{S}_1\mathcal{S}_1}
\bigg\}
}_{\infty}
\geq \epsilon
\Bigg\}
\\
&
=
P\Bigg\{
\max_{k \in \mathcal{S}_1}
\frac{1}{n}\sum_{l \in \mathcal{S}_1}
\abs{
-
l_{kl}^{(2)} (\bm{\beta^0}) 
- E\big\{
-
l_{kl}^{(2)} (\bm{\beta^0}) 
\big\}
} > \epsilon
\Bigg\}.
\end{align*}
\justifying 
In particular, for $\forall \; s, w \in \mathcal{S}_1, s \neq w$,
the row summation may include
\begin{align*}
& \abs{
\frac{\partial^2 -{l}_c(\bm{\beta^0})}
{\partial \theta_{s}^2} 
-
E\big\{  
\frac{\partial^2 -{l}_c(\bm{\beta^0})}
{\partial \theta_{s}^2}  
\big\}
} ,
\quad
& \abs{
\frac{\partial^2 -{l}_c(\bm{\beta^0})}
{\partial \theta_{s} \partial \alpha_{s}} 
-
E\big\{  
\frac{\partial^2  -{l}_c(\bm{\beta^0})}
{\partial \theta_{s} \partial \alpha_{s}}  
\big\}
}  ,
\\
& \abs{
\frac{\partial^2 -{l}_c(\bm{\beta^0})}
{\partial \theta_{s} \partial \alpha_{w}} 
-
E\big\{  
\frac{\partial^2  -{l}_c(\bm{\beta^0})}
{\partial \theta_{s} \partial \alpha_{w}}  
\big\}
} ,
\quad
& \abs{
\frac{\partial^2 -{l}_c(\bm{\beta^0})}
{\partial \theta_{s} \partial \theta_{w}} 
-
E\big\{  
\frac{\partial^2  -{l}_c(\bm{\beta^0})}
{\partial \theta_{s} \partial \theta_{w}}  
\big\}
} ,
\end{align*}
or
\begin{align*}
& \abs{
\frac{\partial^2 -{l}_c(\bm{\beta^0})}
{\partial \alpha_{s}^2} 
-
E\big\{  
\frac{\partial^2 -{l}_c(\bm{\beta^0})}
{\partial \alpha_{s}^2}  
\big\}
} ,
\quad
& \abs{
\frac{\partial^2 -{l}_c(\bm{\beta^0})}
{\partial \alpha_{s} \partial \theta_{s}} 
-
E\big\{  
\frac{\partial^2 -{l}_c(\bm{\beta^0})}
{\partial \alpha_{s} \partial \theta_{s}}  
\big\}
} ,
\\
& \abs{
\frac{\partial^2 -{l}_c(\bm{\beta^0})}
{\partial \alpha_{s} \partial \alpha_{w}} 
-
E\big\{  
\frac{\partial^2  -{l}_c(\bm{\beta^0})}
{\partial \alpha_{s} \partial \alpha_{w}}  
\big\}
} ,
\quad
&
\abs{
\frac{\partial^2 -{l}_c(\bm{\beta^0})}
{\partial \alpha_{s} \partial \theta_{w}} 
-
E\big\{  
\frac{\partial^2  -{l}_c(\bm{\beta^0})}
{\partial \alpha_{s} \partial \theta_{w}}  
\big\}
} .
\end{align*}

Without loss of generality, we consider the first case. For $\forall \; s=\{(a,b)\} \in \mathcal{S}_1$, we have
\begin{align*}
\frac{\partial^2  -{l}_c(\bm{\beta^0})}
{\partial \theta_{s} \partial \theta_{w}}
\begin{cases}
  = \sum_{m=1}^{n}x_m^2 \left[ 
  x_m \theta_{aa}
  + (1-x_m) \alpha_{aa}
  \right]^{-1} y_{mb} y_{mc},& \text{if}\; w=\{(a,c) \} \in \mathcal{E},  \\
  = \sum_{m=1}^{n}x_m^2 \left[ 
  x_m \theta_{bb}
  + (1-x_m) \alpha_{bb}
  \right]^{-1} y_{ma} y_{mc},& \text{if}\; w=\{(b,c) \} \in \mathcal{E},  \\
=0, & \text{otherwise}.
\end{cases}
\end{align*}

For each $\forall \; s = \{(a,b)\},$ the total number of nonzero terms $\frac{\partial^2  {l}_c(\bm{\beta^0})}
{\partial \theta_{s} \partial \theta_{w}}$ is equal to the number of possible choices for $w$, which is less than $2s_n$ in the subset $\mathcal{S}_1$. Likewise, for each $s$, the total number of nonzero terms $\frac{\partial^2  {l}_c(\bm{\beta^0})}
{\partial \theta_{s} \partial \alpha_{w}}$ is less than $2s_n$ as well. 

As a result, for $\forall \; k \in \mathcal{S}_1$, by \textbf{Lemma} \ref{lemma 2.5}, we have
\begin{align*}
& P\bigg\{
\frac{1}{n}\sum_{l: l \in \mathcal{S}_1}
\abs{
-
l_{kl}^{(2)} (\bm{\beta^0}) 
- E\big\{
-
l_{kl}^{(2)} (\bm{\beta^0}) 
\big\}
} \geq \epsilon
\bigg\}
\\
& \leq
 8 \exp\Bigg\{ -\min 
\bigg\{
\frac{C_3n\epsilon^2}{128\mathcal{K}_3},
\frac{n\epsilon}{16c_{\ast}}
\bigg\}
\Bigg\} 
+ 
8 \exp \Bigg\{ -\min 
\bigg\{
\frac{C_3n\epsilon^2}{128\mathcal{K}_3s_n^2},
\frac{n \epsilon}{16s_nc_{\ast}}
\bigg\}
+ \log{s_n}
\Bigg\} .
\end{align*}
Through Bonferroni inequality, we show that
\begin{align*}
&
P\Bigg\{
\vertiii{ -\frac{1}{n} l_c ^{(2)}(\bm{\beta^0})_{\mathcal{S}_1\mathcal{S}_1}
- 
E\bigg\{
-
\frac{1}{n}
l_c ^{(2)}(\bm{\beta^0})_{\mathcal{S}_1\mathcal{S}_1}
\bigg\}
}_{\infty}
\geq \epsilon
\Bigg\}
\\
&
\leq \abs{\mathcal{S}_1}
P\bigg\{
\frac{1}{n}\sum_{l: l \in \mathcal{S}_1}
\abs{
-
l_{kl}^{(2)} (\bm{\beta^0}) 
- E\big\{
-
l_{kl}^{(2)} (\bm{\beta^0}) 
\big\}
} \geq \epsilon
\bigg\}
\\
&
\leq 
8 \exp\Bigg\{ -\min 
\bigg\{
\frac{C_3n\epsilon^2}{128\mathcal{K}_3},
\frac{n\epsilon}{16c_{\ast}}
\bigg\}
+
\log{q_n}
\Bigg\} 
+ 
8 \exp \Bigg\{ -\min 
\bigg\{
\frac{C_3n\epsilon^2}{128\mathcal{K}_3s_n^2},
\frac{n \epsilon}{16s_nc_{\ast}}
\bigg\}
+ \log{s_n} +
\log{q_n}
\Bigg\} .
\end{align*}
Likewise, we can show similar results for another subset $\mathcal{S}^c$ as follows, 
\begin{align*}
&
P\Bigg\{
\vertiii{ -\frac{1}{n} l_c ^{(2)}(\bm{\beta^0})_{\mathcal{S}^c\mathcal{S}_1}
- 
E\bigg\{
-
\frac{1}{n}
l_c ^{(2)}(\bm{\beta^0})_{\mathcal{S}^c\mathcal{S}_1}
\bigg\}
}_{\infty}
\geq \epsilon
\Bigg\}
\\
&
\leq 
8 \exp\Bigg\{ -\min 
\bigg\{
\frac{C_3n\epsilon^2}{128\mathcal{K}_3},
\frac{n\epsilon}{16c_{\ast}}
\bigg\}
+
\log{(p_n^2 - p_n - q_n)}
\Bigg\} 
\\
&
+ 
8 \exp \Bigg\{ -\min 
\bigg\{
\frac{C_3n\epsilon^2}{128\mathcal{K}_3s_n^2},
\frac{n \epsilon}{16s_nc_{\ast}}
\bigg\}
+ \log{s_n} +
\log{(p_n^2 - p_n - q_n)}
\Bigg\}.
\end{align*}
\end{proof}

\begin{lemma}\label{lemma 2.7}
For $\forall \; \epsilon > 0$, we show that
\begin{align*}
& 
P\Bigg\{
\vertiii{  E\bigg\{
-
\frac{1}{n}
l_c ^{(2)}(\bm{\beta^0})_{\mathcal{S}_1\mathcal{S}_1}
\bigg\}
- 
\bigg( -
\frac{1}{n} l_c ^{(2)}(\bm{\beta^0})_{\mathcal{S}_1\mathcal{S}_1}
\bigg)
}_{2}
\geq \epsilon
\Bigg\}
\\
&
\leq
8 \exp \Bigg\{
-\min \bigg\{
\frac{C_3n\epsilon^2}
{128\mathcal{K}_3q_n^2},
\frac{n\epsilon}
{16q_nc_{*}}
\bigg\}
+ \log{q_n}
\Bigg\}
\\
&
+ 8
\exp \Bigg\{
-\min \bigg\{
\frac{C_3n\epsilon^2}
{128\mathcal{K}_3q_n^2s_n^2},
\frac{n\epsilon}
{16q_ns_nc_{*}}
\bigg\}
+ \log{q_n}
+ \log{s_n}
\Bigg\} .
\end{align*}
\end{lemma}

\begin{proof}
Since $\vertiii{A}_{2} \leq
\vertiii{A}_{F}$, which leads to 
$
P\big\{ \vertiii{A}_{2} \geq \epsilon \big\} \leq 
P\big\{ \vertiii{A}_{F} \geq \epsilon \big\}
$, 
we have
\begin{align*}
&
\vertiii{  E\bigg\{
-
\frac{1}{n}
l_c ^{(2)}(\bm{\beta^0})_{\mathcal{S}_1\mathcal{S}_1}
\bigg\}
- \bigg(
-
\frac{1}{n} l_c ^{(2)}(\bm{\beta^0})_{\mathcal{S}_1\mathcal{S}_1}
\bigg)
}_{F}
\\
&
\leq 
\sum_{k,l \in \mathcal{S}_1} \abs{
-
\frac{1}{n} l_{kl} ^{(2)}(\bm{\beta^0})
-
E\bigg\{
-
\frac{1}{n} l_{kl} ^{(2)}(\bm{\beta^0})
\bigg\}
} .
\end{align*}
We start with the following probability
\small
\begin{align*}
& P\Bigg\{
\vertiii{  E\bigg\{
-
\frac{1}{n}
l_c ^{(2)}(\bm{\beta^0})_{\mathcal{S}_1\mathcal{S}_1}
\bigg\}
- 
\bigg(
-\frac{1}{n} l_c ^{(2)}(\bm{\beta^0})_{\mathcal{S}_1\mathcal{S}_1}
\bigg)
}_{F}
\geq \epsilon
\Bigg\}
\\
&
\leq 
P\Bigg\{
\sum_{k,l \in \mathcal{S}_1} \abs{
-
l_{kl} ^{(2)}(\bm{\beta^0})
-
E\bigg\{
-
l_{kl} ^{(2)}(\bm{\beta^0})
\bigg\}
}
\geq 
n\epsilon
\Bigg\}  .
\end{align*}
\normalsize
The possible terms included in the summation $\sum_{k,l \in \mathcal{S}_1}$ are as follows. 

\textbf{Case 1:} For $\forall \; s \in \mathcal{S}_1$, the total number of terms for the following three types 
\begin{align*}
&
\abs{
\frac{\partial^2 -l_c(\bm{\beta^0})}{\partial \theta_s^2}
-
E\bigg\{
\frac{\partial^2 -l_c(\bm{\beta^0})}{\partial \theta_s^2}
\bigg\}
},
&
&
\abs{
\frac{\partial^2 -l_c(\bm{\beta^0})}{\partial \alpha_s^2}
-
E\bigg\{
\frac{\partial^2 -l_c(\bm{\beta^0})}{\partial \alpha_s^2}
\bigg\}
},
\\
&
\abs{
\frac{\partial^2 -l_c(\bm{\beta^0})}{\partial \theta_s \partial \alpha_s}
-
E\bigg\{
\frac{\partial^2 -l_c(\bm{\beta^0})}{\partial \theta_s \partial \alpha_s}
\bigg\}
},
\end{align*}
are $2q_n$. 

\textbf{Case 2:} For $\forall \; s \in \mathcal{S}_1, \forall \; w \in \mathcal{S}_1$ and $s \neq w,$ it can be shown that the total number of nonzero terms for the following four types is less than $4q_ns_n$,
\begin{align*}
& \abs{
\frac{\partial^2 -{l}_c(\bm{\beta^0})}
{\partial \theta_{s} \partial \theta_{w}} 
-
E\bigg\{  
\frac{\partial^2 -{l}_c(\bm{\beta^0})}
{\partial \theta_{s} \partial \theta_{w}}  
\bigg\}
} ,
\quad
& \abs{
\frac{\partial^2 -{l}_c(\bm{\beta^0})}
{\partial \alpha_{s} \partial \alpha_{w}} 
-
E\bigg\{  
\frac{\partial^2  -{l}_c(\bm{\beta^0})}
{\partial \alpha_{s} \partial \alpha_{w}}  
\bigg\}
} ,
\\
& \abs{
\frac{\partial^2 -{l}_c(\bm{\beta^0})}
{\partial \theta_{s} \partial \alpha_{w}} 
-
E\bigg\{  
\frac{\partial^2 -{l}_c(\bm{\beta^0})}
{\partial \theta_{s} \partial \alpha_{w}}  
\bigg\}
} ,
\quad
& \abs{
\frac{\partial^2 -{l}_c(\bm{\beta^0})}
{\partial \alpha_{w} \partial \theta_{s}} 
-
E\bigg\{  
\frac{\partial^2  -{l}_c(\bm{\beta^0})}
{\partial \alpha_{w} \partial \theta_{s}}  
\bigg\}
}.
\end{align*}
This is because for $\forall \; s=\{(a,b) \} \in \mathcal{S}_1$, we know that 
\begin{align*}
&
\frac{\partial^2 - {l}_c(\bm{\beta^0})}
{\partial \theta_{s} \partial \theta_{w}}
\begin{cases}
 = \sum_{m=1}^{n}x_m^2 \left[ 
 x_m \theta_{aa} + (1-x_m)\alpha_{aa}
 \right]^{-1}y_{mb}y_{mc},& \text{if}\; w = \{(a,c) \} \in \mathcal{S}_1, 
 \\
 =
 \sum_{m=1}^{n}x_m^2 \left[ 
 x_m \theta_{bb} + (1-x_m)\alpha_{bb}
\right]^{-1}y_{ma}y_{mc},& \text{if}\; w = \{(b,c) \} \in \mathcal{S}_1, 
 \\
=0, & \text{otherwise}.
\end{cases}
\end{align*}
The number of nonzero terms for each row would be less than $4s_n$. Therefore, by \textbf{Lemma} \ref{lemma 2.5}, for $\forall \; \epsilon > 0,$ the probability can be written as 

\begin{align*}
&
P\Bigg\{
\sum_{k,l \in \mathcal{S}_1} \abs{
-
l_{kl} ^{(2)}(\bm{\beta^0})
-
E\bigg\{
-
l_{kl} ^{(2)}(\bm{\beta^0})
\bigg\}
}
\geq 
n\epsilon
\Bigg\} 
\\
&
\leq 
\sum_{2q_n}
P\Bigg\{
\abs{
\text{Any term in Case 1}
}
\geq
\frac{n\epsilon}{4q_n}
\Bigg\}
+
\sum_{4q_ns_n}
P\Bigg\{
\abs{
\text{Any term in Case 2}
}
\geq
\frac{n\epsilon}{8q_ns_n}
\Bigg\}
\\
&
\leq
8 \exp \Bigg\{
-\min \bigg\{
\frac{C_3n\epsilon^2}
{128\mathcal{K}_3q_n^2},
\frac{n\epsilon}
{16q_n c_{*}}
\bigg\}
+ \log{q_n}
\Bigg\} 
\\
&
+
8\exp \Bigg\{
-\min \bigg\{
\frac{C_3n\epsilon^2}
{128\mathcal{K}_3q_n^2s_n^2},
\frac{n\epsilon}
{16q_n s_n c_{*}}
\bigg\}
+ \log{q_n} + \log{s_n}
\Bigg\},
\end{align*}
which completes the proof. 
\end{proof}

\normalsize

\begin{lemma}\label{lemma 2.8}
Under \textbf{Assumptions 3-5} and $q_n^2p_n^{d}\log{p_n} = o(n)$, there exists some constant $\tau_{-} > 0$ such that $-\frac{1}{n}l_c ^{(2)}(\bm{\beta^0})_{\mathcal{S}_1\mathcal{S}_1}$ has minimum eigenvalue lower bounded by $\tau_{-}$ and is invertible with  probability tending to one.
\end{lemma}

\begin{proof}
According to the \textbf{Assumption 3}, we know that $H(\bm{\beta^0}) = 
E\big\{
-\frac{1}{n} l_c ^{(2)}(\bm{\beta^0})
\big\}
$
has eigenvalues bounded away from zero and infinity. Hence, for the sub-matrix $
H_{\mathcal{S}_1\mathcal{S}_1}^{0} = 
E\big\{
-\frac{1}{n} l_c ^{(2)}(\bm{\beta^0})_{\mathcal{S}_1 \mathcal{S}_1}
\big\},
$ 
there exists some constant $\tau > 0$ such that 
$
\lambda_{\min}(H_{\mathcal{S}_1\mathcal{S}_1}^{0}) \geq \tau .
$
By the Courant-Fischer variational representation, we have 
\begin{align*}
\lambda_{\min} (H_{\mathcal{S}_1\mathcal{S}_1}^{0})
&=
\lambda_{\min} \bigg(
-
\frac{1}{n}l_c ^{(2)}(\bm{\beta^0})_{\mathcal{S}_1\mathcal{S}_1}
+
H_{\mathcal{S}_1\mathcal{S}_1}^{0}
+
\frac{1}{n}l_c ^{(2)}(\bm{\beta^0})_{\mathcal{S}_1\mathcal{S}_1}
\bigg)
\\
& 
= \min_{||v||_2=1} \bigg\{
v^{T}
\bigg(
-
\frac{1}{n}l_c ^{(2)}(\bm{\beta^0})_{\mathcal{S}_1\mathcal{S}_1}
+
H_{\mathcal{S}_1\mathcal{S}_1}^{0}
+
\frac{1}{n}l_c ^{(2)}(\bm{\beta^0})_{\mathcal{S}_1\mathcal{S}_1}
\bigg)
v
\bigg\}
\\
& 
= \min_{||v||_2=1} \bigg\{
v^{T} \bigg(
-
\frac{1}{n}l_c ^{(2)}(\bm{\beta^0})_{\mathcal{S}_1\mathcal{S}_1}
\bigg)
v
+ v^{T} \bigg(
H_{\mathcal{S}_1\mathcal{S}_1}^{0}
+
\frac{1}{n}l_c ^{(2)}(\bm{\beta^0})_{\mathcal{S}_1\mathcal{S}_1}
\bigg)
v
\bigg\}
\\
& 
\leq
y^{T} \bigg(
-
\frac{1}{n}l_c ^{(2)}(\bm{\beta^0})_{\mathcal{S}_1\mathcal{S}_1}
\bigg)
y 
+
y^{T} \bigg(
H_{\mathcal{S}_1\mathcal{S}_1}^{0}
+
\frac{1}{n}l_c ^{(2)}(\bm{\beta^0})_{\mathcal{S}_1\mathcal{S}_1}
\bigg)
y ,
\end{align*}
where $y$ is a unit-norm eigenvector of $-\frac{1}{n}l_c ^{(2)}(\bm{\beta^0})_{\mathcal{S}_1\mathcal{S}_1}.$
Therefore, we show that
\begin{align*}
&
y^{T} \bigg(
-
\frac{1}{n}l_c ^{(2)}(\bm{\beta^0})_{\mathcal{S}_1\mathcal{S}_1}
\bigg)
y 
\geq 
\lambda_{\min} \bigg(
-
\frac{1}{n} l_c ^{(2)}(\bm{\beta^0})_{\mathcal{S}_1\mathcal{S}_1}
\bigg) 
\geq
\lambda_{\min} (H_{\mathcal{S}_1\mathcal{S}_1}^{0}) 
- 
y^{T} 
\bigg(
H_{\mathcal{S}_1\mathcal{S}_1}^{0}
+
\frac{1}{n}l_c ^{(2)}(\bm{\beta^0})_{\mathcal{S}_1\mathcal{S}_1}
\bigg)
y,
\\
&
\lambda_{\min} \bigg(
-
\frac{1}{n} l_c ^{(2)}(\bm{\beta^0})_{\mathcal{S}_1\mathcal{S}_1}
\bigg)  
\geq
\tau
- \vertiii{
H_{\mathcal{S}_1\mathcal{S}_1}^{0}
+
\frac{1}{n}l_c ^{(2)}(\bm{\beta^0})_{\mathcal{S}_1\mathcal{S}_1}
}_{2} .
\end{align*}
Since $\vertiii{A}_2 \leq \vertiii{A}_{F}$, choose $\epsilon$ such that $0 < \epsilon < \tau,$ by \textbf{Lemma} \ref{lemma 2.7}, we have 
\begin{align*}
&
P\Bigg\{
\vertiii{  
H_{\mathcal{S}_1\mathcal{S}_1}^{0} +
\frac{1}{n} l_c ^{(2)}(\bm{\beta^0})_{\mathcal{S}_1\mathcal{S}_1} 
}_{2}
\geq \epsilon
\Bigg\}
\leq 
P\Bigg\{
\vertiii{  
H_{\mathcal{S}_1\mathcal{S}_1}^{0} +
\frac{1}{n} l_c ^{(2)}(\bm{\beta^0})_{\mathcal{S}_1\mathcal{S}_1} 
}_{F}
\geq \epsilon
\Bigg\} 
\\
&
\leq
8 \exp \Big \{ -\min \big\{
\frac{C_3 n \epsilon^2}{128 \mathcal{K}_3 q_n^2},
\frac{n \epsilon}{16 q_n c_{\ast}}
\big\}
+ \log{q_n}
\Big\}
\\
&
+ 8 \exp \Big \{ -\min \big\{
\frac{C_3 n \epsilon^2}{128 \mathcal{K}_3 q_n^2 s_n^2},
\frac{n \epsilon}{16 q_n s_n c_{\ast}}
\big\}
+ \log{q_n} + \log{s_n}
\Big\} .
\end{align*}
As a result, 
$
\lambda_{\min} \Big(
-
\frac{1}{n} l_c ^{(2)}(\bm{\beta^0})_{\mathcal{S}_1\mathcal{S}_1}
\Big) 
\geq
\tau
- \epsilon  = \tau_{-}
$ holds with probability tending to 1.
\end{proof}

\section{Proofs of Theorems}
\textbf{Proof of Theorem 1}

\begin{proof}
The multivariate version of Lindeberg condition can be bounded by 
\begin{align*}
&n^{-1}
\sum_{m=1}^{n} E\Big\{
 \vert \vert \bm{\mathcal{L}}_m^{(1)}(\bm \beta^0) \vert \vert_2 ^2 \cdot
I\Big(|| \bm{\mathcal{L}}_m^{(1)}(\bm \beta^0) ||_2 \geq \epsilon \sqrt{n}  
\Big)
\Big\} 
\\
& \leq n^{-1}
\sum_{m=1}^{n} E\bigg\{ \vert \vert \bm{\mathcal{L}}_m^{(1)}(\bm \beta^0) \vert \vert_2 ^2
\cdot \bigg\{
\frac{\vert \vert \bm{\mathcal{L}}_m^{(1)}(\bm \beta^0) \vert \vert_2 }{\epsilon \sqrt{n}}
\bigg\}^{2}
\bigg\} 
\\
& = \epsilon^{-2}  n^{-2} \sum_{m=1}^{n} E\bigg\{
\vert \vert \bm{\mathcal{L}}_m^{(1)}(\bm \beta^0) \vert \vert_2 ^{4}
\bigg\}.
\end{align*}
Denote each $(i,j)$th element of $\bm{\Sigma}_{m}$ as $\sigma_{ij,m}^2$ and it can be shown that 
\begin{equation}
\begin{split}
&\frac{\partial \bm{\mathcal{L}}_m(\bm{\beta}) }{\partial [\theta_{u}]_{h}} = 
\begin{cases}
\frac{1}{2}x_m^{(h_1)} \bm[\sigma_{jj,m}^2 - y_{mj}^2\bm]   & \text{w.r.t}\;\; [\theta_{jj}]_{h},
\\
x_m^{(h_1)}[\sigma_{ji,m}^2 - y_{mj}y_{mi}] & \text{w.r.t}\;\; [\theta_{ji}]_{h}, 
\end{cases}
\\
&\frac{\partial \bm{\mathcal{L}}_m(\bm{\beta}) }{\partial \alpha_{s}} = 
\begin{cases}
\frac{1}{2}\bm[\sigma_{jj,m}^2 - y_{mj}^2\bm]  & \text{w.r.t}\;\; \alpha_{jj}, \\ \sigma_{ji,m}^2 - y_{mj}y_{mi} & \text{w.r.t}\;\; \alpha_{ji},
\end{cases}
\end{split}
\end{equation}
for all $h=1,\cdots, H, m=1,\cdots, n$. By \textbf{Assumption 1} and 
\textbf{Lemma}  \ref{lemma 2.1}, we show that
\begin{align*}
 E\Big\{
\vert \vert 
\bm{\mathcal{L}}_m^{(1)}(\bm \beta^0) 
\vert \vert_2 ^{4} 
\Big\}
&
=
E \Bigg\{
\bigg(
\underbrace{
\sum_{u, h_1}
\bigg\{
\frac{\partial \bm{\mathcal{L}}_m  (\bm \beta^0)}{\partial [\theta_{u}]_{h_1}} 
\bigg\}^2
+
\sum_{s}
\bigg\{
\frac{\partial \bm{\mathcal{L}}_m  (\bm \beta^0)}{\partial \alpha_{s}} 
\bigg\}^2}_{\text{the sum of } 0.5(H+1)p(p+1) \text{ terms}}
\bigg)^2
\Bigg\}
= \mathcal{O}(p^4),
\end{align*} 
which leads to 
$
\lim_{n \to \infty}  \epsilon^{-1} n^{-2} 
\sum_{m=1}^{n} 
E\Big\{
\vert \vert \bm{\mathcal{L}}_m^{(1)}(\bm \beta^0) \vert \vert_2 ^{4}
\Big\} 
= 0
$. Therefore, by the multivariate version of the Lindeberg-Feller
central limit theorem, we conclude that  
\begin{equation} \label{eq: First term}
\begin{split}
&
n^{-\frac{1}{2}}\bm{\mathcal{L}} ^{(1)}(\bm \beta^0) = n^{-\frac{1}{2}}
\sum_{m=1}^{n} 
\bm{\mathcal{L}}_m^{(1)}(\bm \beta^0)  \xrightarrow{\enskip d \enskip} 
N \Big(
\bm 0, \bm V 
\Big).
\end{split}
\end{equation}
Our next step is to show the existence of a local maximizer $\hat{\bm \beta}$ of $L(\bm{\beta})$ such that $\vert \vert \hat{\bm \beta} - \bm{\beta}^{0} \vert \vert_2 = \mathcal{O}_p(n^{-\frac{1}{2}})$. Define a ball $\mathcal{A} = \Big\{\bm{\beta}: \bm{\beta} \leq \bm{\beta}^{0} + C n^{-\frac{1}{2}}\bm{v} \Big\}$, where $\bm{v}$ is any unit vector and $C$ is a positive constant. Since $\mathcal{L}(\bm{\beta}) - \mathcal{L}(\bm{\beta}^0)$ is a continuous function with $\mathcal{L}(\bm{\beta}^0) - \mathcal{L}(\bm{\beta}^0) = 0$, if one can show that it is strictly negative on the boundary of the ball, we can conclude that there exists a local maximizer of $\mathcal{L}(\bm{\beta})$ inside $\mathcal{A}$. \newline
By Taylor expansion at $\bm{\beta}^0$, we have
\begin{equation}
\begin{split}
\label{Taylor Expansion 1}
\mathcal{L}(\bm{\beta}) -\mathcal{L}(\bm{\beta}^0)
=& 
\bm{\mathcal{L}}^{(1)}(\bm{\beta}^0)^{T}(\bm \beta - \bm{\beta}^0)
+
\frac{1}{2} (\bm \beta - \bm{\beta}^0)^{T}
\bm{\mathcal{L}}^{(2)}(\bm{\beta}^0)
(\bm \beta - \bm{\beta}^0)
\\
&
+
\frac{1}{6}\sum_{r,t,u}
[(\bm \beta - \bm{\beta}^0)]_{[r]}
[(\bm \beta - \bm{\beta}^0)]_{[t]}
[(\bm \beta - \bm{\beta}^0)]_{[u]}
\frac{\partial^3 \mathcal{L}(\bm{\tilde{\beta}})}{\partial \beta_r \partial \beta_t \partial \beta_u},
\end{split}
\end{equation}
where $\bm{\tilde{\beta}}$ is between $\bm{\beta}$ and $\bm{\beta}^0$. From \eqref{eq: First term}, we have $n^{-\frac{1}{2}}\bm{\mathcal{L}}^{(1)}(\bm{\beta}^0) = \mathcal{O}_p(1)$. The first term on the right-hand side of \eqref{Taylor Expansion 1} is of the order $\mathcal{O}_p(1)$. As $\bm{\mathcal{L}}^{(2)}(\bm{\beta}^0)$ is deterministic, the second term can be written as
\begin{align*}
\frac{1}{2} (\bm \beta - \bm{\beta}^0)^{T}
\bm{\mathcal{L}}^{(2)}(\bm{\beta}^0)
(\bm \beta - \bm{\beta}^0)
= 
-\frac{1}{2} n (\bm \beta - \bm{\beta}^0)^{T}
\frac{\bm{\mathcal{I}}(\bm{\beta}^0)}{n}
(\bm \beta - \bm{\beta}^0).
\end{align*}
From \textbf{Assumption 2}, we know that $\frac{\mathcal{I}(\bm{\beta}^0)}{n} \xrightarrow{} \bm{V}$ as $n \xrightarrow{} \infty$ and $\bm{V}$ is positive definite. Thus, the second term is less than $-\frac{1}{2}C^2\lambda_{\min}(\bm{V})$ as $n$ goes to infinity. By choosing a sufficiently large $C,$ the second term can dominate the first term. The third order partial derivative of $\mathcal{L}(\bm{\beta})$, for example, can be given as
\begin{align}
\label{third derivative}
\frac{\partial^3 \mathcal{L}(\bm{\beta})}{\partial [\theta_u]_{h_1} \partial [\theta_v]_{h_2} \partial \alpha_s}
=
\sum_{m=1}^{n}
x_m^{(h_1)} x_m^{(h_2)}
tr(\bm{T}^u \bm{\Sigma}_m \bm{T}^{s} \bm{\Sigma}_m
\bm{T}^{v} \bm{\Sigma}_m).
\end{align}
Other third order partial derivatives take similar forms as \eqref{third derivative}. From \textbf{Assumption 1}, we have $\frac{\partial^3 \mathcal{L}(\bm{\tilde{\beta}})}{\partial \beta_r \partial \beta_t \partial \beta_u} = \mathcal{O}(n)$ and thus, the third term is of the order $\mathcal{O}(n^{-\frac{1}{2}})$. This is dominated by the second term as well. As a result, we have shown that with probability tending to 1, the function $\mathcal{L}(\bm{\beta}) - \mathcal{L}(\bm{\beta}^0)$ is strictly negative on the ball, which completes the proof of the consistency of the estimator. \newline
Lastly, we show the asymptotic normality for $\hat{\bm{\beta}}.$ 
By Taylor expansion, there exists $\bm{\tilde{\beta}}$ between $\bm{\beta}$ and $\bm{\beta}^0$ such that for $\forall \; r,t,u=1,\cdots, \frac{p(p+1)(H+1)}{2},$
\begin{equation}
\begin{split}
\label{eq:taylor 12-16}
\big(\bm{\mathcal{L}} ^{(1)}(\hat{\bm \beta})\big)_{r} = 0
&=
\big(\bm{\mathcal{L}} ^{(1)}(\bm \beta^0)\big)_{r}
+ 
\sum_{t}
\big(\bm{\mathcal{L}}^{(2)}(\bm{\beta}^0)\big)_{rt} \cdot
\big(\hat{\bm \beta} - \bm \beta^0\big)_{t}
\\
& +
\sum_{t,u}
\frac{\partial^3 \bm{\mathcal{L}}(\bm{\tilde{\beta}})}{\partial \beta_{r} \partial \beta_{t} \partial \beta_{u}}\big(\hat{\bm \beta} - \bm \beta^0\big)_{t}
\big(\hat{\bm \beta} - \bm \beta^0\big)_{u},
\end{split}
\end{equation}
where $r,t,u$ denote the row or column index  of the vector and matrices. The third term on the right-hand side is $\mathcal{O}_p(1),$ whereas the second term has an order of $\mathcal{O}_p(n^{\frac{1}{2}}),$ which dominates the third term. Then \eqref{eq:taylor 12-16} can be written as
\begin{align*}
\frac{1}{n}
\big(\bm{\mathcal{L}} ^{(1)}(\bm \beta^0)\big)_{r}
=
\sum_{t}
\big(-\frac{1}{n}\bm{\mathcal{L}}^{(2)}(\bm{\beta}^0)\big)_{rt} \cdot
\big(\hat{\bm \beta} - \bm \beta^0\big)_{t} (1+o_p(1)).
\end{align*}
In matrix form, we have 
\begin{align}
\label{eq: 12}
\frac{1}{n}
\bm{\mathcal{L}} ^{(1)}(\bm \beta^0)
=
-\frac{1}{n}\bm{\mathcal{L}}^{(2)}(\bm{\beta}^0) 
\big(\hat{\bm \beta} - \bm \beta^0\big) + \Big\{-\frac{1}{n}\bm{\mathcal{L}}^{(2)}
(\bm{\beta}^0) \Big\}
\big(\hat{\bm \beta} - \bm \beta^0\big) \bm{o}_p(1),
\end{align}
where $\bm{o}_p(1)$ denotes the residual vector. Since the set of positive definite matrices is an open convex cone and the matrix $\bm{V}$ is positive definite by \textbf{Assumption 2}, we have the positive definiteness of $-\frac{1}{n}\bm{\mathcal{L}}^{(2)}
(\bm{\beta}^0)$. Taking the inverse transformation, \eqref{eq: 12} becomes
\begin{align*}
n^{\frac{1}{2}}
\big(\hat{\bm \beta} - \bm \beta^0\big)
&
= \Big\{-\frac{1}{n}\bm{\mathcal{L}}^{(2)}
(\bm{\beta}^0)
\Big\}^{-1}
\Big\{
n^{-\frac{1}{2}}
\bm{\mathcal{L}} ^{(1)}(\bm \beta^0)
\Big\}
+ \bm{o}_p(1).
\end{align*}
Using \eqref{eq: First term}, 
we show the asymptotic normality for $\bm{\hat{\beta}}$ as
\begin{equation}
n^{\frac{1}{2}} (\hat{\bm \beta} - \bm \beta^0)  \xrightarrow{\enskip d \enskip}
N\Big(\bm 0, \bm V^{-1}  \Big).
\end{equation}
\end{proof}

\textbf{Proof of Theorem 2}
\begin{proof}
The objective function is given by 
\begin{align*}
\min_{\bm \beta}Q(\bm \beta)= -{l}_c(\bm \beta) 
+ n\sum_{s\in \mathcal{E} }\rho_{\lambda}(|\theta_{s}|) 
+ n\sum_{w \in \mathcal{E} }\rho_{\lambda}(|\alpha_{w}|).
\end{align*}
Define a function $G(\bm {\Delta})=Q(\bm \beta)-Q(\bm \beta^0),$ where
$\bm{\Delta} = \bm{\beta} - \bm{\beta}^0.$ If we take a closed convex set
$\mathcal{A}$ which contains $\bm 0$ and further show that the function $G(\bm \Delta)$ is strictly positive everywhere on the boundary $\partial \mathcal{A}$, then it implies
that G has a local minimum inside the convex set $\mathcal{A}$, since $G(\bm \Delta)$ is continuous and $G(\bm{0})=\bm{0}$.
In particular, we define the closed convext set $\mathcal{A}=\{
\bm{\beta}:  
||\bm{\beta} - \bm{\beta^0}||_2 
\leq
C r_n 
\}$
and its boundary $\partial \mathcal{A} =  \{
\bm{\beta}: 
||\bm{\beta} - \bm{\beta^0}||_2 
=
C r_n 
\},
$
where $C$ is some positive constant and $r_n = \bigg\{
\frac{p_n^{1+d} \log{p_n}}{n}
\bigg\}^{\frac{1}{2}}.$
\\
By the definition, we show that
\begin{align*}
G(\bm \Delta)={l}_c(\bm \beta^0) - {l}_c(\bm \beta) + n\sum_{s\in\mathcal{E}} \Big\{\rho_{\lambda}(|\theta_{s}|) - \rho_{\lambda}(|\theta_{s}^{0}|) \Big\} + n\sum_{w\in \mathcal{E}} \Big\{\rho_{\lambda}(|\alpha_{w}|) - \rho_{\lambda}(|\alpha_{w}^{0}|)\Big\}.
\end{align*}
We start with the composite loglikelihood function and  apply Taylor's expansion to ${l}_c(\bm \beta)$ at $\bm{\beta}=\bm{\beta^0}$:
\begin{equation}
{l}_c(\bm \beta) = 
{l}_c(\bm \beta^0) + 
{l}_c^{(1)}(\bm \beta^0)^T \cdot (\bm \beta - \bm \beta^0) 
+ 
\frac{1}{2}(\bm \beta - \bm \beta^0)^T \cdot {l}_c^{(2)}(\bm \beta^{*})\cdot (\bm \beta - \bm \beta^0),
\end{equation}
where $\bm \beta^{*} = a\bm \beta +(1-a)\bm \beta^0$ for some constant $a \in (0,1).$ The difference between the two composite loglikelihood functions can be written as
\begin{equation}
\begin {split}
{l}_c(\bm \beta^0) - {l}_c(\bm \beta) 
&= 
-{l}_c^{(1)}(\bm \beta^0)^T 
\cdot 
(\bm \beta - \bm \beta^0) 
- 
\frac{1}{2}(\bm \beta - \bm \beta^0)^T 
\cdot 
{l}_c^{(2)}(\bm \beta^{*})
\cdot (\bm \beta - \bm \beta^0) 
\\
&= \underbrace{-{l}_c^{(1)}(\bm \beta^0)^T \cdot (\bm \beta - \bm \beta^0)}_{\mathcal{I}_1} + \underbrace{ \frac{1}{2}(\bm \beta - \bm \beta^0)^T \cdot E\{ - {l}_c^{(2)}(\bm \beta^{*}) \} \cdot (\bm \beta - \bm \beta^0)}_{\mathcal{I}_2} 
\\
& + \underbrace{ \frac{1}{2}(\bm \beta - \bm \beta^0)^T \cdot \bigg\{ -{l}_c^{(2)}(\bm \beta^{*}) - E\{ - {l}_c^{(2)}(\bm \beta^{*}) \}\bigg\} \cdot  (\bm \beta - \bm \beta^0)}_{\mathcal{I}_3}.
\end{split}
\end{equation}
By the assumption that $E\{ -\frac{1}{n} {l}_c^{(2)}(\bm \beta^{*}) \}$ has eigenvalues uniformly bounded away from zero and infinity when $\bm{\beta}^{*}$ is in the $\eta$-neighborhood of $\bm{\beta}^0,$ there exist $\kappa_{-}, \kappa_{+}>0$ such that $0< \kappa_{-} \leq \lambda\Big(E\{ -\frac{1}{n} {l}_c^{(2)}(\bm \beta^{*}) \}\Big) \leq \kappa_{+} < \infty$. Therefore, we have 
$
\bm{v}^T E\{ - {l}_c^{(2)}(\bm \beta^{*}) \} \bm{v} 
=\mathcal{O}_p(n) > 0,
$
where $\bm{v}$ is any unit-norm vector. Let $\bm{\beta} = \bm{\beta^0} + Cr_n\bm v$ for some constant $C > 0$. We show that
\begin{align*}
\mathcal{I}_2 & 
= \frac{1}{2}C^2r_n^2
\cdot \bm{v}^T 
E\{ - {l}_c^{(2)}(\bm \beta^{*}) \} 
\bm{v} 
=  \frac{1}{2}C^2r_n^2\cdot \mathcal{O}_p(n) > 0,
\\
\mathcal{I}_3 &= 
\frac{1}{2}C^2r_n^2\cdot \bm{v}^T 
\bigg\{ 
-{l}_c^{(2)}(\bm \beta^{*}) 
- 
E\{ - {l}_c^{(2)}(\bm \beta^{*}) \}\bigg\} 
\bm{v} 
\\
&= \frac{1}{2}C^2 r_n^2\cdot 
\sum_{r,t} v_{r} 
o_p\bigg\{
\left[
E\{ - {l}_c^{(2)}(\bm \beta^{*}) \}
\right]_{rt}
\bigg\}
v_{t},
\end{align*}
which is dominated by $\mathcal{I}_2$. For $\mathcal{I}_1$, we have
\begin{align*}
\abs{\mathcal{I}_1} &= 
\abs{-{l}_c^{(1)}(\bm \beta^0)^T \cdot (\bm \beta - \bm \beta^0)}
=
\abs{\sum_{t}\big[ -{l}_c^{(1)}(\bm \beta^0) \big]_{[t]} \cdot (\bm \beta - \bm \beta^0)_{[t]}} 
\\
&
\leq
\underbrace{
Cr_n 
\sum_{t_1:diagonal}
\abs{
\big[ -{l}_c^{(1)}(\bm \beta^0) \big]_{[t_1]} 
}
\abs{
v_{t_1}
}
}_{\abs{\mathcal{I}_{1.1}}}
+
\underbrace{
Cr_n
\sum_{t_2: \beta_{t_2}^0 \neq 0}
\abs{
\big[ -{l}_c^{(1)}(\bm \beta^0) \big]_{[t_2]} 
}
\abs{
v_{t_2}
}
}_{\abs{\mathcal{I}_{1.2}}}
\\
&
+
\underbrace{
Cr_n
\sum_{t_3: \beta_{t_3}^0 = 0}
\abs{
\big[ -{l}_c^{(1)}(\bm \beta^0) \big]_{[t_3]} 
}
\abs{
v_{t_3}
}
}_{\abs{\mathcal{I}_{1.3}}}.
\end{align*}
Also note that
\begin{align*}
G(\bm \Delta) &= 
{l}_c(\bm \beta^0) - {l}_c(\bm \beta) + n\sum_{s\in \mathcal{E}} \Big\{
\rho_{\lambda}(|\theta_{s}|) - \rho_{\lambda}(|\theta_{s}^{0}|) 
\Big\} 
+ 
n\sum_{w\in \mathcal{E}} 
\Big\{
\rho_{\lambda}(|\alpha_{w}|) - \rho_{\lambda}(|\alpha_{w}^{0}|)
\Big\} 
\\
&= \mathcal{I}_1 + \mathcal{I}_2 + \mathcal{I}_3 + \underbrace{n\sum_{s\in\mathcal{E}} 
\Big\{
\rho_{\lambda}(|\theta_{s}|) - \rho_{\lambda}(|\theta_{s}^{0}|) 
\Big\}}_{ \mathcal{I}_4} 
+ 
\underbrace{n\sum_{w\in \mathcal{E}} 
\Big\{
\rho_{\lambda}(|\alpha_{w}|) - \rho_{\lambda}(|\alpha_{w}^{0}|)
\Big\}}_{\mathcal{I}_5},
\end{align*}
where 
\begin{align*}
\mathcal{I}_4 &=  n\sum_{ \substack{ s \in\mathcal{E}:\\ \theta_s^0 = 0} } 
\Big\{
\rho_{\lambda}(|\theta_{s}|) - \rho_{\lambda}(|\theta_{s}^{0}|) 
\Big\} 
+ 
n\sum_{ \substack{ s \in\mathcal{E}:\\ \theta_s^0 \neq 0} } 
\Big\{
\rho_{\lambda}(|\theta_{s}|) - \rho_{\lambda}(|\theta_{s}^{0}|) 
\Big\} 
\\ 
&= 
\underbrace{n\sum_{ \substack{ s \in\mathcal{E}:\\ \theta_s^0 = 0} } 
\Big\{
\rho_{\lambda}(|\theta_{s}|) 
\Big\}}_{\mathcal{I}_{4.1}} 
+ 
\underbrace{n\sum_{ \substack{ s \in\mathcal{E}:\\ \theta_s^0 \neq 0} } 
\Big\{
\rho_{\lambda}(|\theta_{s}|) - \rho_{\lambda}(|\theta_{s}^{0}|) 
\Big\}}_{\mathcal{I}_{4.2}},
\\
\mathcal{I}_5 
&= 
n\sum_{ \substack{ w\in\mathcal{E}: \\ \alpha_w^0 = 0}} 
\Big\{
\rho_{\lambda}(|\alpha_{w}|) - \rho_{\lambda}(|\alpha_{w}^{0}|) 
\Big\} 
+ 
n\sum_{ \substack{ w\in\mathcal{E}: \\ \alpha_w^0 \neq 0}} 
\Big\{
\rho_{\lambda}(|\alpha_{w}|) - \rho_{\lambda}(|\alpha_{w}^{0}|) 
\Big\}
\\ 
&= \underbrace{n\sum_{ \substack{ w\in\mathcal{E}: \\ \alpha_w^0 = 0}} 
\Big\{
\rho_{\lambda}(|\alpha_{w}|) 
\Big\}}_{\mathcal{I}_{5.1}} 
+ 
\underbrace{n\sum_{ \substack{ w\in\mathcal{E}: \\ \alpha_w^0 \neq 0}} 
\Big\{
\rho_{\lambda}(|\alpha_{w}|) - \rho_{\lambda}(|\alpha_{w}^{0}|) 
\Big\}}_{\mathcal{I}_{5.2}}.
\end{align*}

We start with the second term $\abs{\mathcal{I}_{1.2}}$ as follows.
Since the \# of nonzero true off-diagonal parameters is $q_n$ and 
$\max_{s} \abs{
\frac{ \partial -l_c(\bm \beta^0) }
{\partial \theta_{s}}
} = 
\max_{s} \abs{
\frac{ \partial -l_c(\bm \beta^0) }
{\partial \alpha_{s}}
} = \mathcal{O}_p 
\Big\{
(n s_n^2 \log{p_n})^{\frac{1}{2}}
\Big\}.$ By Cauchy–Schwarz inequality, $\abs{\mathcal{I}_{1.2}}$ becomes
\begin{align*}
\abs{\mathcal{I}_{1.2}} 
=
C r_n 
\sum_{t_2: \beta_{t_2}^0 \neq 0}
\abs{
\big[ -{l}_c^{(1)}(\bm \beta^0) \big]_{[t_2]} 
}
\abs{
v_{t_2}
}
&
\leq
C r_n 
\sqrt{
\sum_{t_2: \beta_{t_2}^0 \neq 0}
\big[ -{l}_c^{(1)}(\bm \beta^0) \big]_{[t_2]}^2 
}
\sqrt{\
\sum_{t_2: \beta_{t_2}^0 \neq 0} [v_{t_2}]^2
}
\\
&
\leq
C r_n \mathcal{O}_p
\Big\{
(q_n n s_n^2 \log{p_n})^{\frac{1}{2}}
\Big\}.
\end{align*}
Divide $\abs{\mathcal{I}_{1.2}}$ by $\mathcal{I}_2$, we show that
\begin{align*}
\frac{\abs{\mathcal{I}_{1.2}}}{\mathcal{I}_2}
\leq
\frac{C r_n \mathcal{O}_p
\Big\{
(q_n n s_n^2 \log{p_n})^{\frac{1}{2}}
\Big\}}
{\frac{1}{2}C^2 r_n^2 n \kappa_{-}}
=
\frac{2 \mathcal{O}_p
\Big\{
(q_n  s_n^2 \log{p_n})^{\frac{1}{2}}
\Big\} }{Cr_n n^{\frac{1}{2}} \kappa_{-}} .
\end{align*}
Since $r_n =
\Big(
\frac{p_n^{1+d} \log{p_n}}{n}
\Big)^{\frac{1}{2}}
$
and $q_n^{\frac{1}{2}}s_n = \mathcal{O}(p_n^{\frac{1+d}{2}}),
\abs{\mathcal{I}_{1.2}} 
$
can be dominated by $\mathcal{I}_2$ with an appropriate $C>0.$ \\
\raggedright \textbf{Third Term} $\abs{\mathcal{I}_{1.3}}:$ \justifying 
by \textbf{Lemma} \ref{lemma 2.2}, we show that 
\begin{align*}
& \abs{\mathcal{I}_{1.3}} =
\sum_{t_3: \bm \beta_{t_3}^0 = 0}
\abs{
\big[ -{l}_c^{(1)}(\bm \beta^0) \big]_{[t_3]} 
}
\cdot 
\abs{
(\bm \beta - \bm \beta^0)_{[t_3]} 
}
=
Cr_n
\sum_{t_3: \beta_{t_3}^0 = 0}
\abs{
\big[ -{l}_c^{(1)}(\bm \beta^0) \big]_{[t_3]} 
}
\abs{
v_{t_3}
}
\\
&
=
C r_n 
\sum_{t_3: \beta_{t_3}^0 = 0}
\mathcal{O}_p
\Big\{
(n s_n^2 \log{p_n} )^{\frac{1}{2}}
\Big\}
\abs{v_{t_3}},
\\
&
\mathcal{I}_{4.1} + \mathcal{I}_{5.1} 
= 
n \lambda \sum_{s: \theta_s^0 = 0} |\theta_{s}| 
+
n \lambda \sum_{w: \alpha_w^0 = 0} |\alpha_{w}|  
\\
&
=
n \lambda \sum_{t_{3}: \theta_s^0 = 0, \alpha_{w}^0 = 0} C r_n \abs{v_{t_3}}
=
C r_n
\sum_{t_3: \theta_{s}^0 = 0, \alpha_w^0 = 0}
n \lambda
\abs{v_{t_3}}
 > 0.
\end{align*}
We choose $\lambda \geq \delta_1 \bigg(
\frac{ s_n^2\log{p_n}}{n}
\bigg)^{\frac{1}{2}}$ for some constant $\delta_1 > 0$ such that $  \abs{\mathcal{I}_{1.3}} \leq \mathcal{I}_{4.1} + \mathcal{I}_{5.1}.$ 
\\
\textbf{First Term} $\abs{\mathcal{I}_{1.1}}:$ By \textbf{Lemma} \ref{lemma 2.2} and Cauchy-Schwarz inequality, the upper bound is
\begin{align*}
\abs{\mathcal{I}_{1.1}} 
&
=
Cr_n 
\sum_{t_1:diagonal}
\abs{
\big[ -{l}_c^{(1)}(\bm \beta^0) \big]_{[t_1]} 
}
\abs{
v_{t_1}
}
\leq
C r_n \sqrt{
\sum_{t_1}
\big[ -{l}_c^{(1)}(\bm \beta^0) \big]_{[t_1]}^2 
}
\sqrt{
\sum_{t_1} v_{t_1}^2
}
\\
&
\leq
C r_n \mathcal{O}_p\Big\{
(n p_n s_n^4 \log{p_n})^{\frac{1}{2}}
\Big\} .
\end{align*}
Divide $\abs{\mathcal{I}_{1.1}}$ by $\mathcal{I}_2,$ we have
\begin{align*}
\frac{\abs{\mathcal{I}_{1.1}}}{\mathcal{I}_2}
\leq
\frac{ C r_n \mathcal{O}_p\Big\{
(n p_n s_n^4 \log{p_n})^{\frac{1}{2}}
\Big\}}
{
C^2 \kappa_{-} r_n^2 n
}
=
\frac{\mathcal{O}_p\Big\{
(p_n s_n^4 \log{p_n})^{\frac{1}{2}}
\Big\} }
{C\kappa_{-}r_n n^{\frac{1}{2}} }.
\end{align*}
If $r_n = \Big(
\frac{p_n^{1+d} \log{p_n} }{n}
\Big)^{\frac{1}{2}},$ then
\begin{align*}
\frac{\abs{\mathcal{I}_{1.1}}}{\mathcal{I}_2}
\leq
\frac{ \mathcal{O}_p(
s_n^2
)}
{C \kappa_{-} p_n^{d/2}}.
\end{align*}
$\mathcal{I}_{2}$ can dominate  $\abs{\mathcal{I}_{1.1}}$ with an appropriate constant $C>0$ if  $s_n^4 = \mathcal{O}( p_n^d )$.
\\
Recall that
\begin{align*}
G(\bm \Delta) &= 
{l}_c(\bm \beta^0) - {l}_c(\bm \beta) + n\sum_{s\in\mathcal{E}} \Big\{
\rho_{\lambda}(|\theta_{s}|) - \rho_{\lambda}(|\theta_{s}^{0}|) 
\Big\} 
+ 
n\sum_{w\in \mathcal{E}} 
\Big\{
\rho_{\lambda}(|\alpha_{w}|) - \rho_{\lambda}(|\alpha_{w}^{0}|)
\Big\} 
\\
&= \mathcal{I}_1 + \mathcal{I}_2 + \mathcal{I}_3 + \underbrace{n\sum_{s\in\mathcal{E}} 
\Big\{
\rho_{\lambda}(|\theta_{s}|) - \rho_{\lambda}(|\theta_{s}^{0}|) 
\Big\}}_{ \mathcal{I}_4} 
+ 
\underbrace{n\sum_{w\in \mathcal{E}} 
\Big\{
\rho_{\lambda}(|\alpha_{w}|) - \rho_{\lambda}(|\alpha_{w}^{0}|)
\Big\}}_{\mathcal{I}_5}.
\end{align*}
We further decompose the penalty terms $\mathcal{I}_{4}$ and $\mathcal{I}_{5},$
\begin{equation}
\begin{split}
\mathcal{I}_4 &=  n\sum_{ \substack{ s \in \mathcal{E}:\\ \theta_s^0 = 0} } 
\Big\{
\rho_{\lambda}(|\theta_{s}|) - \rho_{\lambda}(|\theta_{s}^{0}|) 
\Big\} 
+ 
n\sum_{ \substack{ s \in \mathcal{E}:\\ \theta_s^0 \neq 0} } 
\Big\{
\rho_{\lambda}(|\theta_{s}|) - \rho_{\lambda}(|\theta_{s}^{0}|) 
\Big\} 
\\ 
&= 
\underbrace{n\sum_{ \substack{ s \in \mathcal{E}:\\ \theta_s^0 = 0} } 
\Big\{
\rho_{\lambda}(|\theta_{s}|) 
\Big\}}_{\mathcal{I}_{4.1}} 
+ 
\underbrace{n\sum_{ \substack{ s \in\mathcal{E}:\\ \theta_s^0 \neq 0} } 
\Big\{
\rho_{\lambda}(|\theta_{s}|) - \rho_{\lambda}(|\theta_{s}^{0}|) 
\Big\}}_{\mathcal{I}_{4.2}},
\\
\mathcal{I}_5 
&= 
n\sum_{ \substack{ w\in \mathcal{E}: \\ \alpha_w^0 = 0}} 
\Big\{
\rho_{\lambda}(|\alpha_{w}|) - \rho_{\lambda}(|\alpha_{w}^{0}|) 
\Big\} 
+ 
n\sum_{ \substack{ w\in\mathcal{E}: \\ \alpha_w^0 \neq 0}} 
\Big\{
\rho_{\lambda}(|\alpha_{w}|) - \rho_{\lambda}(|\alpha_{w}^{0}|) 
\Big\}
\\ 
&= \underbrace{n\sum_{ \substack{ w\in\mathcal{E}: \\ \alpha_w^0 = 0}} 
\Big\{
\rho_{\lambda}(|\alpha_{w}|) 
\Big\}}_{\mathcal{I}_{5.1}} 
+ 
\underbrace{n\sum_{ \substack{ w\in\mathcal{E}: \\ \alpha_w^0 \neq 0}} 
\Big\{
\rho_{\lambda}(|\alpha_{w}|) - \rho_{\lambda}(|\alpha_{w}^{0}|) 
\Big\}}_{\mathcal{I}_{5.2}}.
\end{split}
\end{equation}
To begin with, we show that
\begin{align*}
\rho_{\lambda}(|\theta_{s}|) -
\rho_{\lambda}(|\theta_{s}^0|)
=\begin{cases}
\lambda (\theta_s - \theta_{s}^0), & \text{if } \theta_s \geq 0, \theta_{s}^0 >0; \\
-\lambda (\theta_s - \theta_{s}^0), & \text{if } \theta_s \leq 0, \theta_{s}^0 < 0; \\
-\lambda \theta_s - \lambda \theta_{s}^0, & \text{if } \theta_s < 0, \theta_{s}^0 >0; \\
\lambda \theta_s + \lambda \theta_{s}^0, & \text{if } \theta_s > 0, \theta_{s}^0 < 0;
\end{cases} 
\end{align*}
and 
\begin{align*}
&
\text{when } \theta_s < 0, \theta_{s}^0 >0,\,
\abs{ -\lambda \theta_s - \lambda \theta_{s}^0} \leq
\abs{ \lambda (\theta_s - \theta_{s}^0)} , 
\\
&
\text{when } \theta_s > 0, \theta_{s}^0 < 0,\,
\abs{\lambda \theta_s + \lambda \theta_{s}^0} \leq
\abs{ \lambda (\theta_s - \theta_{s}^0)}.
\end{align*}
As a result, we can conclude that
\begin{align*}
\abs{
\mathcal{I}_{4.2}
}
&
= \Big| n\sum_{\substack{ s \in\mathcal{E}:\\ \theta_s^0 \neq 0}} 
\Big\{
\rho_{\lambda}(|\theta_{s}|) - \rho_{\lambda}(|\theta_{s}^{0}|) 
\Big\}
\Big |
\leq
n 
\sum_{\substack{ s \in \mathcal{E}:\\ \theta_s^0 \neq 0}} 
\abs{
\rho_{\lambda}(|\theta_{s}|) - \rho_{\lambda}(|\theta_{s}^{0}|) 
}
\\
&
\leq
n 
\sum_{\substack{ s \in\mathcal{E}:\\ \theta_s^0 \neq 0}} 
\abs{
\lambda (\theta_s - \theta_{s}^0)
}
=
n \lambda \sum_{ \substack{ s \in\mathcal{E}:\\ \theta_s^0 \neq 0} } 
C r_n \abs{v_s}
\leq
n \lambda C r_n \sum_{ \substack{ s \in \mathcal{E}:\\ \theta_s^0 \neq 0} } \abs{v_s} \cdot {1},
\\
\abs{
\mathcal{I}_{5.2}
}
&
\leq
n 
\sum_{\substack{ w \in \mathcal{E}:\\ \alpha_w^0 \neq 0}} 
\abs{
\rho_{\lambda}(|\alpha_{w}|) - \rho_{\lambda}(|\alpha_{w}^{0}|) 
}
\leq
n \lambda C r_n \sum_{ \substack{ w \in\mathcal{E}:\\ \alpha_w^0 \neq 0} } \abs{v_w} \cdot {1}.
\end{align*}
Since the \# of terms included in the summation $\big( 
\sum_{\substack{ s \in\mathcal{E}:\\ \theta_s^0 \neq 0}} 
+ \sum_{\substack{ w \in \mathcal{E}:\\ \alpha_w^0 \neq 0}} 
\big)$ is $q_n$, by Cauchy-Schwarz inequality, the upper bound of $\abs{\mathcal{I}_{4.2}}+ \abs{\mathcal{I}_{5.2}}$ can be given as
\begin{equation}
\begin{split}
\abs{\mathcal{I}_{4.2}} 
+
\abs{\mathcal{I}_{5.2}}
&
\leq 
n \lambda C r_n 
\bigg(
\sum_{ \substack{ s \in \mathcal{E}:\\ \theta_s^0 \neq 0} } \abs{v_s} \cdot {1}
+
\sum_{ \substack{ w \in \mathcal{E}:\\ \alpha_w^0 \neq 0} } \abs{v_w} \cdot {1}
\bigg)
\\
&
\leq
n \lambda C r_n
\sqrt{
\sum_{s,w}
(v)^2
}
\sqrt{
\sum_{s,w}
(1)^2
}
\\
& \leq  
n \lambda C r_n  \cdot (q_n)^{\frac{1}{2}}.
\end{split}
\end{equation}
We divide it by $\mathcal{I}_2$
\begin{align*}
\frac{\abs{\mathcal{I}_{4.2}} 
+
\abs{\mathcal{I}_{5.2}}}{\mathcal{I}_2}
\leq
\frac{2n\lambda C r_n q_n^{\frac{1}{2}}}
{
 C^2 r_n^2 n \kappa_{-}
}
=
\frac{2\lambda q_n^{\frac{1}{2}}}
{Cr_n \kappa_{-}}
\leq
\frac{ 2 \delta_2 } {C \kappa_{-} }
\bigg(
\frac{s_n^3 q_n}{p_n^{1+d} }
\bigg)^{\frac{1}{2}}.
\end{align*}
The numerator can be dominated by $\mathcal{I}_2$ with an appropriate constant $C>0$
if $\lambda \leq \delta_2 \Big(\frac{s_n^3 \log{p_n}}{n}\Big)^{\frac{1}{2}}$ for some constant $\delta_2 > 0$.

In summary, we have shown that $G(\bm \Delta)  > 0$ on the boundary $\partial \mathcal{A}$
by choosing $C>0$ sufficiently large and $\delta_1 \Big(\frac{s_n^2 \log{p_n}}{n}\Big)^{\frac{1}{2}} \leq \lambda \leq \delta_2 \Big(\frac{s_n^3 \log{p_n}}{n}\Big)^{\frac{1}{2}}$ for some constants $\delta_1, \delta_2 >0$,
which completes the proof.
\end{proof}

\textbf{Proof of Theorem 3}
\begin{proof}
Since the property of sign consistency is regarding the edge parameters, 
we treat diagonal parameters as known in this theorem. Recall that the composite likelihood estimator $\bm{\hat{\beta}}$ is obtained by
\begin{align} \label{eq:Thm3-1}
\bm{\hat{\beta}} =
\argmin_{\bm{\beta}}
Q(\bm{\beta}) = 
\argmin_{\bm{\beta}}
\big\{
-{l}_c(\bm {\beta}) + n\sum_{s\in \mathcal{E}}\rho_{\lambda}(|\theta_s|) + n\sum_{w\in \mathcal{E}}\rho_{\lambda}(|\alpha_w|) 
\big\}.
\end{align}
Here $\bm{\beta}$ excludes diagonal parameters. Similar to Ravikumar and Yu (2008), we introduce another estimator
\begin{align} \label{eq:Thm3-2}
\bm{\tilde{\beta}} =
\argmin_{\bm{\beta}}
Q(\bm{\beta}) = 
\argmin_{\bm{\beta} = (\bm{\beta}_{\mathcal{S}_1}, \bm{0} )}
\big\{
-{l}_c(\bm {\beta}) + n\sum_{s\in \mathcal{E}}\rho_{\lambda}(|\theta_s|) + n\sum_{w\in \mathcal{E}}\rho_{\lambda}(|\alpha_w|) 
\big\},
\end{align}
where $\bm{\tilde{\beta}} = (\bm{\tilde{\beta}}_{\mathcal{S}_1},\bm{\tilde{\beta}}_{\mathcal{S}^c})$ with the constraint $\bm{\tilde{\beta}}_{\mathcal{S}^c} = \bm{0}.$ The set of first derivative equations  $ \frac{\partial Q(\bm{\beta})}{\partial \bm{\beta}} = 0$ can be partitioned into two sets of equations
\begin{align}
\label{condition 1}
&
\frac{1}{n}l^{(1)}_c(\bm{\tilde{\beta}})_{\mathcal{S}_1} = \lambda \tilde{z}_{\mathcal{S}_1} ,
\\
\label{condition 2}
&
\frac{1}{n}l^{(1)}_c(\bm{\tilde{\beta}})_{\mathcal{S}^c} = \lambda \tilde{z}_{\mathcal{S}^c},
\end{align}
where the sub-differential of the absolute value function is defined as
\begin{align*}
\tilde{z}
= \frac{\partial \abs{x}}{\partial x} = 
\begin{cases}
1 & \text{if } x > 0,
\\
[-1, 1] & \text{if } x = 0,
\\
-1 & \text{if } x < 0.
\end{cases}
\end{align*}
If the estimator $\bm{\tilde{\beta}}$ satisfies conditions \eqref{condition 1} and \eqref{condition 2}, then $\bm{\tilde{\beta}}$ is the local minimizer $\bm {\hat{\beta}}$ of the objective function $Q(\bm{\beta})$. We show the proof by 2 steps. \\
\textbf{Step 1:} Verify $\| \tilde{z}_{\mathcal{S}^c} \|_{\infty} < 1.$ \\
By Taylor's expansion, we have
\begin{align}
\frac{1}{n}l^{(1)}_c(\bm{\tilde{\beta}})
=
\frac{1}{n} l^{(1)}_c(\bm{\beta}^0)
+
\frac{1}{n} l^{(2)}_c(\bm{\beta}^0) (\bm{\tilde{\beta}} - \bm{\beta}^0)
+
\underbrace{
\frac{1}{n}\big\{
 l^{(2)}_c(\bm{\beta}^{\ast})
 -
 l^{(2)}_c(\bm{\beta}^0)
\big\}
(\bm{\tilde{\beta}} - \bm{\beta}^0)
}_{R}  ,
\end{align}
where $\bm{\beta}^{\ast} = a \bm{\beta}^0 + (1-a)\bm{\tilde{\beta}}$ for some constant $a \in [0,1].$ In block format, we can write
\begin{align*}
\begin{pmatrix}
\frac{1}{n}l^{(1)}_c(\bm{\beta}^0)_{\mathcal{S}_1}  \\
 \frac{1}{n}l^{(1)}_c(\bm{\beta^0})_{\mathcal{S}^c}
\end{pmatrix}
+ \frac{1}{n}
\begin{pmatrix}
l^{(2)}_c(\bm{\beta}^0)_{\mathcal{S}_1 \mathcal{S}_1} &
l^{(2)}_c(\bm{\beta}^0)_{\mathcal{S}_1 \mathcal{S}^c} \\
l^{(2)}_c(\bm{\beta}^0)_{\mathcal{S}^c \mathcal{S}_1} &
l^{(2)}_c(\bm{\beta}^0)_{\mathcal{S}^c \mathcal{S}^c}     
\end{pmatrix}
\begin{pmatrix}
\big( \bm{\tilde{\beta}} - \bm{\beta}^0 \big)_{\mathcal{S}_1} \\
\bm{0} \\
\end{pmatrix}
+
\begin{pmatrix}
R_{\mathcal{S}_1} \\
R_{\mathcal{S}^c}
\end{pmatrix}
=
\lambda 
\begin{pmatrix}
\tilde{z}_{\mathcal{S}_1} \\
\tilde{z}_{\mathcal{S}^c}
\end{pmatrix},
\end{align*}
where $R_{\mathcal{S}_1},
R_{\mathcal{S}^c}$ are the partition of the residual vector. From the block format, we know that
\begin{equation} 
\begin{split} \label{eq:Thm3-3}
&
\big( \bm{\tilde{\beta}} - \bm{\beta}^0 \big)_{\mathcal{S}_1} 
=
\bigg(
\frac{1}{n} l^{(2)}_c (\bm{\beta}^0)_{\mathcal{S}_1\mathcal{S}_1}
\bigg)^{-1}
\bigg(
\lambda \tilde{z}_{\mathcal{S}_1}
-
R_{\mathcal{S}_1}
-
\frac{1}{n}l^{(1)}_c(\bm{\beta}^0)_{\mathcal{S}_1}
\bigg), 
\end{split}
\end{equation}
\begin{equation} \label{eq:Thm3-4}
\begin{split}
&
\tilde{z}_{\mathcal{S}^c}
=
\frac{1}{\lambda}
\bigg(
\frac{1}{n} l^{(1)}_c(\bm{\beta}^0)_{\mathcal{S}^c}
+
\frac{1}{n} l^{(2)}_c(\bm{\beta}^0)_{\mathcal{S}^c \mathcal{S}_1}
\big(
\bm{\tilde{\beta}} - \bm{\beta}^0
\big)_{\mathcal{S}_1}
+
R_{\mathcal{S}^c}
\bigg) .
\end{split}
\end{equation}
By equations \eqref{eq:Thm3-3} and \eqref{eq:Thm3-4}, we have
\begin{align*}
\tilde{z}_{\mathcal{S}^c}
=
&
\frac{1}{\lambda}
\Bigg\{
\frac{1}{n} l^{(1)}_c(\bm{\beta}^0)_{\mathcal{S}^c}
+
\frac{1}{n} l^{(2)}_c(\bm{\beta}^0)_{\mathcal{S}^c \mathcal{S}_1}
\bigg(
\frac{1}{n} l^{(2)}_c (\bm{\beta}^0)_{\mathcal{S}_1\mathcal{S}_1}
\bigg)^{-1}
\bigg(
\lambda \tilde{z}_{\mathcal{S}_1}
-
R_{\mathcal{S}_1}
-
\frac{1}{n}l^{(1)}_c(\bm{\beta}^0)_{\mathcal{S}_1}
\bigg)
+
R_{\mathcal{S}^c}
\Bigg\}
\\
= &
\underbrace{
\frac{1}{\lambda}
\Bigg\{
\frac{1}{n} l^{(1)}_c(\bm{\beta}^0)_{\mathcal{S}^c}
-
\frac{1}{n} l^{(2)}_c(\bm{\beta}^0)_{\mathcal{S}^c \mathcal{S}_1}
\bigg(
\frac{1}{n} l^{(2)}_c (\bm{\beta}^0)_{\mathcal{S}_1\mathcal{S}_1}
\bigg)^{-1}
\bigg(
\frac{1}{n}l^{(1)}_c(\bm{\beta}^0)_{\mathcal{S}_1}
\bigg)
\Bigg\}
}_{\mathcal{I}_1}
\\
+
&
\underbrace{
\frac{1}{\lambda}
\Bigg\{
R_{\mathcal{S}^c}
-
\frac{1}{n} l^{(2)}_c(\bm{\beta}^0)_{\mathcal{S}^c \mathcal{S}_1}
\bigg(
\frac{1}{n} l^{(2)}_c (\bm{\beta}^0)_{\mathcal{S}_1\mathcal{S}_1}
\bigg)^{-1}
R_{\mathcal{S}_1}
\Bigg\}
}_{\mathcal{I}_2}
+
\underbrace{
\frac{1}{\lambda}
\Bigg\{
\frac{1}{n} l^{(2)}_c(\bm{\beta}^0)_{\mathcal{S}^c \mathcal{S}_1}
\bigg(
\frac{1}{n} l^{(2)}_c (\bm{\beta}^0)_{\mathcal{S}_1\mathcal{S}_1}
\bigg)^{-1}
\lambda \tilde{z}_{\mathcal{S}_1}
\Bigg\}
}_{\mathcal{I}_3}.
\end{align*}
Since
$
\| \tilde{z}_{\mathcal{S}^c} \|_{\infty}
=
\|
\mathcal{I}_1 +
\mathcal{I}_2 +
\mathcal{I}_3
\|_{\infty}
\leq 
\| \mathcal{I}_1 \|_{\infty}
+
\| \mathcal{I}_2 \|_{\infty}
+
\| \mathcal{I}_3 \|_{\infty},
$
we start with the upper bound of $\| \mathcal{I}_2 \|_{\infty}$. By Taylor's expansion, we show that each element of $R$ can be written as
\begin{align*}
(R)_{u}
=
\frac{1}{n} 
( \bm{\beta^{\ast}} - \bm{\beta^0})^{T}
\frac{\partial^{3} l_c(\bm \beta^{\ast\ast}) }
{\partial \bm{\beta}
 \partial \bm{\beta}^T
 \partial \beta_{u}
}
(
\tilde{\bm{\beta}} - \bm{\beta^0}
),
\end{align*}
where $\beta_u$ denotes any off-diagonal parameter and 
$\bm{\beta}^{\ast\ast} = b \bm{\beta}^0 + (1-b)\bm{\beta^{\ast}}$ for some constant $b \in [0, 1]$. Also note that $( \bm{\beta^{\ast}} - \bm{\beta^0})^{T} = (1-a) (\tilde{\bm{\beta}} - \bm{\beta}^0)^T$, we can rewrite it as
\begin{align*}
(R)_{u}
&
=
\frac{1}{n} (1-a)
( \tilde{\bm{\beta}} - \bm{\beta^0})^{T}
\frac{\partial^{3} l_c(\bm \beta^{\ast\ast}) }
{\partial \bm{\beta}
 \partial \bm{\beta}^T
 \partial \beta_{u}
}
(
\tilde{\bm{\beta}} - \bm{\beta^0}
) 
\\
&
=
\underbrace{
\frac{1}{n} (1-a)
( \tilde{\bm{\beta}} - \bm{\beta^0})^{T}
\bigg\{
\frac{\partial^{3} l_c(\bm \beta^{\ast\ast}) }
{\partial \bm{\beta}
 \partial \bm{\beta}^T
 \partial \beta_{u}
}
-
E\bigg\{
\frac{\partial^{3} l_c(\bm \beta^{\ast\ast}) }
{\partial \bm{\beta}
 \partial \bm{\beta}^T
 \partial \beta_{u}
}
\bigg\}
\bigg\}
(
\tilde{\bm{\beta}} - \bm{\beta^0}
) 
}_{\mathcal{M}_1}
\\
&
-
\underbrace{
\frac{1}{n} (1-a)
( \tilde{\bm{\beta}} - \bm{\beta^0})^{T}
E\bigg\{
\frac{\partial^{3} -l_c(\bm \beta^{\ast\ast}) }
{\partial \bm{\beta}
 \partial \bm{\beta}^T
 \partial \beta_{u}
}
\bigg\}
(
\tilde{\bm{\beta}} - \bm{\beta^0}
)
}_{\mathcal{M}_2}.
\end{align*}
Following the same methodology as in \textbf{Theorem 2}, we can show $\Vert \tilde{\bm{\beta}} - \bm{\beta}^{0} \rVert_2 = \mathcal{O}_p\Big\{ \Big(\frac{p_n^{1+d}\log{p_n}}{n}\Big)^{\frac{1}{2}}\Big\}$. Thus, we have $\mathcal{M}_2
\leq
\frac{1}{n} (1-a) \mathcal{O}_p\Big\{
\frac{p_n^{1+d}\log{p_n}}{n}
\Big\}
\bm{v}^{T}
E\Big\{
\frac{\partial^{3} -l_c(\bm \beta^{\ast\ast}) }
{\partial \bm{\beta}
 \partial \bm{\beta}^T
 \partial \beta_{u}
}
\Big\}
\bm{v}$ and 
\begin{align*}
\big\| (R)_u \big\|_{\infty} 
&
=
\big\| \mathcal{M}_1 - \mathcal{M}_2 \big\|_{\infty} 
\leq
\big\| \mathcal{M}_1 \big\|_{\infty} 
+
\big\| \mathcal{M}_2 \big\|_{\infty} ,
\\
&
\leq
\big\| \mathcal{M}_2 \big\|_{\infty}  \big(1+o_p(1)\big)
\leq
\frac{1}{n}
\mathcal{O}_p\Big\{
\frac{p_n^{1+d}\log{p_n}}{n}
\Big\}
\vertiii{
E\bigg\{
\frac{\partial^{3} -l_c(\bm \beta^{\ast\ast}) }
{\partial \bm{\beta}
 \partial \bm{\beta}^T
 \partial \beta_{u}
}
\bigg\}
}_{2}\big( 1 + o_p(1) \big).
\end{align*}
By \textbf{Assumptions 3-4}, we 
have $\vertiii{E\Big\{
\frac{\partial^{3} -l_c(\bm \beta^{\ast\ast}) }
{\partial \bm{\beta}
 \partial \bm{\beta}^T
 \partial \beta_{u}
}
\Big\}
}_{2} \leq n \mathcal{W}^{\ast}$ for some constant $\mathcal{W}^{\ast} > 0$ and show that
\begin{align*}
&
\big\| (R)_u \big\|_{\infty} 
\leq
\mathcal{O}_p\Big\{
\frac{p_n^{1+d}\log{p_n}}{n}
\Big\}
\frac{1}{n}
n \mathcal{W}^{\ast} \big( 1 + o_p(1) \big)
\\
&
\leq
\mathcal{O}_p\Big\{
\frac{p_n^{1+d}\log{p_n}}{n}
\Big\}
\big( 1 + o_p(1) \big)
= o_p(1).
\end{align*}
Therefore, by \textbf{Lemma} \ref{lemma 2.4}, 
the upper bound of $\| \mathcal{I}_2 \|_{\infty}$ can be given as
\begin{align*}
\| \mathcal{I}_2 \|_{\infty} &= 
\frac{1}{\lambda}
\vertiii{
R_{\mathcal{S}^c}
-
\frac{1}{n} l^{(2)}_c(\bm{\beta}^0)_{\mathcal{S}^c \mathcal{S}_1}
\bigg(
\frac{1}{n} l^{(2)}_c (\bm{\beta}^0)_{\mathcal{S}_1\mathcal{S}_1}
\bigg)^{-1}
R_{\mathcal{S}_1}
}_{\infty}
\\
&
\leq
\frac{1}{\lambda}
\|
R_{\mathcal{S}^c}
\|_{\infty}
+
\frac{1}{\lambda}
\vertiii{
\frac{1}{n} l^{(2)}_c(\bm{\beta}^0)_{\mathcal{S}^c \mathcal{S}_1}
\bigg(
\frac{1}{n} l^{(2)}_c (\bm{\beta}^0)_{\mathcal{S}_1\mathcal{S}_1}
\bigg)^{-1}
}_{\infty}
\|
R_{\mathcal{S}_1}
\|_{\infty}
\\
&
\leq
\frac{1}{\delta_1}
\bigg(
\frac{n}{s_n^2 \log{p_n}}
\bigg)^{\frac{1}{2}}
\mathcal{O}_p \Big\{
\frac{p_n^{1+d} \log{p_n} }{n}
\Big\}
\big( 1+o_p(1) \big)
\\
&
+
\frac{1}{\delta_1}
\bigg(
\frac{n}{s_n^2 \log{p_n}}
\bigg)^{\frac{1}{2}}
(1-\frac{1}{3}\xi^2)
\mathcal{O}_p \Big\{
\frac{p_n^{1+d} \log{p_n} }{n}
\Big\} \big( 1+o_p(1) \big)
\\
&
\leq 
\mathcal{O}_p
\bigg\{
 p_n^{1+d}
\bigg(
\frac{ \log{p_n} }{ns_n^2}
\bigg)^{\frac{1}{2}}
\bigg\}
= o_p(1).
\end{align*}
Similarly, by \textbf{Lemma} \ref{lemma 2.2} and \textbf{Lemma} \ref{lemma 2.4}, the rest two terms can be bounded by
\begin{align*}
\|
\mathcal{I}_1 
\|_{\infty} &= 
\frac{1}{\lambda}
\Bigg \|
\frac{1}{n} l^{(1)}_c(\bm{\beta}^0)_{\mathcal{S}^c}
-
\frac{1}{n} l^{(2)}_c(\bm{\beta}^0)_{\mathcal{S}^c \mathcal{S}_1}
\bigg(
\frac{1}{n} l^{(2)}_c (\bm{\beta}^0)_{\mathcal{S}_1\mathcal{S}_1}
\bigg)^{-1}
\bigg(
\frac{1}{n}l^{(1)}_c(\bm{\beta}^0)_{\mathcal{S}_1}
\bigg)
\Bigg \|_{\infty}
\\
&
\leq
\frac{1}{\lambda}
\Bigg\{
\Big \|
\frac{1}{n} l^{(1)}_c(\bm{\beta}^0)_{\mathcal{S}^c}
\Big \|_{\infty}
+
\vertiii{
\frac{1}{n} l^{(2)}_c(\bm{\beta}^0)_{\mathcal{S}^c \mathcal{S}_1}
\bigg(
\frac{1}{n} l^{(2)}_c (\bm{\beta}^0)_{\mathcal{S}_1\mathcal{S}_1}
\bigg)^{-1}
}_{\infty}
\Big \|
\frac{1}{n}l^{(1)}_c(\bm{\beta}^0)_{\mathcal{S}_1}
\Big \|_{\infty}
\Bigg\}
\\
&
\leq
\frac{1}{\lambda}
\Big\{
\epsilon + (1-\frac{1}{3}\xi^2)
\epsilon
\Big\}
\leq
\frac{1}{\lambda}
(2- \frac{1}{3}\xi^2) \epsilon ,
\\
\|
\mathcal{I}_3
\|_{\infty} &= 
\frac{1}{\lambda}
\Bigg \|
\frac{1}{n} l^{(2)}_c(\bm{\beta}^0)_{\mathcal{S}^c \mathcal{S}_1}
\bigg(
\frac{1}{n} l^{(2)}_c (\bm{\beta}^0)_{\mathcal{S}_1\mathcal{S}_1}
\bigg)^{-1}
\lambda \tilde{z}_{\mathcal{S}_1}
\Bigg \|_{\infty} 
\leq
\vertiii{
\frac{1}{n} l^{(2)}_c(\bm{\beta}^0)_{\mathcal{S}^c \mathcal{S}_1}
\bigg(
\frac{1}{n} l^{(2)}_c (\bm{\beta}^0)_{\mathcal{S}_1\mathcal{S}_1}
\bigg)^{-1}
}_{\infty}
\\
&
\leq
1- \frac{1}{3}\xi^2.
\end{align*}
Therefore, we choose $\epsilon = \frac{1}{6}\xi^2 \lambda
$ and
\begin{align*}
\| \tilde{z}_{\mathcal{S}^c} \|_{\infty}
\leq 
\| \mathcal{I}_1 \|_{\infty}
+
\| \mathcal{I}_2 \|_{\infty}
+
\| \mathcal{I}_3 \|_{\infty}
&
\leq
\frac{1}{\lambda} (2 - \frac{1}{3}\xi^2) \epsilon
+o(1) + 1 - \frac{1}{3} \xi^2
\\
&
\leq
\frac{1}{3} \xi^2 - \frac{1}{18} \xi^4 + 1 - \frac{1}{3} \xi^2
+ o(1) < 1.
\end{align*}
Once the first step is completed, it demonstrates that the solution $\bm{\tilde{\beta}}$ to the constrained problem \eqref{eq:Thm3-2} is equivalent to the solution $\bm{\hat{\beta}}$
to the original unrestricted problem \eqref{eq:Thm3-1}. \\
\textbf{Step 2:} Verify $sign(\hat{\bm{\beta}}_{\mathcal{S}_1}) = sign(\bm{\beta}^{0}_{\mathcal{S}_1})$. From equation \eqref{eq:Thm3-3} and \textbf{Lemma} \ref{lemma 2.2}, \textbf{Lemma} \ref{lemma 2.8},
we show that for some constant $c^{\ast} = \frac{1 + \xi^2/6}{\tau_{-}} \big(  1+o_p(1) \big)$,
\begin{align*}
&
\Big\|
\big(
\bm{\tilde{\beta}} - \bm{\beta}^0 
\big)_{\mathcal{S}_1}
\Big\|_{\infty}
\leq
\vertiii{
\bigg(
-
\frac{1}{n} l^{(2)}_c (\bm{\beta}^0)_{\mathcal{S}_1\mathcal{S}_1}
\bigg)^{-1}
}_{\infty}
\Bigg\|
\Big(
-
\lambda \tilde{z}_{\mathcal{S}_1}
+
R_{\mathcal{S}_1}
+
\frac{1}{n}l^{(1)}_c(\bm{\beta}^0)_{\mathcal{S}_1}
\Big)
\Bigg\|_{\infty}
\\
&
\leq
\sqrt{
\abs{\mathcal{S}_1}
}
\vertiii{\bigg(
-
\frac{1}{n} l^{(2)}_c (\bm{\beta}^0)_{\mathcal{S}_1\mathcal{S}_1}
\bigg)^{-1}
}_{2}
\Big(
\lambda
\|
 \tilde{z}_{\mathcal{S}_1}
\|_{\infty}
+
\|
 R_{\mathcal{S}_1}
\|_{\infty} 
+
\big\|
\frac{1}{n}l^{(1)}_c(\bm{\beta}^0)_{\mathcal{S}_1}
\big\|_{\infty}
\Big)
\\
&
\leq
\frac{\sqrt{q_n}}{\tau_{-}}
\Big(
\lambda + 
\mathcal{O}_p\bigg\{
\frac{p_n^{1+d} \log{p_n} }{n}
\bigg\}
\big(
1+o_p(1)
\big)
+ \frac{1}{6}\xi^2 \lambda
\Big)
\\
&
\leq
\frac{\sqrt{q_n}}{\tau_{-}} 
(1 + \frac{1}{6}\xi^2 ) \lambda
\big(
1 + o_p(1)
\big)
=  
c^{\ast} \sqrt{q_n} \lambda
\leq
\min_{u \in \mathcal{S}_1}\big\{
\abs{
\beta_{u}^{0}}
\big\}.
\end{align*}
Based on all the results presented above, we have $sign(\hat{\bm{\beta}}) = 
sign(\bm{\beta}^0)$ with probability tending to one.
\end{proof}


\end{document}